\documentclass[english]{tlp}
\usepackage{aopmath}

\usepackage{babel}
\usepackage{amsfonts}
\usepackage{amssymb}
\usepackage{latexsym}
\usepackage{stmaryrd}
\usepackage{vaucanson-g}
\usepackage{graphicx}

\newtheorem{definition}{Definition} % [section]
\newtheorem{example}{Example}       % [section]
\newtheorem{remark}{Remark}         % [section]
\newtheorem{notation}[theorem]{Notation} % added by Alb 22.05.2010

\newcommand{\If}{\leftarrow}

\newcommand{\Iff}{\leftrightarrow}

\newcommand{\vars}{${\it{vars}}$}
\newcommand{\true}{\mathit{true}}

\newcommand{\gel}{\geq_{\mathit{lex}}}
\newcommand{\gl}{>_{\mathit{lex}}}
\newcommand{\bif}{\;{\mathit{if}}\;}
\newcommand{\bth}{\;{\mathit{then}}\;}
\newcommand{\bel}{\;{\mathit{else}}\;}
\newcommand{\tA}{{\widetilde{A}}}

\newcommand{\tH}{{\widetilde{H}}}
\newcommand{\tK}{{\widetilde{K}}}
\newcommand{\tL}{{\widetilde{L}}}
\newcommand{\tM}{{\widetilde{M}}}

\renewcommand{\mathit}{\displaystyle}  %%%%% \mathit for TPLP
\newcommand{\Mathit}[1]{\mbox{\it #1}}  %%%%% better for \_ in TPLP

\newcommand{\DefsPrk}{\Mathit{Defs}$_{k}^{\,\prime}$}
\newcommand{\Defsk}{\Mathit{Defs}$_{k}$}
\newcommand{\Defsn}{\Mathit{Defs}$_{n}$}
\newcommand{\Defsz}{\Mathit{Defs}\/$_{0}$}
\newcommand{\Defs}{\Mathit{Defs}}

\makeatother

\begin{document}

\thispagestyle{empty}

\long\def\comment#1{}

\title[Transformations of Logic Programs on Infinite
Lists] {Transformations of Logic Programs on\\ Infinite Lists}

\author[A. Pettorossi, M. Proietti, and V. Senni]
{ALBERTO PETTOROSSI\\
DISP, University of Rome Tor Vergata,
Via del Politecnico 1, I-00133 Rome, Italy\\
\email{pettorossi@disp.uniroma2.it} \and
MAURIZIO PROIETTI\\
IASI-CNR, Viale Manzoni 30, I-00185 Rome, Italy\\
\email{maurizio.proietti@iasi.cnr.it} \and
VALERIO SENNI\\
DISP, University of Rome Tor Vergata,
Via del Politecnico 1, I-00133 Rome, Italy\\
\email{senni@disp.uniroma2.it} }

\pagerange{\pageref{firstpage}--\pageref{lastpage}}
\volume{\textbf{nn} (n):} \jdate{month 2010} \setcounter{page}{1}
\pubyear{2010}

\maketitle

\label{firstpage}

\maketitle
\begin{abstract}
We consider an extension of logic
programs, called {\em $\omega$-programs}, that can be used to define predicates
over {\em infinite lists}. 
\mbox{$\omega$-programs} allow us to specify properties of the 
infinite behavior of
reactive systems and, in general, properties of infinite sequences of events. 
The  semantics of \mbox{$\omega$-programs} is
an extension of the perfect model semantics. 
We present variants of the familiar unfold/fold rules which can be
used for transforming $\omega$-programs. We show that these new rules are correct, that is,
their application preserves the perfect model semantics. 
Then we outline a general methodology based on program transformation
for verifying properties of $\omega$-programs.
We demonstrate the power of our transformation-based verification methodology
by proving some properties of B\"uchi automata and \mbox{$\omega$-regular} languages.

\smallskip

\end{abstract}

\section{Introduction
\label{sec:introduction}}The problem of specifying and verifying
properties of {\em reactive systems}, such as protocols and
concurrent systems, has received much attention over the past fifty
years or so. The main peculiarity of reactive systems is that they
perform nonterminating computations and, in order to specify and
verify the properties of these computations, various formalisms
dealing with infinite sequences of events have been proposed. Among
these we would like to mention: (i)~B\"uchi automata and other
classes of finite automata on infinite sequences~\cite{Tho90},
(ii)~$\omega$-languages~\cite{Sta97}, and (iii)~various temporal and
modal logics (see~\cite{Cl&99} for a brief overview of these
logics).

Also logic programming has been proposed as a formalism for
specifying computations over infinite structures, such as infinite
lists or infinite trees (see, for
instance,~\cite{Col82,Llo87,Si&06,MiG09}). One advantage of using
logic programming languages is that they are general purpose
languages and, together with a model-theoretic semantics, they also have
an operational semantics. Thus, logic programs over infinite
structures can be used for {\em specifying\/} infinite computations
and, in fact, providing executable specifications for them. However,
very few techniques which use logic programs over infinite
structures, have been proposed in the literature for {\em verifying\/}
properties of infinite computations. We are aware only of a recent
work presented in~\cite{Gu&07}, which is based on {\rm coinductive}
logic programming, that is, a logic programming language whose
semantics is based on greatest models.

In this paper our aim is to develop a 
methodology based on the familiar unfold/fold transformation rules~\cite{BuD77,TaS84} 
for reasoning about infinite structures and verifying properties
of programs over such structures. In order to do so,
we do not introduce a new 
programming language, but we consider a simple 
extension of logic programming on finite terms by introducing
the class of the so-called {\em $\omega$-programs}, which are 
logic programs on infinite lists. 
Similarly to the case of logic programs,
for the class of {\em locally stratified}  $\omega$-programs
we define the {\em perfect model}
semantics (see~\cite{ApB94} for a survey on negation in 
logic programming). 

We extend to
$\omega$-programs the transformation 
rules for locally stratified programs presented
in~\cite{Fi&04a,PeP00a,Ro&02,Sek91,Sek09} and, 
in particular: (i)~we introduce
an {\em instantiation\/} rule which is specific for programs on
infinite lists, (ii)~we weaken the applicability conditions for
the {\em  negative unfolding\/} rule, and 
(iii)~we consider a more powerful {\em negative
folding\/} rule (see Sections~\ref{sec:rules} and \ref{sec:corr_of_rules}
for more details).
We prove that these rules preserve
the perfect model semantics of $\omega$-programs.

Then we extend to $\omega$-programs the transformation-based
methodology for verifying properties of programs presented in~\cite{PeP00a}.
We  demonstrate the power of our verification methodology
through some examples. In particular, we prove: (i)~the non-emptiness of the
language recognized by a B\"uchi automaton,
 and (ii)~the containment between languages
denoted by $\omega$-regular expressions.

The paper is structured as follows. In
Section~\ref{sec:omega-programs} we introduce the class of 
\mbox{$\omega$-pro}{\-grams} and we define the perfect model semantics for locally stratified
\mbox{$\omega$-programs}. In Section~\ref{sec:rules} we present the transformation rules
and in Section~\ref{sec:corr_of_rules} we prove that they preserve the  semantics of $\omega$-programs. In Section~\ref{sec:verification} we
present the transformation-based verification method and we see it in action in
some examples. Finally, in Section~\ref{sec:related} we discuss related work in
the area of program transformation and program verification.

%\vspace{-4mm}

\section{Programs on Infinite Lists
\label{sec:omega-programs}}

Let us consider a first order language $\mathcal{L}_{\omega}$ given
by a set {\it{Var}} of variables, a set {\it{Fun}} of function
symbols, and a set {\it{Pred}} of predicate symbols. We assume that
{\it{Fun}} includes: (i)~a {\em finite}, {\em non-empty} set $\Sigma$ of
constants, (ii)~the constructor $\llbracket
\_|\_\rrbracket $ of the infinite lists of
elements of $\Sigma$, and (iii)~at least one constant not in 
$\Sigma$. Thus, $\llbracket s|t\rrbracket $ is an
infinite list whose head is $s\in \Sigma$ and whose tail is the
infinite list~$t$. Let $\Sigma^{\omega}$ denote the set of the
infinite lists of elements of $\Sigma$.

We assume that $\mathcal{L}_{\omega}$ is a typed
language~\cite{Llo87} with three basic types:  (i)~\texttt{fterm},
which is the type of the finite terms, (ii)~\texttt{elem}, which is
the type of the constants in $\Sigma$, and (iii)~\texttt{ilist},
which is the type of the infinite lists of $\Sigma^{\omega}$.
Every function symbol in {\it{Fun}}$\,- (\Sigma \cup \{\llbracket
\_|\_\rrbracket \})$, with arity $n\,(\geq \!0)$, has type
$(\mathtt{fterm} \!\times\!\cdots\!\times\!\mathtt{fterm})\rightarrow\mathtt{
fterm}$, where ${\mathtt{fterm}}$ occurs $n$ times to the left of
$\rightarrow$. The function symbol $\llbracket \_|\_\rrbracket $ has
type $(\mathtt{elem}\!\times\!\mathtt{ilist})\!\rightarrow\!\mathtt{ilist}$.
A predicate symbol of arity $n\,(\geq \!0)$ in {\it{Pred}} has type
of the form $\tau_1\!\times\!\cdots\!\times\!\tau_n$, where
$\tau_1,\ldots,\tau_n\in\{\mathtt{fterm},\mathtt{elem},\mathtt{ilist}\}$.
For every term  (or formula) $t$,
we denote by $\mathit{vars}(t)$ 
the set of variables occurring in~$t$. 

An {\mbox{{\em $\omega$-clause} $\gamma$}} is a formula of the form
$A\leftarrow L_1\!\wedge \ldots \wedge\! L_m$, with $m\!\geq\!0$,
where $A$ is an atom and $L_1, \ldots, L_m$ are (positive or
negative) literals, constructed as usual from symbols in the typed
language~$\mathcal{L}_{\omega}$, with the following extra
condition:~every predicate in $\gamma$ has, among its
arguments, {\em at most one} argument of type $\mathtt{ilist}$. This
condition makes it easier to prove the correctness of the positive and
negative unfolding rules
(see Section~\ref{sec:rules} for further details).
We denote by {\it true} the empty conjunction of literals.
An {\em $\omega$-program} is a set of
$\omega$-clauses. 

Let \mbox{\it{HU}} be the Herbrand universe constructed from the set
{\it{Fun}}$\,-(\Sigma \cup \{\llbracket \_|\_\rrbracket \})$ of
function symbols. 
An interpretation for our typed language
$\mathcal{L}_{\omega}$, called an
{\mbox{$\omega$-in\-ter}}\-pretation, is a function $I$ such that:
%%%
(i)~$I$ assigns to the types \texttt{fterm}, \texttt{elem},
 and \texttt{ilist}, respectively,
the sets $\mbox{\it {HU}}$, $\Sigma$, and $\Sigma^{\omega}$
(which by our assumptions are non-empty),
%%%
(ii)~$I$ assigns to the function symbol~$\llbracket \_|\_\rrbracket
$, the function $\llbracket \_|\_\rrbracket _I$ such that, for any
element $s\in \Sigma$, for any infinite list $t \in\Sigma^\omega$,
$\llbracket s|t\rrbracket _I$ is the infinite list $\llbracket
s|t\rrbracket $,
%%%
(iii)~$I$ is an Herbrand interpretation for all function symbols in
$\mbox{\it{Fun}}-(\Sigma \cup \{\llbracket \_|\_\rrbracket\})$,
%%%
and~(iv)~$I$ assigns to every $n$-ary predicate $p\in $\,{\it{Pred}}
of type \mbox{$\tau_1\!\times\!\ldots\!\times\!\tau_n$}, a relation
on~$D_1\!\times\!\cdots\!\times\!D_n$, where, for $i=1,\ldots,n$,
$D_i$ is~either $\mbox{\it {HU}}$ or $\Sigma$ or $\Sigma^\omega$, if
$\tau_i$ is~either $\mathtt{fterm}$ or $\mathtt{elem}$ or
$\mathtt{ilist}$, respectively. We say that an
$\omega$-interpretation $I$ is an {\mbox{$\omega$-{\it model}}} of
an \mbox{$\omega$-program $P$} if for every clause $\gamma\!\in\! P$
we have that $I\vDash \forall X_1\ldots\forall X_k\,\gamma$, where
$\vars(\gamma)=\{X_1,\ldots, X_k\}$.

A \emph{valuation} is a function $v\! :$ {\it{Var}} $\rightarrow $
\mbox{\it{HU}} $\, \cup\, \Sigma\, \cup\, \Sigma^{\omega}$ such
that: (i) if $X$ has type $\mathtt{fterm}$ then
$v(X)\!\in\!\mbox{\it{HU}}$, (ii)~if $X$ has type $\mathtt{elem}$
then $v(X)\!\in\!\Sigma$, and (iii)~if $X$ has type $\mathtt{ilist}$
then $v(X)\!\in\!\Sigma^{\omega}$. The valuation function $v$ can be
extended to any term~$t$, or literal $L$, or clause $\gamma$, by
making the function $v$ act on the variables occurring in $t$, or $L$, or
$\gamma$.
We extend the notion of {\em Herbrand base}~\cite{Llo87} to
$\omega$-programs by defining it to be the set $\mathcal{B}_{\omega}=
\{p(v(X_{1}),\ldots ,v(X_{n})) \mid p $ is an $n$-ary predicate 
symbol in {\it Pred} and $v$
is a valuation$\}$. Thus, any $\omega$-interpretation can be identified with a
subset of~\( \mathcal{B}_{\omega} \).

A \emph{local stratification} is a function \( \sigma  \): \(
\mathcal{B}_{\omega}\rightarrow W \), where \( W \) is the set of
countable ordinals. Given $A\in \mathcal{B}_{\omega}$, we define
$\sigma(\neg A) = \sigma(A)\!+\!1$. Given an $\omega$-clause $\gamma$ of the
form \( H\leftarrow L_{1}\wedge \ldots \wedge L_{m} \) 
and a local stratification \( \sigma \), we
say that \( \gamma \) is \emph{locally stratified} w.r.t.~\( \sigma
\) if, for \( i=1,\ldots ,m \), for every valuation $v$, \emph{\(
\sigma (v(H))\geq \sigma (v(L_i)) \)}. An $\omega$-program \( P \)
is \emph{locally stratified w.r.t.~}\( \sigma  \), or \( \sigma \)
is a \emph{local stratification for} \( P \), if every clause in \(
P \) is locally stratified w.r.t.~\( \sigma  \). An $\omega$-program
\( P \) is \emph{locally stratified} if there exists a local
stratification \( \sigma  \) such that \( P \) is \emph{locally
stratified w.r.t.~}\( \sigma  \). 

A \emph{level mapping} is a function $\ell:
\mathit{Pred}\rightarrow\mathbb{N}$. A level mapping is extended to
literals as follows: for any literal $L$ having predicate $p$, if
$L$ is a positive literal, then $\ell(L)=\ell(p)$ and, if $L$ is a
negative literal then $\ell(L)=\ell(p)+1$. An $\omega$-clause
$\gamma$ of the form \( H\leftarrow L_{1}\wedge \ldots \wedge L_{m}
\) is {\em stratified} w.r.t.~$\ell$ if, for \( i=1,\ldots ,m \),
\(\ell (H)\geq \ell(L_i)\). An $\omega$-program $P$ is
\emph{stratified} if there exists a level mapping $\ell$ such that
all clauses of $P$ are stratified w.r.t.~$\ell$~\cite{Llo87}. 
Clearly, every stratified
$\omega$-program is a locally stratified $\omega$-program.
Similarly to the case of logic programs on finite terms, for every locally
stratified $\omega$-program~$P$, we can construct a unique {\em perfect
$\omega$-model} (or {\em perfect model}, for short) denoted by
$M(P)$ (see~\cite{ApB94} for the case of logic programs on finite
terms). Now we present an example of this construction.

\begin{example}\label{ex:omega_prog}
Let: (i)~$\Sigma\!=\!\{a,b\}$ be the set of constants of type
${\texttt{elem}}$, (ii) $S$ be a variable of type
${\texttt{elem}}$, and (iii)~$X$ be a variable of type
${\texttt{ilist}}$. Let $p$~and $q$ be predicates of type
${\texttt{ilist}}$. Let us consider the
following $\omega$-program $P$:

\smallskip

\makebox[45mm][l]{$p(X) \leftarrow \neg q(X)$}
\makebox[35mm][l]{$q(\llbracket b|X\rrbracket ) \leftarrow $}
\makebox[30mm][l]{$q(\llbracket a|X\rrbracket )
\leftarrow q(X)$}

\smallskip

\noindent We have that: (i)~$p(w)$ holds iff $w$ is an infinite list of $a$'s 
and (ii)~$q(w)$ holds iff at least one $b$ occurs in $w$.
Program $P$ is  stratified w.r.t.~the level mapping $\ell$ such that
$\ell(q)\!=\!0$ and {\mbox{$\ell(p)\!=\!1$}}. 
The perfect model~$M(P)$ is constructed by
starting from the ground atoms of level 0 (i.e., those with predicate
$q$). We have that, for all~$w\in\{a,b\}^{\omega}$,
$q(w)\! \in\! M(P)$ iff  $w\!\in\! a^* b (a\!+\!b)^\omega$, that is, 
$q(w)\! \not \in \!M(P)$ iff  $w\!\in\! a^\omega$. Then, we
consider the ground atoms of level~1 (i.e., those with predicate~$p$). 
For all~$w\!\in\!\{a,b\}^{\omega}$, $p(w)\!\in \!M(P)$ iff $q(w)\!\not\in\! M(P)$.
Thus, $p(w)\! \in\! M(P)$ iff $w\!\in\! a^\omega$. 
\end{example}

\section{Transformation Rules\label{sec:rules}}

Given an $\omega$-program $P_0$, a \emph{transformation sequence} is a sequence
$P_{0},\ldots,P_{n}$, with $n\!\geq\!0$, of \mbox{$\omega$-programs} 
constructed as follows. Suppose that we have constructed a sequence
$P_{0},\ldots,P_{k}$, for $0\! \leq \! k\! \leq \! n\! -\! 1 $. Then, the next
program $P_{k+1}$ in the sequence is derived from program $P_{k}$
by applying one of the following transformation rules \mbox{R1--R7}.

First we have the \emph{definition introduction} rule which allows
us to introduce a new predicate definition.

\medskip

\noindent \textbf{R1. Definition Introduction.} Let us consider \( m
\) (\( \geq \! 1 \)) clauses of the form:
\smallskip{}

\( \delta _{1}: \) {\it{newp}}$(X_{1},\ldots ,X_{d})\leftarrow
B_{1}, \) \( \ \ \ldots , \) \ \ \( \delta _{m}: \) \(
${\it{newp}}$(X_{1},\ldots ,X_{d})\leftarrow B_{m}$

\smallskip{}

\noindent where: (i) {\it{newp}} is a predicate symbol not occurring
in \( \{P_{0},\ldots ,P_{k}\} \), (ii) \( X_{1},\ldots ,X_{d} \) are
distinct variables occurring in \( \{ B_{1},\ldots ,B_{m}\} \),
(iii)~none of the $B_i$'s is the empty conjunction of literals,
and (iv)~every predicate symbol occurring in \( \{B_{1},\ldots ,B_{m}\}
\) also occurs in \( P_{0} \). The set $\{\delta _{1}, \ldots,
\delta _{m}\}$ of clauses is said to be the {\em definition} of
{\it{newp}}.

\noindent By \emph{definition introduction} from program \(
P_{k} \) we derive the new program \( P_{k+1}\!=\!P_{k}\cup \{\delta
_{1},\ldots ,\delta _{m}\} \). For \( n\!\geq \!0 \),
\Defsn~denotes the set of clauses introduced by the definition rule
during the transformation sequence \( P_{0},\ldots ,P_{n} \). In
particular, \Defsz$\,=\!\{\}$.

\medskip{}

\noindent In the following {\em instantiation} rule we assume that the set
 of constants of type~\texttt{elem} in the language
$\mathcal{L}_{\omega}$ is the finite set
$\Sigma\!=\!\{s_1,\ldots,s_h\}$.

\smallskip

% INSTANTIATION
\noindent \textbf{R2. Instantiation.}\label{rule:inst} Let~$\gamma$:
$H\leftarrow B$ be a clause in program~$P_k$ and $X$ be a variable
of type~\texttt{ilist} occurring in~$\gamma$. By
\emph{instantiation} of~$X$ in~$\gamma$, we get the clauses:

\smallskip

$\gamma_{1}$: $(H\leftarrow B)\{X/\llbracket s_1|X\rrbracket \}$,
~~\ldots,~~ $\gamma_{h}$: $(H\leftarrow B)\{X/\llbracket
s_h|X\rrbracket \}$

\smallskip

\noindent and we say that clauses $\gamma_1,\ldots,\gamma_h$ are
{\it{derived from}} $\gamma$. From~$P_{k}$ we derive the new
program~$P_{k+1}=(P_{k}-\{\gamma\})\cup\{\gamma_{1},\ldots,\gamma_{h}\}$.

\medskip

The \emph{unfolding} rule consists in replacing an atom $A$
occurring in the body of a clause by its definition in $P_k$. We
present two unfolding rules: (1)~the {\em positive unfolding}, and
(2)~the {\em negative unfolding}. They correspond, respectively, to
the case where $A$ or $\neg A$ occurs in the body of the clause to
be unfolded.

\medskip

% POSITIVE UNFOLDING
\noindent \textbf{R3. Positive Unfolding.} Let \( \gamma :\, \,
H\leftarrow B_{L}\wedge A\wedge B_{R} \) be a clause in program \(
P_{k} \) and let \( P'_{k} \) be a variant of \( P_{k} \) without
variables in common with \( \gamma  \). Let

\smallskip

\( \gamma _{1}:\, \, K_{1}\leftarrow B_{1},
    \ \ \ldots,\ \
    \gamma _{m}:\, \, K_{m}\leftarrow B_{m} \) \ \ \ (\( m\geq 0 \))

\smallskip{}

\noindent be all clauses of program \( P'_{k} \) such that, for \(
i=1,\ldots ,m \), \( A \) is unifiable with \( K_i \), with most
general unifier \( \vartheta _{i}\).

\noindent By \emph{unfolding \( \gamma  \) w.r.t.~\( A \)} we get
the clauses $\eta _{1},\ldots ,\eta _{m}$, where for \( i=1,\ldots
,m \), \( \eta_{i} \) is \( (H\leftarrow B_{L}\wedge B_{i}\wedge
B_{R})\vartheta _{i} \), and we say that clauses $\eta _{1},\ldots
,\eta _{m}$ are {\it{derived from}} $\gamma $. From \( P_{k} \) we
derive the new program \(P_{k+1}=(P_{k}-\{\gamma \})\cup \{\eta
_{1},\ldots ,\eta _{m}\}\).

\medskip{}

In rule R3, and also in the following rule R4, the 
most general unifier can be computed by using a unification
algorithm for finite terms~(see, for instance,~\cite{Llo87}). 
Note that 
this is correct, even in the presence on infinite terms, 
because in any $\omega$-program
each predicate has at most one argument of type~\texttt{ilist}.
On the contrary, if predicates may have
more than one argument of type~\texttt{ilist},
in the unfolding rule it is necessary to use 
a unification algorithm for infinite structures~\cite{Col82}.
For reasons of simplicity, here
we do not make that extension of the unfolding rule
and we stick to our assumption
that every predicate has at most one argument of type~\texttt{ilist}.

\medskip

The {\em existential variables} of a clause $\gamma$ are the
variables occurring in the body of~$\gamma$ and not in its head.

\medskip

% NEGATIVE UNFOLDING
\noindent \textbf{R4. Negative Unfolding.}\label{rule:neg-unfold}
Let $\gamma$: $H\leftarrow B_{L}\wedge \neg\,A\wedge B_{R}$ be a
clause in program $P_{k}$ and let $P'_{k}$ be a variant of $P_{k}$
without variables in common with $\gamma $. Let

\smallskip

\( \gamma _{1}:\, \, K_{1}\leftarrow B_{1},
    \ \ \ldots,\ \
    \gamma _{m}:\, \, K_{m}\leftarrow B_{m} \) \ \ \ (\( m\geq 0 \))

\smallskip{}

\noindent be all clauses of program $P'_{k}$, such that, for \(
i=1,\ldots ,m \), \( A \) is unifiable with \( K_i \), with most
general unifier \( \vartheta _{i}\). Assume that: (1)~\(
A=K_{1}\vartheta _{1}=\cdots = K_{m}\vartheta _{m} \), that is, for
\( i=1,\ldots ,m \), \( A \) is an instance of \( K_{i}\), (2)~for
\( i=1,\ldots ,m \), \( \gamma_{i} \) has no existential variables,
and (3)~from \(\neg (B_{1}\vartheta _{1}\vee \ldots \vee
B_{m}\vartheta_{m})\) we get a logically equivalent disjunction \(
D_{1}\vee \ldots \vee D_{r} \) of conjunctions of literals, with \(
r\geq 0 \), by first pushing \( \neg \) inside and then pushing \(
\vee \) outside.

\noindent By \emph{unfolding $\gamma$  w.r.t.~$\neg A$}
 \emph{using $P_k$} we get the clauses $\eta_{1},\ldots ,
\eta_{r},$ where, for $ i\!=\!1,\ldots,r$, clause $\eta_{i}$ is
$H\leftarrow B_L\wedge D_{i}\wedge B_R$, and we say that clauses
$\eta _{1},\ldots ,\eta _{r}$ are {\it{derived from}} $\gamma $.
From $ P_{k}$ we derive the new program $P_{k+1}=(P_{k}-\{\gamma
\})\cup \{\eta _{1},\ldots ,\eta _{r}\}$.

\medskip

\noindent The following \emph{subsumption} rule allows us to remove
from $P_{k}$ a clause $\gamma$ such that
$M(P_{k})\!=\!M(P_{k}\!-\{\gamma \})$.

\medskip

% SUBSUMPTION
\noindent \textbf{R5. Subsumption.}\label{rule:subsumption} Let
$\gamma_1$: $H\leftarrow$ be a clause in program $P_{k}$ and let
$\gamma_2$ in $P_{k}-\{\gamma_1\}$ be a variant of $(H\leftarrow
B)\vartheta$, for some conjunction of literals $B$ and substitution
$\vartheta$. Then, we say that $\gamma_2$ {\it is subsumed} by
$\gamma_1$ and by  \emph{subsumption}, from $ P_{k}$ we derive the
new program $P_{k+1}=P_{k}-\{\gamma_2\}$.

\medskip

\noindent The \emph{folding} rule consists in replacing instances of
the bodies of the clauses that define an atom $A$ by the
corresponding instance of $A$. Similarly to the case of the
unfolding rule, we have two folding rules: (1)~\emph{positive folding}
and (2)~\emph{negative folding}. They correspond, respectively, to the
case where folding is applied to positive or negative occurrences of
literals.

\smallskip{}

% POSITIVE FOLDING
\noindent \textbf{R6. Positive Folding.} Let $\gamma$ be a clause in
$P_{k}$ and let $\mathit{Defs_k^{\prime}}$ be a variant of
$\mathit{Defs_k}$ without variables in common with $\gamma$. Let the
definition of a predicate in \DefsPrk~consist of the clause $\delta: \,
K\leftarrow B$, where $B$ is a non-empty conjunction of literals.
Suppose that there exists a substitution \( \vartheta  \) such that
clause $\gamma$ is of the form \( H\leftarrow B_{L}\wedge B\vartheta
\wedge B_{R} \) and, for every variable \( X \!\in
\,\vars(B)\,-\,\vars(K) \), the following conditions hold: (i) \(
X\vartheta \) is a variable not occurring in \( \{H,B_{L},B_{R}\}
\), and (ii)~\( X\vartheta  \) does not occur in the term \(
Y\vartheta  \), for any variable \( Y \) occurring in \( B \) and
different from \( X \).

\noindent By \emph{folding $\gamma$ using $\delta$} we get the
clause \(\eta  \): \( H\leftarrow B_{L}\wedge K\vartheta \wedge
B_{R} \), and we say that clause $\eta$  is {\it{derived from}}
$\gamma$. From \( P_{k} \) we derive the new program \(
P_{k+1}=(P_{k}-\{\gamma\})\cup \{\eta \} \).

\medskip{}

% NEGATIVE FOLDING
\noindent \textbf{R7. Negative Folding.} Let $\gamma$ be a clause in
\( P_{k} \) and let \DefsPrk~be a variant of \Defsk$\,$without
variables in common with $\gamma$. Let the definition of a predicate
in \DefsPrk~consist of the $q$ clauses \( \delta_{1}\!: K\leftarrow
L_{1},\ldots,\delta_{q}\!: K\leftarrow L_{q}\), with $q\!\geq\! 1$, such that, for
\(i=1,\ldots,q \), \( L_{i} \) is a literal and  \( \delta_{i} \)
has no existential variables. Suppose that there exists a
substitution \( \vartheta \) such that clause $\gamma$ is of the
form \( H\leftarrow B_{L}\wedge (M_1\wedge\ldots\wedge M_q)\vartheta
\wedge B_{R} \), where, for $i=1,\ldots,q$, if~$L_i$ is the negative
literal~$\neg\,A_i$ then~$M_i$ is~$A_i$, and if~$L_i$ is the
positive literal~$A_i$ then~$M_i$ is~$\neg\,A_i$.

\noindent By \emph{folding $\gamma$ using
\(\delta_{1},\ldots,\delta_{q}\)} we get the clause \(\eta\): \(
H\leftarrow B_{L}\wedge \neg\, K\vartheta \wedge B_{R} \), and we
say that clause $\eta$  is {\it{derived from}} $\gamma$. From \(
P_{k} \) we derive the program \( P_{k+1}=(P_{k}-\{\gamma\})\cup
\{\eta \}\).

\medskip

Note that the negative folding 
rule is not included in the sets of transformation rules
presented in~\cite{Ro&02,Sek91,Sek09}.
The negative folding rule presented in~\cite{Fi&04a,PeP00a} corresponds 
to our rule~R7 in the case where $q\!=\!1$. 

%\vspace{-2mm}

\section{Correctness of the Transformation Rules}
\label{sec:corr_of_rules}

Now let us introduce the notion of correctness of a\,transformation
sequence w.r.t.\,the perfect model semantics.

\begin{definition}
[Correctness of a Transformation Sequence]
\label{def:correctness-of-transf-sequence} Let $P_{0}$ be a locally
stratified $\omega$-program and $P_{0},\ldots ,P_{n}$, with $n\!\geq\!0$,
 be a transformation
sequence. We say that $P_{0},\ldots ,P_{n}$ is {\em correct} if
(i)~$P_{0}\,\cup\,$\Defsn~and $P_n$ are locally stratified
$\omega$-programs and (ii)~$M(P_{0}\,\cup\,$\Defsn$)=M(P_n)$.
\end{definition}

In order to guarantee the correctness of a transformation sequence
$P_{0},\ldots ,P_{n}$ (see Theorem~\ref{th:corr_of_rules} below) we
will require that the application of the transformation rules
satisfy some suitable conditions that refer to a given local
stratification $\sigma$. In order to state those
conditions we need the following definitions.
\vspace{0mm}
 
\begin{definition}
[$\sigma$-Maximal Atom] \label{def:sigma-maximal} Consider a clause
$\gamma$: $H\leftarrow G$. An atom $A$ in $G$ is
said to be {\em $\sigma$-maximal} if, for every
valuation~$v$ and for every literal $L$ in $G$, we have
$\sigma(v(A))\!\geq \!\sigma(v(L))$.
\end{definition}
\vspace{-2mm}

\begin{definition}
[$\sigma$-Tight Clause]\label{def:sigma-tight-definition}
A clause $\delta$: $H\leftarrow G$ is said to be {\mbox{$\sigma$-{\em
tight}}} if there exists a $\sigma$-maximal atom $A$ in $G$
such that, for every valuation~$v$, $\sigma (v(H))\!=\!\sigma(v(A))$.
\end{definition}
\vspace{-2mm}

\begin{definition}
[Descendant Clause] \label{def:descen-clause-definition} A clause
$\eta$ is said to be a {\em descendant} of a clause $\gamma$ if
{\it{either}} $\eta$ is $\gamma$ itself {\it{or}} there exists a
clause~$\delta$ such that $\eta$ is {\it{derived from}} $\delta$ by
using a rule in $\{ \mbox{R2}, \mbox{R3}, \mbox{R4}, \mbox{R6},
\mbox{R7}\}$, and~$\delta$ is a descendant of $\gamma$.
\end{definition}
\vspace{-2mm}

\begin{definition}
[Admissible Transformation Sequence] \label{def:adm-transformation}
Let $P_{0}$ be a locally stratified $\omega$-program and let
$\sigma$ be a local stratification for $P_0$. A transformation
sequence $P_{0},\ldots ,P_{n}$, with $n\!\geq\!0$, is said to be \emph{admissible} if:

\noindent\textup{(1)}~every clause in \Defsn \ is locally stratified w.r.t.~$\sigma$,

\noindent \textup{(2)}~for $k\!=\!0,\ldots,n\! -\! 1$, if $P_{k+1}$ is
derived from $P_{k}$ by positive folding of clause 
\(\gamma\) using clause \(\delta\), then: (2.1)~$\delta$ is 
$\sigma$-tight and {\em either} (2.2.i)~the head
predicate of $\gamma$ occurs in $P_0$, {\em or}~(2.2.ii)~$\gamma$ is a
descendant of a clause $\beta$ in $P_{j}$, with $0\!<\! j\!\leq\!
k$, such that $\beta$ has been derived by positive unfolding of a clause
$\alpha$ in $P_{j-1}$ w.r.t. an atom 
which is $\sigma$-maximal in the body of~$\alpha$
and whose predicate occurs in $P_0$, and

\noindent \textup{(3)}~for $k=0,\ldots ,n\! -\! 1$, if $P_{k+1}$ is
derived from $P_{k}$ by applying the negative folding rule thereby
deriving a clause $\eta$, then $\eta$ is locally stratified
w.r.t.~$\sigma$.
\end{definition}

Note that Condition (1) can always be fulfilled because the
predicate introduced in program $P_{k+1}$ by rule~R1 does
not occur in any of the programs $P_0,\ldots, P_k$.
Conditions (2) and (3) cannot be checked in an algorithmic way for
arbitrary programs and local stratification functions. In particular,
the program property of being locally stratified is undecidable.
However, there are significant classes of programs, such as the 
stratified programs, where these conditions are decidable and easy to
verify.

The following Lemma~\ref{lem:sigma-preservation} and
Theorem~\ref{th:corr_of_rules}, whose proofs can be found in
the Appendix, show that: (i)~when constructing an
admissible transformation sequence $P_0,\ldots,P_n$, the application
of the transformation rules preserves the local
stratification~$\sigma$ for the initial program $P_0$ and, thus, all
programs in the transformation sequence are locally stratified
w.r.t.~$\sigma$, and (ii)~any admissible transformation sequence
preserves the perfect model of the initial program.

\begin{lemma}[Preservation of Local Stratification]
\label{lem:sigma-preservation} Let  $P_{0}$ be a locally stratified
$\omega$-program, $\sigma$ be a local stratification for $P_0$, and
$P_{0},\ldots,P_{n}$\,be\,an\,admissible transformation sequence.\,Then
the\,programs\,$P_{0} \cup$\Defsn, $P_{1},\ldots,P_{n},$ are all
locally stratified w.r.t.~$\sigma$.
\end{lemma}

\begin{theorem}[Correctness of Admissible Transformation
Sequences]\label{th:corr_of_rules} Every admissible transformation
sequence is correct.
\end{theorem}

Now let us make a few comments on Condition (2) of
Definition~\ref{def:adm-transformation} and related conditions
presented in the literature. Transformation sequences of stratified
programs over finite terms constructed by using rules R1, R3, and R6
have been first considered in~\cite{Sek91}. In that paper there is a
sufficient condition, called (F4), for the preservation of the
perfect model. Condition~(F4) is like our Condition~(2) except that (F4) 
does not require the $\sigma$-maximality of the atom w.r.t.~which positive
unfolding is performed. A set of transformation rules which includes
also the negative unfolding rule R4, was proposed in~\cite{PeP00a}
for locally stratified logic programs, and in~\cite{Fi&04a} for
locally stratified constraint logic programs. 
In~\cite{Sek09} Condition (F4) is shown
to be insufficient for the preservation of the perfect model if
rule R4 is used together with rules R1, R3, and R6, as demonstrated
by the following example.

\begin{example}
\label{ex:seki} Let us consider the initial program  $P_0 = \{m
\leftarrow, \ e \leftarrow \neg m, \ e \leftarrow e\}$.
By rule R1 we introduce the clause 
$\delta_1$: $f\leftarrow m \wedge
\neg e$ 
and we derive program $P_1\!=\! P_0 \!\cup\! \{\delta_1\}$ and $\Defs_1\!=\!\{\delta_1\}$. 
By rule R3 we unfold $\delta_1$ w.r.t.~$m$ and we get the clause
$\delta_2$: $f\leftarrow \neg e$.
We derive program $P_2\!=\!P_0 \cup \{\delta_2\}$. Thus, Condition~(F4) is satisfied.
By rule R4 we unfold  $\delta_2$ w.r.t.~$\neg e$
and we get 
$\delta_3$: $f\leftarrow m \wedge \neg e$.
We derive program $P_3=P_0 \cup \{\delta_3\}$.
By rule R6 we fold clause $\delta_3$ using clause $\delta_1$, and we get 
$\delta_4$: $f\leftarrow f$. 
\noindent
We derive program $P_4=P_0\cup\{\delta_4\}$ and $\Defs_4\!=\!\{\delta_1\}$. We have that $f\in
M(P_0\cup\mbox{\Defs}_4)$ and $f\not\in M(P_4)$.  Thus, the 
transformation sequence $P_0,\ldots,P_4$ is not correct.
\end{example}

In order to guarantee the preservation of the perfect model
semantics,~\cite{Sek09} has proposed the following stronger
applicability condition for negative unfolding:

\smallskip \noindent
{\em{Condition}} (NU):
the negative unfolding rule R4 can be applied only if it does not
increase the number of positive occurrences of atoms in
the body of any derived clause. 

\smallskip \noindent
Indeed, in the incorrect transformation
sequence of Example~\ref{ex:seki} the negative unfolding does not
comply with this Condition~(NU).
However, Condition~(NU) is very restrictive, because it forbids  the
unfolding of a clause w.r.t.~a negative literal $\neg A$ when
 the body of a clause defining $A$ contains an occurrence of a
negative literal. Unfortunately, many of the correct
transformation strategies proposed in~\cite{PeP00a,Fi&04a} would be
ruled out if Condition~(NU) is enforced. Our Condition (2) is more
liberal than Condition (NU) and, in particular, it allows us to
unfold w.r.t.~a negative literal $\neg A$ also if the body of a
clause defining~$A$ contains occurrences of negative literals.
The following is an example of a correct, admissible transformation 
sequence which violates Condition (NU).

\begin{example}\label{ex:not_seki}\nopagebreak 
Let us consider the initial program  $P_0 =
\{\mathit{even}(0)\!\leftarrow$, \ \ 
$\mathit{even}(s(s(X)))\!\leftarrow\!\mathit{even}(X),$\ \ 
$\mathit{odd}(s(0))\!\leftarrow$, \ \ 
$\mathit{odd}(s(X))\!\leftarrow\!\neg\,\mathit{odd}(X)\}$
and the transformation sequence we now construct starting from~$P_0$.
By rule R1 we introduce the following clause

$\delta_1$: $p\leftarrow \mathit{even}(X)\wedge \neg\,\mathit{odd}(s(X))$

\noindent and we derive $P_1= P_0 \cup \{\delta_1\}$. 
By taking a local stratification function $\sigma$ such that, 
for all ground terms~$t_1$ and $t_2$, $\sigma(p)\!=\!\sigma(\mathit{even}(t_1))\!>\!
\sigma(\mathit{odd}(t_2))$, we have that $\delta_1$ is \mbox{$\sigma$-tight} and
$even(X)$ is a \mbox{$\sigma$-maximal} atom in its body.
By unfolding $\delta_1$ w.r.t. $\mathit{even}(X)$ we derive 
$P_2=P_0 \cup \{\delta_2,\delta_3\}$,
where

$\delta_2$: $p\leftarrow \neg\,\mathit{odd}(s(0))$\hspace{2cm}

$\delta_3$: $p\leftarrow \mathit{even}(X)\wedge \neg\,\mathit{odd}(s(s(s(X))))$

\noindent
By unfolding, clause $\delta_2$ is removed and we derive
$P_3=P_0 \cup \{\delta_3\}$. By unfolding $\delta_3$
w.r.t.~$\neg \mathit{odd}(s(s(s(X))))$ we derive $P_4=P_0 \cup \{\delta_4\}$, where

$\delta_4$: $p\leftarrow \mathit{even}(X)\wedge \mathit{odd}(s(s(X)))$

\noindent
By unfolding
$\delta_4$ w.r.t.~$\mathit{odd}(s(s(X)))$, we derive $P_5=P_0 \cup \{\delta_5\}$, where

$\delta_5$: $p\leftarrow \mathit{even}(X)\wedge \neg\,\mathit{odd}(s(X))$

\noindent
By applying rule R6, we fold clause $\delta_5$ using
clause $\delta_1$ and derive the final program $P_6=P_0\cup\{\delta_6\}$, where

$\delta_6$: $p\leftarrow p$.

\noindent The transformation
sequence $P_0,\ldots,P_6$ is admissible and, thus, correct. In particular, the
application of rule R6 satisfies Condition~(2) of
Definition~\ref{def:adm-transformation} because $\delta_1$ is $\sigma$-tight
and $\delta_5$
is a descendant of $\delta_3$ which has been derived by unfolding 
w.r.t.~a $\sigma$-maximal atom whose predicate occurs in $P_0$. 

Note that, $P_0,\ldots,P_6$ violates Condition (NU) 
because, by unfolding clause $\delta_3$ 
w.r.t.~$\neg \mathit{odd}(s(s(s(X))))$,
the number of positive occurrences of atoms in
the body of the derived clause $\delta_4$ is larger
than that number in $\delta_3$.
\end{example}

Finally, note that the incorrect transformation sequence of
Example~\ref{ex:seki} is {\it{not}} an admissible transformation
sequence in the sense of our
Definition~\ref{def:adm-transformation}, because it does not comply
with Condition (2). Indeed, consider any local stratification $\sigma$.
The atom $m$ is not $\sigma$-maximal in $m \wedge \neg
e$ because $e$ depends  on $\neg m$ and, hence,
$\sigma(\neg e)\!>\!\sigma(m)$. Thus, the positive folding rule R6 is
applied to the clause $\delta_3$ which is not a descendant of any
clause derived by unfolding w.r.t.~a $\sigma$-maximal atom.

\section{Verifying Properties of $\omega$-Programs by
Program Transformation}
\label{sec:verification}

In this section we will outline a general method, based on the
transformation rules presented in Section~\ref{sec:rules},
for verifying properties of $\omega$-programs.
Then we will see our transformation-based
verification method in action in the proof of: (i)~the
non-emptiness of the language accepted by a B\"{u}chi automaton, 
and (ii)~containment between $\omega$-regular languages.

We assume that we are given an $\omega$-program $P$ 
defining a unary predicate {\it prop} of type {\tt ilist},
which specifies a property of interest,
and we want to check whether or not $M(P)\models\exists X\, {\it prop}(X)$.
Our verification method consists of two steps. 

\smallskip

\noindent {\em Step 1.} 
By using the transformation rules for  $\omega$-programs presented in 
Section~\ref{sec:rules} we derive a {\em monadic} $\omega$-program $T$
(see Definition~\ref{def:monadic_omega_programs} below), such that

$M(P)\models\exists X\, {\it prop}(X)$ iff $M(T)\models\exists X\, {\it prop}(X)$.

\smallskip

\noindent {\em Step 2.} 
We apply to $T$ the decision procedure of~\cite{Pe&09b} for
monadic $\omega$-programs and we check whether or not 
$M(T)\models\exists X\, {\it prop}(X)$.

\smallskip

Our verification method is an extension to $\omega$-programs of
the transformation-based method for proving properties of logic programs on
finite terms presented in~\cite{PeP00a}. 
Furthermore, our method is more powerful than the transformation-based
method for verifying CTL$^*$ properties of finite state reactive
systems presented in~\cite{Pe&09b}. Indeed, at Step 1 of the verification
method proposed here,
(i)~we start from an arbitrary $\omega$-program,
instead of an $\omega$-program which encodes the
branching time temporal logic CTL$^*$, and 
(ii)~we use transformation rules more powerful than those in~\cite{Pe&09b}.
In particular, similarly to~\cite{PeP00a}, the rules
applied at Step 1 allow us to
eliminate the existential variables from program $P$, while the
transformation presented in~\cite{Pe&09b} consists of a specialization
of the initial program w.r.t.~the property to be verified.

Note that there exists no algorithm which always succeeds in transforming an
\mbox{$\omega$-program} into a monadic $\omega$-program.
Indeed, (i) the problem of verifying whether or not, for any $\omega$-program
$P$ and unary predicate {\it prop},
$M(P)\models\exists X\, {\it prop}(X)$ is undecidable, 
because the class of \mbox{$\omega$-programs} includes
the locally stratified logic programs on finite terms,
and (ii) the proof method for monadic
$\omega$-programs presented in~\cite{Pe&09b} is complete.
However, we believe that automatic transformation strategies
can be proposed for significant subclasses of \mbox{$\omega$-programs} along the lines 
of~\cite{PrP95a,PeP00a}.

\begin{definition} [Monadic $\omega$-Programs]
\label{def:monadic_omega_programs}\rm A~{\em monadic}
\mbox{$\omega$-clause} is an $\omega$-clause of the form
$A_0\leftarrow L_1\wedge\ldots\wedge L_m$, with $m\!\geq\! 0$, such
that: (i)~$A_0$ is an atom of the form~$p_0$ or~$q_0(\llbracket s|X_0
\rrbracket)$,
where~$q_0$ is a predicate of type~\texttt{ilist} and
$s\!\in\!\Sigma$, (ii)~for $i\!=\!1,\ldots,m,$ $L_i$ is either an atom
$A_i$ or a negated atom $\neg A_i$, where $A_i$ is of the form $p_i$
or $q_i(X_i)$, and $q_i$ is a predicate of type~\texttt{ilist}, and
(iii)~there exists a level mapping $\ell$ such that, for
$i\!=\!1,\ldots,m,$ if $L_i$ is an atom and
$\mathit{vars}(A_0)\!\not\supseteq\! \mathit{vars}(L_i)$, then
$\ell(A_0)\!>\!\ell(L_i)$ else $\ell(A_0)\!\geq\! \ell(L_i)$. A~{\em monadic
$\omega$-program} is a finite set of monadic $\omega$-clauses.
\end{definition}

\vspace*{-2mm}

\begin{example}[Non-Emptiness of Languages Accepted by B\"{u}chi Automata]
\label{ex:buechi}

In this first application of our verification method,
we will consider {\em B\"{u}chi automata},
which are finite automata acting on infinite words~\cite{Tho90}, 
and we will check whether or not the language 
accepted by a {B\"{u}chi automaton} is empty. 
It is
well known that this verification problem has important applications
in the area of model checking (see, for instance,~\cite{Cl&99}).

A {\em B\"{u}chi automaton} $\mathcal A$ is a nondeterministic finite automaton
 $\langle\Sigma,Q,q_0, \delta,F\rangle$, where, as usual, $\Sigma$ is the input
alphabet, $Q$ is the set of states, $q_0$ is the initial state,
$\delta\subseteq Q\!\times\!\Sigma\!\times\! Q$ is the transition
relation, and $F$ is the set of final states. A~{\em run} of the
automaton~$\mathcal A$ on an infinite input word
$w\!=\!a_0\,a_1\ldots\in\Sigma^\omega$ is an infinite sequence
$\rho\!=\!\rho_0\,\rho_1\ldots\in Q^\omega$ of states such that $\rho_0$
is the initial state $q_0$ and, for all~$n\!\geq\! 0$, $\langle
\rho_n,a_n,\rho_{n+1}\rangle \in \delta$. Let $\mbox{\it{Inf}}\,(\rho)$
denote the set of states that occur infinitely often in the
infinite sequence~$\rho$ of states. An infinite word
$w\in\Sigma^\omega$ is {\em accepted} by $\mathcal A$ if there
exists a run $\rho$ of $\mathcal A$ on $w$ such that
$\mbox{\it{Inf}}\,(\rho)\cap F\neq\emptyset$ or, equivalently, if there is
no state $\rho_m$ in $\rho$ such that every  state $\rho_n$,
with $n\geq m$, is not final. The {\em language accepted} by $\mathcal A$
is the subset of $\Sigma^{\omega}$, denoted $\mathcal L(\mathcal A)$, of the
infinite words accepted by~$\mathcal A$.
In order to check whether or not the language $\mathcal L(\mathcal A)$ is
empty, we construct an $\omega$-program which defines a unary predicate 
{\it{accepting\_run}} such that:

\smallskip
\noindent\makebox[5mm][l]{$(\alpha)$}~$\mathcal L(\mathcal
A)\neq\emptyset$~~iff~~$\exists
X\,\mbox{\it{accepting\_run}}(X)$

\smallskip
\noindent The predicate $\mbox{\it{accepting\_run}}$ is defined by the
following formulas:

\smallskip
\noindent\makebox[5mm][l]{($1$)}~$\mbox{\it{accepting\_run}}(X)
\equiv_{\scriptsize{\Mathit{def}}}\Mathit{run}(X)\wedge \neg\,\mbox{\Mathit{rejecting}}(X)$

\smallskip
\noindent\makebox[5mm][l]{(2)}~$\Mathit{run}(X)\equiv_{\scriptsize{\Mathit{def}}}\exists
S\,(\Mathit{occ}(0,X,S) \wedge \Mathit{initial}(S))\,\wedge$\\
\noindent\makebox[9mm][l]{} $\forall N\,\forall S_1\,\forall S_2\,
(\Mathit{nat}(N) \wedge \Mathit{occ}(N,X,S_1) \wedge \Mathit{occ}(s(N),X,S_2)
\rightarrow \exists A\, \Mathit{tr\/}(S_1,A,S_2)))$

\smallskip
\noindent\makebox[5mm][l]{(3)}~$\mbox{\it{rejecting}}(X)\equiv_{\scriptsize{\Mathit{def}}}\exists
M(\Mathit{nat}(M)\wedge\forall N\forall S
(\Mathit{geq}(N,M)\wedge \Mathit{occ}(N,X,S) \rightarrow
\neg\, \Mathit{final}(S)))$

\smallskip
\noindent where, for all $n\!\geq\! 0$, for all $\rho\!=\!\rho_0\,\rho_1\ldots\in Q^\omega$, for all
$q,q_1, q_2\!\in\! Q$, for all $a\!\in\!\Sigma$, 
(i)~$\Mathit{occ}(s^n(0),\rho,q)$ iff $\rho_n\!= \!q$, 
(ii)~$\Mathit{initial}(q)$ iff $q\!=\!q_0$, 
(iii)~$\Mathit{nat}(s^n(0))$ iff $n\!\geq \!0$, 
(iv)~$\Mathit{tr\/}(q_1,a,q_2)$ iff 
     $\langle q_1,a,q_2\rangle\!\in\!\delta$, 
(v)~$\Mathit{geq}(s^n(0),s^m(0))$ iff $n\!\geq\! m$, and
(vi)~$\Mathit{final}(q)$ iff $q\!\in\! F$.

By $(\alpha)$ and (1)--(3) above, $\mathcal L(\mathcal A)\neq\emptyset$ iff there
exists an infinite sequence $\rho\!=\!\rho_0\,\rho_1\ldots\in Q^\omega$ of
states such that: (i)~$\rho_0$ is the initial state $q_0$, (ii)~for 
all $n\!\geq\!0$, there exists $a\in\Sigma$ such that 
$\langle \rho_n,a,\rho_{n+1}\rangle\in\delta$ (see (2)), and (iii)~there exists no state 
$\rho_m$, with $m\!\geq\!0$, in $\rho$  such that, for all $n\!\geq\! m$, $\rho_n\notin F$ (see (3)).

Now we introduce an 
$\omega$-program $P_{\mathcal A}$ defining the predicates
{\it {accepting\_run}}, {\it run}, {\it{rejecting}},
$\Mathit{nat}$, $\Mathit{occ}$, and~$\Mathit{geq}$.
In particular, clause~1 corresponds to formula (1),
clauses~2--4 correspond to formula (2), and
clauses~5 and~6 correspond to formula~(3). 
(Actually, 
clauses 1--6 can be derived from formulas~(1)--(3)
by applying the Lloyd-Topor transformation~\cite{Llo87}.)
In program
$P\!_{\mathcal A}$ any infinite sequence $\rho_0\rho_1\!\ldots$ of states is
represented by the infinite list
$\llbracket\rho_0,\rho_1,\ldots\rrbracket\!$ of constants.

%%% OMEGA PROGRAM
Given a B\"{u}chi automaton $\mathcal A =
\langle\Sigma,Q,q_0,\delta,F\rangle$, the encoding $\omega$-program
$P\!_{\mathcal{A}}$ consists of the following clauses (independent of
$\mathcal A$):

\smallskip

\makebox[2mm][r]{1.}~$\mbox{\it{accepting\_run}}(X)\leftarrow
\mbox{\it{run}}(X) \wedge \neg\,\mbox{\it{rejecting}}(X)$

\makebox[2mm][r]{2.}~$\mbox{\it{run}}(X)\leftarrow
\mathit{occ}(0,X,S)\wedge\mbox{\it{initial}}(S)
\wedge \neg\,\mbox{\it{not\_a\_run}}(X)$

\makebox[2mm][r]{3.}~$\mbox{\it{not\_a\_run}}(X)\leftarrow
\mbox{\it{nat}}(N)\wedge \mathit{occ}(N,\!X,\!S_1)\wedge
\mathit{occ}(s(N),\!X,\!S_2)\,\wedge$
$\neg\,\mbox{\it{exists\_tr}}(S_1,\!S_2)$

\makebox[2mm][r]{4.}~$\mbox{\it{exists\_tr}}(S_1,S_2) \leftarrow
\mbox{\it{tr}}(S_1,A,S_2)$

\makebox[2mm][r]{5.}~$\mbox{\it{rejecting}}(X)\leftarrow
\mbox{\it{nat}}(M)\wedge \neg\,\mbox{\it{exists\_final}}(M,X)$

\makebox[2mm][r]{6.}~$\mbox{\it{exists\_final}}(M,X)\leftarrow
\mbox{\it{geq}}(N,M)\wedge \mathit{occ}(N,X,S)\wedge
\mbox{\it{final}}(S)$

\makebox[60mm][l]{\makebox[2mm][r]{7.}~$\mbox{\it{nat}}(0)\leftarrow$}

\makebox[2mm][r]{8.}~$\mbox{\it{nat}}(s(N))\leftarrow nat(N)$

\makebox[60mm][l]{\makebox[2mm][r]{9.}~$\mathit{occ}(0,\llbracket S|X\rrbracket,S)\leftarrow$}

\makebox[2mm][r]{10.}~$\mathit{occ}(s(N),\llbracket S|X\rrbracket,R) \leftarrow
\mathit{occ}(N,X,R)$

\makebox[60mm][l]{\makebox[2mm][r]{11.}~$\mbox{\it{geq}}(N,0)\leftarrow$}

\makebox[2mm][r]{12.}~$\mbox{\it{geq}}(s(N),s(M))\leftarrow \mathit{geq}(N,M)$

\smallskip

\noindent together with the clauses (depending on $\mathcal A$)
which define the predicates $\mbox{\it{initial}}$, $\mbox{\it{tr}}$, and
$\mbox{\it{final}}$, where: for all states~$s,s_1,s_2\!\in\!Q$, for all
symbols~$a\!\in\!\Sigma$, (i)~$\mbox{\it{initial}}(s)$ holds iff
$s$ is $q_0$, (ii)~$\mbox{\it{tr}}(s_1,\!a,\!s_2)$ holds iff $\langle
s_1,\!a,\!s_2\rangle\!\in\!\delta$, and (iii)~$\mbox{\it{final}}(s)$
holds iff $s\!\in\! F$.
 
The $\omega$-program $P_\mathcal{A}$ is locally stratified w.r.t.~the
stratification function $\sigma$ defined as follows: for every atom
$A$ in $\mathcal B_\omega$, $\sigma(A)\!=\!0$, except that: for every
element $n$ in \mbox{$\{s^k(0)\mid k\!\geq\! 0\}$,} for every infinite list
$\rho$ in $Q^\omega$,
(i)~$\sigma(\Mathit{rejecting}(\rho))\! =\!
\sigma(\Mathit{not\_a\_run}(\rho))\! =\!\sigma(\mathit{nat}(n))\! =
\!1$, and (ii)~$\sigma(\mathit{run}(\rho))\!=$ $
\sigma(${\Mathit{ac\-cept\-ing\_run}}$(\rho))\!=\!2$.

Now,  let us consider a B\"{u}chi automaton
$\mathcal A$ such that:

$\Sigma\!=\!\{a,b\}$, $Q\!=\!\{1,2\}$, $q_0\!=\!1$,
$\delta\!=\!\{\langle 1,a,1\rangle,\langle 1,b,1\rangle,\langle
1,a,2\rangle,\langle 2,a,2\rangle\}$, $F\!=\!\{2\}$

\noindent which can be represented by the following graph:

\vspace{-1.5mm}
\begin{center}
\includegraphics[bb=80mm 244mm 125mm 263mm,clip,width=35mm]{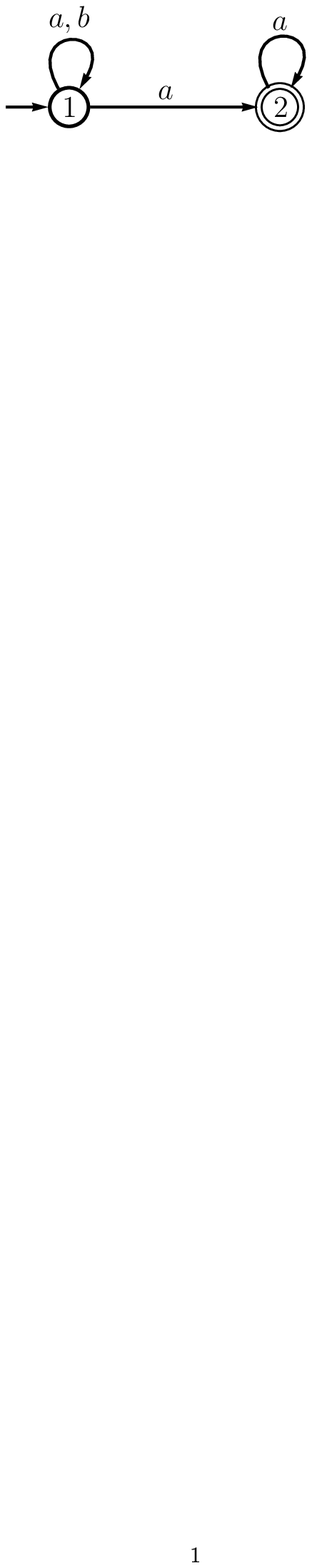}
\end{center}
\vspace{-1mm}

\noindent For this automaton $\mathcal{A}$, program $P_\mathcal{A}$ consists of clauses 1--12 
and the following clauses~\mbox{13--18} that encode
 the initial state (clause~13), the transition relation \linebreak
(clauses~14--17), and the final state (clause~18):

\smallskip

\makebox[30mm][l]{\makebox[3mm][r]{13.}~$\mathit{initial}(1)\leftarrow$}
\makebox[30mm][l]{\makebox[3mm][r]{14.}~$\mathit{tr}(1,a,1)\leftarrow$}
\makebox[30mm][l]{\makebox[3mm][r]{15.}~$\mathit{tr}(1,b,1)\leftarrow$}

\makebox[30mm][l]{\makebox[3mm][r]{16.}~$\mathit{tr}(1,a,2)\leftarrow$}
\makebox[30mm][l]{\makebox[3mm][r]{17.}~$\mathit{tr}(2,a,2)\leftarrow$}
\makebox[3mm][r]{18.}~$\mathit{final}(2)\leftarrow$

\smallskip
\noindent \noindent \noindent In order to check whether or not $\mathcal
L(\mathcal A)\!=\! \emptyset$ we proceed in two steps as indicated at the 
beginning of this Section~\ref{sec:verification}. 
In the first step we use
the rules of Section~\ref{sec:rules} for transforming the
$\omega$-program $P_\mathcal{A}$ into a monadic $\omega$-program~$T$.
This transformation aims at the elimination of the existential variables
from clauses~1--6, with the objective of deriving unary
predicates of type {\texttt{ilist}}. We start from clause~6 and, by
instantiation of the variable $X$ of type {\texttt{ilist}}, we get:

\smallskip

\makebox[3mm][r]{19.}~$\Mathit{exists\_final}(M,\llbracket
1|X\rrbracket) \leftarrow geq(N,M)\wedge {\Mathit{occ}}(N,\llbracket
1|X\rrbracket,S)\wedge final(S)$

\makebox[3mm][r]{20.}~$\Mathit{exists\_final}(M,\llbracket
2|X\rrbracket) \leftarrow \Mathit{geq}(N,M)\wedge
{\Mathit{occ}}(N,\llbracket 2|X\rrbracket,S)\wedge final(S)$

\smallskip

\noindent By some unfolding and subsumption steps, from clauses 19 and 20 we
get:

\smallskip

\makebox[3mm][r]{21.}~$\Mathit{exists\_final}(0,\llbracket
1|X\rrbracket) \leftarrow {\Mathit{occ}}(N,X,S)\wedge final(S)$

\makebox[3mm][r]{22.}~$\Mathit{exists\_final}(s(M),\llbracket
1|X\rrbracket) \leftarrow \Mathit{geq}(N,M)\wedge
{\Mathit{occ}}(N,X,S)\wedge final(S)$

\makebox[3mm][r]{23.}~$\Mathit{exists\_final}(0,\llbracket
2|X\rrbracket)\leftarrow$

\makebox[3mm][r]{24.}~$\Mathit{exists\_final}(s(M),\llbracket
2|X\rrbracket) \leftarrow \Mathit{geq}(N,M)\wedge
{\Mathit{occ}}(N,X,S)\wedge final(S)$

\smallskip

\noindent Note that clauses 21--24 are descendants of clauses derived by
unfolding clauses~19 and 20 w.r.t. the $\sigma$-maximal atom
$\Mathit{geq}(N,M)$. By rule R1, we introduce:

\smallskip

\makebox[3mm][r]{25.}~$\Mathit{new}_1(X) \leftarrow
{\Mathit{occ}}(N,X,S)\wedge final(S)$

\smallskip

\noindent This clause is $\sigma$-tight by taking, for
every infinite list $\rho$ of states, $\sigma(\Mathit{new}_1(\rho))\!=\!0$. By
folding clause 21 using clause 25, and folding clauses 22 and 24 using clause~6
(indeed, without loss of generality, we may assume that clauses 1--6 have been
introduced by rule R1), we get:

\smallskip

\makebox[3mm][r]{26.}~$\Mathit{exists\_final}(0,\llbracket
1|X\rrbracket) \leftarrow \Mathit{new}_1(X)$

\makebox[3mm][r]{27.}~$\Mathit{exists\_final}(s(M),\llbracket
1|X\rrbracket) \leftarrow \Mathit{exists\_final}(M,X)$

\makebox[3mm][r]{28.}~$\Mathit{exists\_final}(s(M),\llbracket
2|X\rrbracket) \leftarrow \Mathit{exists\_final}(M,X)$

\smallskip

\noindent By instantiation of the variable $X$ and by some unfolding and
subsumption steps, from clause 25 we get:

\smallskip

\makebox[3mm][r]{29.}~$\Mathit{new}_1(\llbracket 1|X\rrbracket)
\leftarrow {\Mathit{occ}}(N,X,S)\wedge final(S)$\hspace{2cm}

\makebox[3mm][r]{30.}~$\Mathit{new}_1(\llbracket 2|X\rrbracket)
\leftarrow$

\smallskip

\noindent Note that clause 29 is a descendant of clause~25, that has
been unfolded w.r.t. the $\sigma$-maximal atom ${\Mathit
occ}(N,X,S)$. By folding clause 29 using clause 25 we get:

\smallskip

\makebox[3mm][r]{31.}~$\Mathit{new}_1(\llbracket 1|X\rrbracket)
\leftarrow \Mathit{new}_1(X)$

\smallskip

\noindent At this point we have obtained the definitions of the
predicates $\Mathit{exists\_final}$ and $\Mathit{new}_1$ (that is,
clauses~23, 26--28, 30, and 31) that do not have existential variables.

Now the transformation of program $P_{\mathcal A}$ proceeds by
performing on clauses~1--5 a sequence of transformation steps,
which is similar to the one we have performed above on clause~6 for
eliminating its existential variables. By doing so, we get:

\smallskip
\makebox[3mm][r]{32.}~$\Mathit{accepting\_run}
(\llbracket1|X\rrbracket)\leftarrow
\neg\,\Mathit{not\_a\_run}(X)\wedge \Mathit{new}_1(X)\wedge
\neg\,\Mathit{rejecting}(X)$

\makebox[3mm][r]{33.}~$\Mathit{run}(\llbracket1|X\rrbracket)
\leftarrow \neg\, \Mathit{not\_a\_run}(X)$

\makebox[3mm][r]{34.}~$\Mathit{not\_a\_run}(\llbracket1|X\rrbracket)
\leftarrow \Mathit{not\_a\_run}(X)$

\makebox[65mm][l]{\makebox[3mm][r]{35.}~$\Mathit{not\_a\_run}
(\llbracket2|X\rrbracket) \leftarrow \Mathit{new}_2(X)$}

\makebox[65mm][l]{\makebox[3mm][r]{36.}~$\Mathit{not\_a\_run}
(\llbracket2|X\rrbracket)\leftarrow \Mathit{not\_a\_run}(X)$}

\makebox[65mm][l]{\makebox[3mm][r]{37.}~$\Mathit{new}
_2(\llbracket1|X\rrbracket)\leftarrow$}

\makebox[3mm][r]{38.}~$\Mathit{rejecting}(\llbracket1|X\rrbracket)
\leftarrow \neg\,new_1(X)$

\makebox[3mm][r]{39.}~$\Mathit{rejecting}(\llbracket1|X\rrbracket)
\leftarrow \Mathit{rejecting}(X)$

\makebox[3mm][r]{40.}~$\Mathit{rejecting}(\llbracket2|X\rrbracket)
\leftarrow \Mathit{rejecting}(X)$

\smallskip
\noindent The final $\omega$-program $T$ obtained from program
$P_{\mathcal A}$, consists of clauses 30--40 and it
is a monadic $\omega$-program.

Now, in the second step of our verification method, we check whether
or not $\exists X\,\Mathit{accepting\_run}(X)$ holds in $M(T)$ by
applying the proof method of~\cite{Pe&09b}. We construct
the tree depicted in Figure~\ref{fig:proof-ex1}, where the literals
occurring in the two lowest levels are the same (see the two rectangles) 
and, thus, we have detected an infinite loop.
According to the conditions given in Definition~6 of~{\protect\cite{Pe&09b}}, this tree
is a proof of $\exists X\,\Mathit{accepting\_run}(X)$. The run
$\rho\!=\!12^\omega$ is a witness for $X$ and corresponds to the
accepted word~$a^\omega$. Thus, $\mathcal
L(\mathcal A)\neq\emptyset$.

\begin{figure}
 \setlength{\unitlength}{1mm}%{4144sp}%
\begin{picture}(120,32)(0,-32)%(2232,2637)(1876,-3358)
\thicklines

%%% PROOF FOR \exists X\, accepting_run(X)
% ROOT
\put(37,-4){\makebox[16mm][l]{$\exists X\,
\Mathit{accepting\_run}(X)$}}
\put(50,-5){\line(6,-1){38}}\put(24,-8){\footnotesize $1$}
\put(50,-5){\line(-6,-1){38}}\put(58.5,-10){\footnotesize $1$}
\put(50,-5){\line(3,-2){10}}\put(74,-8){\footnotesize $1$}

% DEPTH 1
\put(0,-14){\makebox[16mm][l]{$\neg\,\Mathit{not\_a\_run}(X)$}}
\put(53,-14){\makebox[16mm][l]{$\Mathit{new}_1(X)$}}
\put(78,-14){\makebox[16mm][l]{$\neg\, \Mathit{rejecting}(X)$}}

\put(12,-15.3){\line(0,-1){6.5}}\put(9.5,-19){\footnotesize $2$}
\put(12,-15.3){\line(3,-1){18}}\put(23,-18){\footnotesize $2$}
\put(60,-15){\line(0,-1){6.5}}\put(61,-19){\footnotesize $2$}
\put(90,-15){\line(0,-1){6.5}}\put(91,-19){\footnotesize $2$}

% DEPTH 2
\put(5,-24){\makebox[16mm][l]{$\neg\,\Mathit{new}_2(X)$}}
\put(23,-24){\makebox[16mm][l]{$\neg\, \Mathit{not\_a\_run}(X)$}}
\put(57,-24){\makebox[16mm][l]{$\Mathit{true}$}}
\put(78,-24){\makebox[16mm][l]{$\neg\, \Mathit{rejecting}(X)$}}

\put(12,-25){\line(0,-1){6.5}}\put(9.5,-29){\footnotesize $2$}
\put(36,-25.3){\line( 0,-1){6.2}}\put(33.5,-29){\footnotesize $2$}
\put(36,-25.3){\line(3,-1){18}}\put(47,-28){\footnotesize $2$}
\put(90,-25){\line(0,-1){6.5}}\put(91,-29){\footnotesize $2$}

% DEPTH 3
\put(9,-34){\makebox[16mm][l]{$\Mathit{true}$}}
\put(24,-34){\makebox[16mm][l]{$\neg\,\Mathit{new}_2(X)$}}
\put(44,-34){\makebox[16mm][l]{$\neg\, \Mathit{not\_a\_run}(X)$}}
\put(78,-34){\makebox[16mm][l]{$\neg\, \Mathit{rejecting}(X)$}}

% LOOP
\thinlines
% first box
\put(3,-35.5){\line(1,0){100}}
\put(3,-30.5){\line(1,0){100}}
\put(3,-30.5){\line(0,-1){5}}
\put(103,-30.5){\line(0,-1){5}}
% second box
\put(3,-25.5){\line(1,0){100}}
\put(3,-20.5){\line(1,0){100}}
\put(3,-20.5){\line(0,-1){5}}
\put(103,-20.5){\line(0,-1){5}}
% arc
\put(104.4,-28.5){\oval(5,10)[r]}
\put(104.4,-33.55){\line(-1,0){0.8}}
% arrow
\put(104.4,-23.53){\vector(-1,0){1}}

\put(115,-3.5){\circle*{1.5}}%\rule{2mm}{2mm}}% \circle{2}
\put(115,-4.5){\vector(0,-1){8}}
\put(116,-9){\footnotesize $1$}
\put(115,-13.5){\circle*{1.5}}%\rule{2mm}{2mm}}% \circle{2}
\put(115,-14.5){\vector(0,-1){8}}
\put(116,-19){\footnotesize $2$}
\put(115,-23.5){\circle*{1.5}}%\rule{2mm}{2mm}}% \circle{2}
\put(115,-24.5){\vector(0,-1){8}}
\put(116,-29){\footnotesize $2$}
\put(115,-33.5){\circle*{1.5}}%\rule{2mm}{2mm}}% \circle{2}
% arc
\put(117,-28.5){\oval(5,10)[r]}
\put(117,-33.55){\line(-1,0){0.8}}
% arrow
\put(117,-23.53){\vector(-1,0){1}}
%\put(121,-29){\footnotesize $2$}
\end{picture}
\vspace{2mm}
\caption{Proof of $\exists X\, \Mathit{accepting\_run}(X)$
w.r.t.~the monadic $\omega$-program $T$. On the right we have shown the infinite loop and
the associated accepting run $122^\omega$ (that is, $12^\omega$).
\label{fig:proof-ex1}}
\vspace{-2mm}
\end{figure}
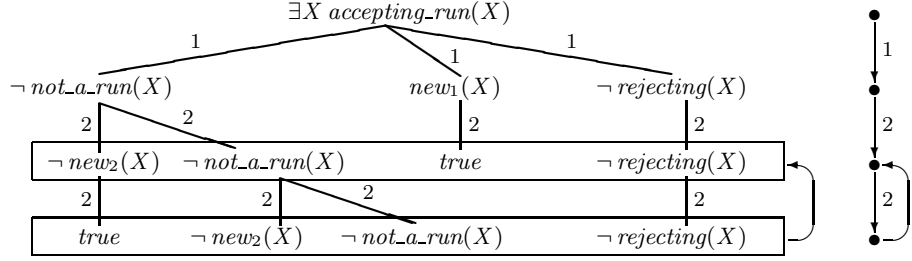
\end{example}

\begin{example}[Containment Between $\omega$-Regular Languages]

In this second application of our verification method, we will consider
regular sets of infinite words over a finite alphabet $\Sigma$~\cite{Tho90}. 
These sets are
denoted by $\omega$-regular expressions whose syntax is defined as follows: 

\smallskip
\makebox[45mm][l]{$e\,::=\,a\ |\ e_1e_2\ |\ e_1\!+\!e_2\ |\
e^*$}\makebox[25mm][l]{with $a\in \Sigma$}(regular expressions)

\smallskip
\makebox[70mm][l]{$f\,::=\,e^\omega\ |\ e_1e_2^\omega\ |\
f_1\!+\!f_2$}($\omega$-regular expressions)

\smallskip
\noindent Given a regular (or an $\omega$-regular) expression $r$,
by $\mathcal L(r)$ we indicate the set of all words in $\Sigma^*$ 
(or $\Sigma^\omega$, respectively) denoted by~$r$.
In particular, given a regular expression $e$, we have that
$\mathcal L(e^\omega)=\{w_0w_1\ldots\in\Sigma^\omega\ |\ $for~$i\!\geq\! 0,
w_i\in\mathcal L(e)\!\subseteq\! \Sigma^*\}$.

Now we introduce an $\omega$-program, called
$P_f$, which defines the predicate $\omega$-${\Mathit acc}$ such
that for any $\omega$-regular expression $f$, for any infinite word~$w$,
\mbox{$\omega$-${\Mathit acc}(f,w)$} holds iff $w\in \mathcal L(f)$.
Any infinite word $a_0a_1\ldots\in\Sigma^\omega$ is represented by the infinite
list  $\llbracket a_0,a_1,\ldots\rrbracket$ of symbols in $\Sigma$. 
The $\omega$-program $P_f$ is made out of the following clauses:

\smallskip

% REGULAR EXPRESSIONS ===============================================

% ALPHABET - type (fin x fin)
\makebox[3mm][r]{1.}~${\Mathit acc}(E,[E])\leftarrow
\Mathit{symb}(E)$

% CONCATENATION - type (fin x fin)
\makebox[3mm][r]{2.}~${\Mathit{acc}}(E_1E_2,X) \leftarrow
\Mathit{app}(X_1,X_2,X)\wedge {\Mathit{acc}}(E_1,X_1)\wedge {\Mathit
acc}(E_2,X_2)$

% UNION - type (fin x fin)
\makebox[3mm][r]{3.}~${\Mathit{acc}}(E_1\!+\!E_2,X) \leftarrow
{\Mathit acc}(E_1,X)$

\makebox[3mm][r]{4.}~${\Mathit{acc}}(E_1\!+\!E_2,X) \leftarrow
{\Mathit acc}(E_2,X)$

% STAR - type (fin x fin)
\makebox[3mm][r]{5.}~${\Mathit{acc}}(E^*,[\,])\leftarrow$

\makebox[3mm][r]{6.}~${\Mathit{acc}}(E^*,X) \leftarrow
{\Mathit{app}}(X_1,X_2,X)\wedge {\Mathit{acc}}(E,X_1)\wedge
{\Mathit{acc}}(E^*,X_2)$

% OMEGA-REGULAR EXPRESSIONS =========================================

% UNION - type (fin x inf)
\makebox[65mm][l]{\makebox[3mm][r]{7.}~$\omega$-$\Mathit{acc}(F_1\!+\!F_2,X)
\leftarrow \omega$-$\Mathit{acc}(F_1,X)$}

\makebox[3mm][r]{8.}~$\omega$-$\Mathit{acc}(F_1\!+\!F_2,X) \leftarrow
\omega$-$\Mathit{acc}(F_2,X)$

% OMEGA - type (fin x inf)
\makebox[3mm][r]{9.}~$\omega$-$\Mathit{acc}(E^\omega,X) \leftarrow
\neg\, \Mathit{new}_1(E,X)$

%% CONCATENATION - type (fin x inf)
%\makebox[3mm][r]{10.}~$\omega$-$\Mathit{acc}(E_1E_2^\omega,X)
%\leftarrow \omega$-$\Mathit{app}(X_1,X_2,X)\wedge {\Mathit
%acc}(E_1,X_1)\wedge \omega$-$\Mathit{acc}(E_2^\omega,X_2)$

% CONCATENATION - type (fin x inf)
\makebox[3mm][r]{10.}~$\omega$-$\Mathit{acc}(E_1E_2^\omega,X)
\leftarrow \Mathit{prefix}(X,N,X_1)\wedge {\Mathit
acc}(E_1,X_1)\wedge \omega$-$\Mathit{acc}1(E_2^\omega,X_1,X)$

% type (fin x inf)
\makebox[3mm][r]{11.}~$\Mathit{new}_1(E,X) \leftarrow
\Mathit{nat}(M)\wedge \neg\, \Mathit{new}_2(E,M,X)$

% type (fin x fin x inf)
% reg/3 predicate encoding the "exists accepted prefix" condition
\makebox[3mm][r]{12.}~$\Mathit{new}_2(E,M,X) \leftarrow
\Mathit{geq}(N,M)\wedge \Mathit{prefix\/}(X,N,V)\wedge
{\Mathit acc}(E^*,V)$

% acc1
\makebox[3mm][r]{13.}~$\omega$-$\Mathit{acc}1(E,[\,],X)
\leftarrow \omega$-$\Mathit{acc}(E,X)$

\makebox[3mm][r]{14.}~$\omega$-$\Mathit{acc}1(E,[H|T],\llbracket H|X\rrbracket)
\leftarrow \omega$-$\Mathit{acc}1(E,T,X)$

% type (fin x fin)
\makebox[47mm][l]{\makebox[3mm][r]{15.}~$\Mathit{geq\/}(N,0)\leftarrow$}

\makebox[3mm][r]{16.}\makebox[70mm][l]{~$\Mathit{geq\/}(s(N),s(M))
\leftarrow \Mathit{geq\/}(N,M)$}

\makebox[47mm][l]{\makebox[3mm][r]{17.}~$\Mathit{nat\/}(0)\leftarrow$}

\makebox[3mm][r]{18.}~$\Mathit{nat\/}(s(N)) \leftarrow
\Mathit{nat\/}(N)$

% type (fin x fin x inf)
\makebox[47mm][l]{\makebox[3mm][r]{19.}~$\Mathit{prefix\/}(X,0,[\,])\leftarrow$}

\makebox[3mm][r]{20.}~$\Mathit{prefix\/}(\llbracket S|X\rrbracket,s(N),[S|Y])
\leftarrow \Mathit{prefix\/}(X,N,Y)$

%% type (fin x inf x inf)
%\makebox[47mm][l]{\makebox[3mm][r]{19.}~$\omega$-$\Mathit{app}([\,],Y,
%Y)\leftarrow$ }

%\makebox[3mm][r]{20.}~$\omega$-$\Mathit{app}([S|X],Y,\llbracket
%S|Z\rrbracket) \leftarrow \omega$-$\Mathit{app}(X,Y,Z)$

% type (fin x fin x fin)
\makebox[47mm][l]{\makebox[3mm][r]{21.}~$\Mathit{app}([\,],Y,Y)\leftarrow$}

\makebox[3mm][r]{22.}~$\Mathit{app}([S|X],Y,[S|Z]) \leftarrow
\Mathit{app}(X,Y,Z)$

\smallskip

\noindent together with the clauses defining the predicate
${\Mathit{symb}}$, where $\Mathit{symb}(a)$ holds iff $a\in \Sigma$.
We have that $\Mathit{prefix\/}(X,N,Y)$ holds iff $Y$ is the list of the $N\,(\geq\!0)$
leftmost symbols of the infinite list $X$.
%Note that: (i)~$\Mathit{prefix\/}(X,N,Y)$ holds iff $Y$ is the list of the $N\,(\geq\!0)$
%leftmost symbols of the infinite list $X$, and (ii)~$\omega$-$\Mathit{app}(X,Y,Z)$ holds
%iff the concatenation of the finite list $X$ and the infinite list $Y$ is the
%infinite list $Z$.
Clauses 1--6 stipulate that, for any finite word $w$ and regular expression~$e$,
${\Mathit acc}(e,w)$ holds iff $w\in\mathcal L(e)$. Analogously,
clauses 7--14 stipulate that, for any infinite word~$w$ and
$\omega$-regular expression~$f$, \mbox{$\omega$-${\Mathit acc}(f,w)$} holds iff
$w\in\mathcal L(f)$. In particular, clauses 9, 11, and 12 correspond to the
following definition:

\smallskip
$\omega$-$\Mathit{acc}(E^\omega,X) \equiv_{\scriptsize{\Mathit{def}}}$

\hspace*{15mm}$\forall M (\Mathit{nat}(M)\rightarrow \exists N \exists V 
(\Mathit{geq}(N,M) \wedge {\Mathit{prefix\/}}(X,N,V)\wedge
{\Mathit{acc}}(E^*,V)))$

\smallskip
\noindent The $\omega$-program $P_f$ is stratified and, thus, locally stratified.

Now, let us consider the $\omega$-regular expressions  $f_1 \equiv_{\Mathit{def}} a^\omega$
and $f_2 \equiv_{\Mathit{def}} (b^*a)^\omega$. The following two clauses:

\smallskip
\makebox[60mm][l]{\makebox[3mm][r]{23.}~$\Mathit{expr}_1(X)\leftarrow
\omega$-$\Mathit{acc}(a^\omega,X)$} \makebox[3mm][r]{24.}~$\Mathit{expr}_2(X)
\leftarrow \omega$-$\Mathit{acc}((b^*a)^\omega,X)$

\smallskip
\noindent together with program $P_f$, define the predicates
$\Mathit{expr}_1$ and $\Mathit{expr}_2$ such that, for every
infinite word $w$, $\Mathit{expr}_1(w)$ holds iff $w\in\mathcal
L(f_1)$ and $\Mathit{expr}_2(w)$ holds iff $w\in\mathcal L(f_2)$.
If we introduce the following clause:

\smallskip
\makebox[3mm][r]{25.}~$\Mathit{not\_contained}(X)\leftarrow
\Mathit{expr}_1(X)\wedge\neg\, \Mathit{expr}_2(X)$

\smallskip
\noindent 
we have that $\mathcal L(f_1)\subseteq \mathcal
L(f_2)$ iff $M(P_f \cup \{23,24,25\}) \not\models \exists X \Mathit{not\_contained}(X)$.
By performing a sequence of transformation steps which is similar to the one
we have performed in Example~\ref{ex:buechi}, from program
$P_f\cup\{23,24,25\}$ we get the following monadic $\omega$-program $T$:

\smallskip
\makebox[78mm][l]{\makebox[1mm][r]{26.}\,$\Mathit{not\_contained}(\llbracket
a|X\rrbracket)\leftarrow \neg \Mathit{new}_3(X)\wedge
\Mathit{new}_4(X)$}
\makebox[1mm][r]{31.}\,$\Mathit{new}_5(\llbracket a|X\rrbracket
)\leftarrow\Mathit{new}_4(X)$

\makebox[78mm][l]{\makebox[1mm][r]{27.}\,$\Mathit{new}_3(\llbracket
a|X\rrbracket)\leftarrow\Mathit{new}_3(X)$}
\makebox[1mm][r]{32.}\,$\Mathit{new}_5
(\llbracket b|X\rrbracket)\leftarrow \Mathit{new}_5(X)$

\makebox[78mm][l]{\makebox[1mm][r]{28.}\,$\Mathit{new}_3(\llbracket
b|X\rrbracket)\leftarrow$}
\makebox[1mm][r]{33.}\,$\Mathit{new}_5(\llbracket
b|X\rrbracket)\leftarrow \neg\,\Mathit{new}_6(X)$

\makebox[78mm][l]{\makebox[1mm][r]{29.}\,$\Mathit{new}_4
(\llbracket a|X\rrbracket)\leftarrow
\Mathit{new}_4(X)$}
\makebox[1mm][r]{34.}\,$\Mathit{new}_6
(\llbracket a|X\rrbracket)\leftarrow$

\makebox[78mm][l]{\makebox[1mm][r]{30.}\,$\Mathit{new}_4
(\llbracket b|X\rrbracket)\leftarrow
\Mathit{new}_5(X)$}
\makebox[1mm][r]{35.}\,$\Mathit{new}_6
(\llbracket b|X\rrbracket)\leftarrow \Mathit{new}_6(X)$

\smallskip
\noindent By using the proof
method for monadic $\omega$-programs
of~\cite{Pe&09b} we have that $M(T)\nvDash \exists X\,
\Mathit{not\_contained}(X)$ and, thus, $\mathcal L(f_1)\subseteq
\mathcal L(f_2)$.
\end{example}
%\vspace{-3mm}

\section{Related Work and Conclusions}
\label{sec:related}

There have been various proposals for extending logic programming languages to
infinite structures (see, for instance,~\cite{Col82,Llo87,MiG09,Si&06}). 
In order to provide the semantics of infinite
structures, these
languages introduce new concepts, 
such as {\em complete Herbrand interpretations}, {\em rational
trees}, and {\em greatest models}. Moreover, the operational semantics of these
languages requires an extension of SLDNF-resolution by means of equational reasoning 
 and new inference rules, such as the so-called {\em coinductive hypothesis} rule.

On the contrary, the semantics of $\omega$-programs we consider in this paper is very
close to the usual perfect model semantics for logic programs
on finite terms, and we do not define any new operational semantics.
Indeed, the main objective of this paper is {\em not} to provide a
new model for computing over infinite structures, but to present a
methodology, based on unfold/fold transformation rules, for reasoning about
such structures and proving their properties.

Very little work has been done for applying transformation techniques to
logic languages that specify the (possible infinite) 
computations of reactive systems.  Notable
exceptions are~\cite{UeF88} and~\cite{Et&01}, where the unfold/fold
transformation rules have been studied in the context of {\em
guarded Horn clauses} (GHC) and {\em concurrent constraint programs}
(CCP). However, GHC and CCP programs are {\em definite} programs and
do not manipulate terms denoting infinite lists. 

The transformation rules presented in this paper extend to 
$\omega$-programs the
rules for general programs proposed 
in~\cite{Fi&04a,PeP00a,Ro&02,Sek91,Sek09}. In 
Sections~\ref{sec:rules} and~\ref{sec:corr_of_rules}
we discuss in detail 
the relationship of the rules in those papers with our rules here.

In Section~\ref{sec:verification} we have used our transformation 
rules for extending to infinite lists a verification
methodology proposed in~\cite{PeP00a} 
and, as an example, we have shown how to verify
properties of the infinite behaviour of B\"uchi automata and 
properties of $\omega$-regular languages.
This extends our previous work (see~\cite{Pe&09b}), as already
illustrated at the beginning of Section~\ref{sec:verification}.

The verification methodology based on transformations we have 
proposed here is very general. It can be applied to the proof 
of properties of infinite state reactive systems; thus it goes 
beyond the capabilities of finite state model checkers.
The focus of our paper has been the proposal of correct 
transformation rules, that is, rules which preserve the
perfect model, while the automation of the verification 
methodology itself is 
left for future work.
This automation 
requires the design of suitable transformation strategies that can be 
defined by adapting to $\omega$-programs some strategies 
already developed in the case of logic programs on finite 
terms (see, for instance,~\cite{PrP95a,PeP00a}).

Many other papers use logic programming, possibly with constraints,
for specifying and verifying properties of finite or infinite state 
reactive systems~(see, 
for instance,~\cite{AbM89,DeP01,FrO97a,Ja&04,LeM99,NiL00,Ra&97}), but 
they do not consider terms which explicitly represent infinite
structures. As we have seen in the examples of Section~\ref{sec:verification}, 
infinite lists are very convenient for specifying those properties
and the use of infinite lists avoids
 ingenious encodings which would have been otherwise required.

%\vspace{-4mm}

%\section{Acknowledgements}
%We thank Hirohisa Seki for stimulating conversations on
%the topics of this paper. We also thank John Gallagher for
%his comments and 
%the anonymous referees for their constructive criticism.

%\vspace{-4mm}

%\bibliographystyle{acmtrans}
%\bibliography{Transformation}
\newpage

\appendix

\section{Proofs for Section~\ref{sec:corr_of_rules}}
\label{appendix:corr_of_rules}

We start off by showing that admissible transformation sequences
preserve the local stratification~\( \sigma \) for the initial
program \( P_{0} \) as stated in the following lemma.

\medskip

\noindent {\bf Lemma~\ref{lem:sigma-preservation} (Preservation of
Local Stratification)}

\noindent Suppose that $P_{0}$ is a locally stratified
$\omega$-program, $\sigma$ is a local stratification for $P_0$, and
$P_{0},P_{1},\ldots ,$ $P_{n}$ is an admissible transformation sequence. Then
the programs $P_{0}\,\cup\,${\it{Defs}}$_{n}$, $P_{1},\ldots ,P_{n}$
are locally stratified
w.r.t.~$\sigma$.

\begin{proof}
Since $P_{0},\ldots ,P_{n}$ is an admissible transformation
sequence, every definition in $\mathit{Defs}_{n}$ 
%%is
%%$\sigma$-tight (see Point~(1) of
%%Definition~\ref{def:adm-transformation}). Thus, $\mathit{Defs}_{n}$
is locally stratified w.r.t.~$\sigma$ (see Point~(1) of
Definition~\ref{def:adm-transformation}). Since, by hypothesis, $P_0$
is locally stratified w.r.t.~$\sigma$, also $P_{0}\cup
\mathit{Defs}_{n}$ is locally stratified w.r.t.~$\sigma$.

\noindent Now we will prove that, for \( k=0,\ldots ,n \), \( P_{k}
\) is locally stratified w.r.t.~\( \sigma  \) by induction on \( k
\).

\medskip

\noindent \emph{Basis} (\( k=0 \)). By hypothesis \( P_{0} \) is
locally stratified w.r.t.~\( \sigma  \).

\medskip

\noindent \emph{Step}. We assume that \( P_{k} \) is
locally stratified w.r.t.~\( \sigma  \) and we show that \( P_{k+1}
\) is locally stratified w.r.t.~\( \sigma  \). We proceed by cases
depending on the transformation rule which is applied to derive \(
P_{k+1} \) from~\( P_{k} \).

\medskip

\noindent \emph{Case} 1. Program \( P_{k+1} \) is derived by
definition introduction (rule R1). We have that \( P_{k+1}=P_{k}\cup
\{\delta _{1},\ldots ,\delta _{m}\} \), where \( P_{k} \) is locally
stratified w.r.t.~\( \sigma  \) by the inductive hypothesis. Since
$P_{0},\ldots ,P_{n}$ is an admissible transformation sequence, 
\(\{\delta _{1},\ldots ,\delta _{m}\} \) 
%%is $\sigma$-tight 
%%(see Point~(1) of Definition~\ref{def:adm-transformation}) and,
%%therefore, $\{\delta _{1},\ldots ,\delta _{m}\}$ 
is locally stratified w.r.t.~$\sigma$
(see Point~(1) of Definition~\ref{def:adm-transformation}). Thus, \( P_{k+1} \) is locally
stratified w.r.t.~\( \sigma  \).

\medskip

\noindent \emph{Case} 2. Program \( P_{k+1} \) is derived by
instantiation (rule R2). We have that
$P_{k+1}=(P_{k}-\{\gamma\})\cup\{\gamma_{1},\ldots,\gamma_{h}\}$,
where $\gamma$ is the clause $H\leftarrow B$ and, for
$i=1,\ldots,h$, $\gamma_{i}$ is the clause $(H\leftarrow
B)\{X/\llbracket s_i|X\rrbracket \}$. 

\noindent Take any $i\in\{1,\ldots,h\}$. 
Let $L\{X/\llbracket
s_i|X\rrbracket \}$ be a literal in the body of $\gamma_i$. Let $v$
be any valuation and $v'$ be the valuation such that
$v'(X)=\llbracket s_i|v(X)\rrbracket $ and $v'(Y)=v(Y)$ for every
variable $Y$ different from $X$. We have:

\smallskip

\makebox[3.3cm][l]{$\sigma(v(H\{X/\llbracket s_i|X\rrbracket \}))$}
\makebox[4cm][l]{$= \sigma(v'(H))$}
 (definition of $v'$)

\hspace*{3.3cm} \makebox[4cm][l]{$ \geq\sigma(v'(L))$} ($\gamma$ is
locally stratified w.r.t.~$\sigma$)

\hspace*{3.3cm} \makebox[4cm][l]{$=\sigma(v(L\{X/\llbracket
s_i|X\rrbracket \}))$} (definition of $v'$)

\smallskip

\noindent Thus, $\gamma_i$ is locally stratified w.r.t.~$\sigma$.
Hence, $P_{k+1}$ is locally stratified w.r.t.~$\sigma$.

\medskip

\noindent \emph{Case} 3. Program \( P_{k+1} \) is derived by
positive unfolding (rule R3). We have that \(
P_{k+1}=(P_{k}-\{\gamma \})\cup \{\eta _{1},\ldots ,\eta _{m}\} \),
where \( \gamma  \) is a clause in \( P_{k} \) of the form \(
H\leftarrow G_{L}\wedge A\wedge G_{R} \) and clauses \( \eta
_{1},\ldots ,\eta _{m} \) are derived by unfolding \( \gamma  \)
w.r.t.~\( A \). Since, by the induction hypothesis, \(
(P_{k}-\{\gamma \}) \) is locally stratified w.r.t.~\( \sigma  \),
it remains to show that, for \( i=1,\ldots ,m \), clause \( \eta
_{i} \) is locally stratified w.r.t.~\( \sigma \). For \( i=1,\ldots
,m \), \( \eta _{i} \) is of the form \( (H\leftarrow G_{L}\wedge
B_{i}\wedge G_{R})\vartheta_i \), where $\gamma_i$: \( K
_{i}\leftarrow B_i \) is a clause in a variant of $P_{k}$  
such that $\gamma_i$ has no variable in common with $\gamma$ and
$A\vartheta_i=K_i\vartheta_i$. Take any valuation \( v \) and let
$v'$ be a valuation such that, for every variable $X$ occurring in
$\gamma$ or $\gamma_i$, $v'(X)=v(X\vartheta_i)$.

Let $G_L\wedge B_i\wedge G_R$ be the conjunction of $s\ (\geq 0)$
literals $L_1,\ldots, L_s$. Without loss of generality, we assume
that $G_{L}\wedge G_{R}$ is $L_1\wedge\ldots\wedge L_r$ and $B_i$ is
$L_{r+1}\wedge\ldots\wedge L_s$, with $0\leq r \leq s$.

\noindent For $j=1,\ldots,r$, we have:

\smallskip

\makebox[1.75cm][l]{$\sigma(v(H\vartheta_i))$} \makebox[3cm][l]{$=
\sigma(v'(H))$} (definition of $v'$)

\hspace*{1.75cm} \makebox[3cm][l]{$ \geq\sigma(v'(L_j))$} ($\gamma$
is locally stratified w.r.t.~$\sigma$)

\hspace*{1.75cm} \makebox[3cm][l]{$=\sigma(v(L_j\vartheta_i))$}
(definition of $v'$)

\smallskip

\noindent \noindent For $j=r+1,\ldots,s$, we have:

\smallskip

\makebox[1.75cm][l]{$\sigma(v(H\vartheta_i))$} \makebox[3cm][l]{$=
\sigma(v'(H))$} (definition of $v'$)

\hspace*{1.75cm} \makebox[3cm][l]{$ \geq\sigma(v'(A))$} ($\gamma$ is
locally stratified w.r.t.~$\sigma$)

\hspace*{1.75cm} \makebox[3cm][l]{$=\sigma(v'(K_i))$} (definition of
$v'$ and because $A\vartheta_i=K_i\vartheta_i$)

\hspace*{1.75cm} \makebox[3cm][l]{$ \geq\sigma(v'(L_j))$}
($\gamma_i$ is locally stratified w.r.t.~$\sigma$)

\hspace*{1.75cm} \makebox[3cm][l]{$=\sigma(v(L_j\vartheta_i))$}
(definition of $v'$)

\smallskip

\noindent Thus, the clause \( \eta _{i} \) is locally stratified
w.r.t.~\( \sigma  \).

\medskip

\noindent \emph{Case} 4. Program \( P_{k+1} \) is derived by
negative unfolding (rule R4). We have that \(
P_{k+1}=(P_{k}-\{\gamma \})\cup \{\eta _{1},\ldots ,\eta _{r}\} \),
where \( \gamma  \) is a clause in \( P_{k} \) of the form \(
H\leftarrow G_{L}\wedge \neg A\wedge G_{R} \) and clauses \( \eta
_{1},\ldots ,\eta _{r} \) are derived by negative unfolding \(
\gamma \) w.r.t.~\( \neg A \). Since, by the inductive hypothesis,
\( (P_{k}-\{\gamma \}) \) is locally stratified w.r.t.~\( \sigma \),
it remains to show that, for \( j=1,\ldots ,r \), clause \( \eta
_{j} \) is locally stratified w.r.t.~\( \sigma  \).

Let $\gamma_1$: $K_{1}\leftarrow B_{1},$ $\ldots ,$ $\gamma_m$:
$K_{m}\leftarrow B_{m}$ be the clauses in a variant of \( P_{k} \)
such that, for \( i=1,\ldots ,m \), \(A\! =\! K_{i}\vartheta_i\) for
some substitution $\vartheta_i$. Then,  for \( j=1,\ldots ,r \), \(
\eta _{j} \) is of the form \( H\leftarrow L_{j1}\wedge\ldots\wedge
L_{js}\) and, by construction, for $p=1,\ldots,s$, $L_{jp}$ is a
literal such that either (Case~a)~$L_{jp}$ is an atom that occurs
positively in $G_L\wedge G_R$, or (Case~b)~$L_{jp}$ is a negated
atom that occurs in $G_L\wedge G_R$, or (Case~c)~$L_{jp}$ is an atom
$M$ and $\neg M$ occurs in $B_i\vartheta_i$, for some $i\in
\{1,\ldots ,m\} $, or (Case~d)~$L_{jp}$ is a negated atom $\neg M$
and $M$ is an atom that occurs positively in $B_i\vartheta_i$, for
some $i\in \{1,\ldots ,m\} $.

Take any $j\in\{1,\ldots,h\}$.
Take any $p\in\{1,\ldots,s\}$. Take any valuation $v$. In Cases~(a) and (b) we have
$\sigma(v(H))\geq \sigma(v(L_{jp}))$ because, by the inductive
hypothesis, $\gamma$ is locally stratified w.r.t.~$\sigma$. 
In Case~(c) we have:

\smallskip

\makebox[1.4cm][l]{$\sigma(v(H))$} \makebox[3cm][l]{$
>\sigma(v(A))$} ($\gamma$ is locally stratified w.r.t.~$\sigma$ and

\makebox[48mm][l]{}
$\neg A$ occurs in the body of $\gamma$)

\hspace*{1.4cm} \makebox[3cm][l]{$=\sigma(v(K_i\vartheta_i))$}
($A=K_i\vartheta_i$)

\hspace*{1.4cm} \makebox[3cm][l]{$>\sigma(v(L_{jp}))$} ($\gamma_i$
is locally stratified w.r.t.~$\sigma$)

\smallskip

\noindent In Case (d) we have:
\smallskip

\makebox[1.4cm][l]{$\sigma(v(H))$} \makebox[3cm][l]{$
\geq\sigma(v(A)) +1$} ($\gamma$ is locally stratified
w.r.t.~$\sigma$ and 

\makebox[48mm][l]{}
$\neg A$ occurs in the body of $\gamma$)

\hspace*{1.4cm} \makebox[3cm][l]{$=\sigma(v(K_i\vartheta_i))+1$}
($A=K_i\vartheta_i$)

\hspace*{1.4cm} \makebox[3cm][l]{$\geq\sigma(v(L_{jp}))+1$}
($\gamma_i$ is locally stratified w.r.t.~$\sigma$)

\smallskip

\noindent Thus, \( \eta _{j}\) is locally stratified w.r.t.~\(
\sigma  \). Hence, $P_{k+1}$ is locally stratified w.r.t.~$\sigma$.

\medskip

\noindent \emph{Case} 5. Program \( P_{k+1} \) is derived by
subsumption (rule R5). \( P_{k+1} \) is locally stratified w.r.t.~\(
\sigma  \) by the inductive hypothesis because \( P_{k+1}\subseteq
P_{k} \).

\medskip

\noindent \emph{Case} 6. Program \( P_{k+1} \) is derived by
positive folding (rule R6). We have that \( P_{k+1}=(P_{k}-\{\gamma
\})\cup \{\eta \} \), where \( \eta  \) is a clause of the form \(
H\leftarrow  B_L\wedge K\vartheta\wedge B_R \) derived by positive
folding of clause \( \gamma \) of the form \( H\leftarrow B_L\wedge
B\vartheta\wedge B_R \) using a clause \( \delta \) of the form \(
K\leftarrow B \ \in \mathit{Defs}_k\). We have to show that $\eta$
is locally stratified w.r.t.~$\sigma$, that is, for every valuation
\(v\), \(\sigma(v(H))\geq \sigma(v(K)\vartheta) \). 

Take any
valuation~$v$. 
By the inductive hypothesis, since $\gamma$ is locally stratified 
w.r.t.~$\sigma$, we have that: $(\alpha)$~for every literal $L$
occurring in $B_L\wedge B\vartheta \wedge B_R$, we have \(\sigma( v(H))\geq
\sigma(v(L))\). 

By the applicability conditions of rule R6, clause
$\delta$ is the unique clause defining the predicate of its head and, by the
hypothesis that the transformation sequence is admissible, this
definition is $\sigma$-tight (see Point~(2) of
Definition~\ref{def:adm-transformation}). Thus, for every
valuation~$v'$, we have that: (1)~for every $L$ in $B$, 
$\sigma(v'(K))\geq \sigma(v'(L))$,
and (2)~there exists an atom $A$ in $B$ such that 
$\sigma(v'(K)) = \sigma(v'(A))$.

Let the valuation $v'$ be defined as follows: for every variable $X$,
$v'(X)=v(X\vartheta)$. Then, we have that: 
$(\beta.{\rm 1})$~for every $L$ in $B$, $\sigma(v(K\vartheta))\geq \sigma(v(L\vartheta))$,
and $(\beta.{\rm 2})$~there exists an atom $A$ in $B$ such that 
$\sigma(v(K\vartheta)) = \sigma(v(A\vartheta))$. 
Thus, from $(\alpha)$, $(\beta.{\rm 1})$, and $(\beta.{\rm 2})$, we get that
$\sigma(v(H)) \geq \sigma(v(K\vartheta))$. Hence, $\eta$ is locally stratified
w.r.t.~$\sigma$.

\medskip

\noindent \emph{Case} 7. Program \( P_{k+1} \) is derived by
negative folding (rule R7). We have that \( P_{k+1}=(P_{k}-\{\gamma
\})\cup \{\eta \} \) and, by the hypothesis that the transformation
sequence is admissible, \( \eta  \) is locally stratified
w.r.t.~$\sigma$ (see Point~(3) of
Definition~\ref{def:adm-transformation}).
\end{proof}

\medskip
In the rest of this Appendix we will consider:

\noindent
\makebox[6mm][r]{(i)}~a local stratification~\( \sigma:
\mathcal{B}_{\omega} \rightarrow W \), 

\noindent
\makebox[6mm][r]{(ii)}~an $\omega$-program $P_0$ 
which is locally stratified w.r.t.~$\sigma$, and 

\noindent
\makebox[6mm][r]{(iii)}~an
admissible transformation sequence $P_0,\ldots,P_n$.

\begin{definition}[Old and New Predicates, Old and New Literals]
\label{def:old-new}$ $Each predicate
occurring in $P_0$ is called an {\em old predicate} and each
predicate introduced by rule {\rm{R1}} is called a {\em new
predicate}. An {\em{old literal}} is a literal with an old predicate.
A {\em{new literal}} is a literal with a new predicate.
\end{definition}

Thus, the new predicates are the ones
which occur in the heads of the clauses of $\mathit{Defs}_n$.

Without loss of generality, we will assume that the admissible
transformation sequence \( P_{0},\ldots ,$ $P_{n} \) is of the form \(
P_{0},\ldots ,P_{d},\ldots ,P_{n} \), with \( 0\! \leq \! d\! \leq
\! n \), where:

\smallskip

\noindent \textup{(1)} the sequence \( P_{0},\ldots ,P_{d} \), with
\( d\! \geq \! 0 \), is constructed by applying \( d \) times the
definition introduction rule, and

\noindent \textup{(2)} the sequence \( P_{d},\ldots ,P_{n} \), is
constructed by applying any rule, except the definition introduction
rule R1.

\smallskip{}

\noindent Thus, $P_{d} = P_{0}\cup\mathit{Defs}_n$. In order to
prove the correctness of the admissible transformation sequence
\(P_{0},\ldots ,P_{n}\) (see
Proposition~\ref{prop:preserv-mu-consistency} below) we will show that
$M(P_d)=M(P_n)$. In order to prove
Proposition~\ref{prop:preserv-mu-consistency}, we introduce
the notion of a {\em proof tree}  which is the proof-theoretic counterpart of the perfect model semantics (see
Theorem~\ref{th:proof_trees_perfect_model} below). 
A proof tree for an atom \( A
\in\mathcal{B}_{\omega} \) and a locally stratified $\omega$-program
\( P \) is constructed by transfinite induction as indicated in the
following definition.

\begin{definition}[Proof Tree for Atoms and Negated Atoms]
\label{def:proof_tree}  Let \( A \) be an atom in
$\mathcal{B}_{\omega}$, let \( P \) be a locally stratified
$\omega$-program, and let \( \sigma \) be a local stratification for
\( P \). Let \( PT_{<A} \) denote the set of proof trees for \( H \) and
\( P \), where $H\in\mathcal{B}_{\omega}$ and \( \sigma (H)<\sigma
(A) \). 

\noindent
A \emph{proof tree} for \( A \) and \( P \) is a finite tree
\( T \) such that\/{\rm :} 

\hangindent=8mm\noindent
\makebox[6mm][r]{\rm{(i)}}~the root of \( T \) is labeled by \( A \),

\noindent
\makebox[6mm][r]{\rm{(ii)}}~a node \( N \) of \( T \) has children labeled by  \(
L_{1},\ldots ,L_{r} \) iff \( N \) is labeled by an atom
\(H\in\mathcal{B}_{\omega}\) and there exist a clause \( \gamma \in
P \) and a valuation \( v \) such that \( v(\gamma ) \) is \(
H\leftarrow L_{1}\wedge \ldots \wedge L_{r} \), and 

\noindent
\makebox[6mm][r]{\rm{(iii)}}~every leaf
of \( T \) is either labeled by the empty conjunction \(
\mathit{true} \) or by a negated atom~\( \neg H \), with
\(H\in\mathcal{B}_{\omega}\), such that there is no proof tree for
\( H \) and \( P \) in \( PT_{<A} \).

\hangindent=0mm\smallskip
\noindent Let $A$ be an atom in $\mathcal{B}_{\omega}$
and \( P \) be a locally stratified
$\omega$-program. 

\noindent A \emph{proof tree} for~$\neg A$ and $P$ exists iff there are no proof
trees for~$A$ and $P$. There exists at most one proof tree 
for~$\neg A$ and $P$ and, when it exists, it consists of the single root node labeled by~$\neg A$.

%%\hangindent=0mm\smallskip
%%\noindent Let $A$ be an atom in $\mathcal{B}_{\omega}$ such that
%% there exists no proof
%%tree for~$A$ and $P$. Then the only proof tree for the negated atom 
%%$\neg A$ and $P$ 
%%is a tree consisting of the single root node labeled by $\neg A$.
\end{definition}

\vspace{-3mm} 
\noindent
\begin{remark}\label{rem:proof-tree} 
{\rm(i)}~For any 
$A\in\mathcal{B}_{\omega}$ 
if  there is a proof tree for~$A$ and $P$, then
there is no proof tree for $\neg A$ and~$P$.

\noindent
{\rm(ii)}~In any proof tree if a node $H$ is an ancestor of a 
node $A$ then
$\sigma(H) \geq \sigma (A)$.
\end{remark}

The following theorem, whose proof is omitted, shows that proof trees
can be used for defining a semantics equivalent to the perfect model
semantics.

\begin{theorem}[Proof Tree and Perfect
Model] \label{th:proof_trees_perfect_model}  Let \( P \) be a
locally stratified $\omega$-program. For every \(
A\in\mathcal{B}_{\omega} \), there exists a proof tree for \( A \)
and \( P \) iff \( A\in M(P) \).
\end{theorem}

In order to show that $M(P_d) = M(P_n)$, we will use
Theorem~\ref{th:proof_trees_perfect_model} and we will show that,
given any atom \( A \in\mathcal{B}_{\omega}\), there exists a proof
tree for \( A \) and \( P_{d} \) iff there exists a proof tree for
\( A \) and~\( P_n\). 

\smallskip
In the following, we will use suitable \emph{measures} which we now introduce.

\begin{definition}[Three Measures: $\mathit{size},$ $\mathit{weight},$ $\mu$] 
\label{def:measure} \hangindent=7mm{\rm{(i)}}~For any proof tree $T$, $\mathit{size}(T)$
denotes the number of nodes in~\( T \) labeled by atoms in
$\mathcal{B}_{\omega}$.

\noindent {\rm{(ii)}}~For any atom $A\in \mathcal{B}_{\omega}$, the ordinal $\sigma(A)$ is said 
to be the {\em{stratum}} of $A$. 

\hangindent=7mm\hspace{3mm}For any ordinal $\alpha\in W$, for any proof tree $T$,
$\mbox{\it weight}(\alpha,T)$ is the number of nodes 
of $T$ whose label is an atom with stratum $\alpha$.
{\rm{(}}Recall that $\mathit{true}$, that is,
the empty conjunction of literals, is {\it{not}} an atom.{\rm{)}}

%\hspace*{8mm}$\{ N \mid $ there exists an atom $A\in \mathcal{B}_{\omega}$ 
%which labels node $N$ and 
%$\sigma(A)\!=\!\alpha\}$.

%number of nodes $N$ in $T$
%such that there exists an atom \(A\in\mathcal{B}_{\omega} \) which is the label
%of $N$ and $\sigma(A)=\alpha$.

\noindent {\rm{(iii)}}~For any atom \(A\in\mathcal{B}_{\omega} \), we
define\/{\rm{:}}

\smallskip

\hspace*{8mm}$\mbox{\it min-weight}(A) =_{\mathit{def}} \min \{\mbox{\it
weight}(\alpha,T) \mid \sigma(A)\!=\!\alpha \mbox{~and~}$

\makebox[69mm][l]{}$T \mbox{ is a proof tree for $A$ and } P_d\}$.

\smallskip

\noindent {\rm{(iv)}}~For any atom \(A\in\mathcal{B}_{\omega} \) such that 
 there exists at least a proof tree for $A$ and $P_d$, we
define\/{\rm{:}}
\smallskip{}

\hspace*{8mm}\makebox[18mm][l]{$ \mu (A)=_{\mathit{def}}$}
\makebox[45mm][l]{$\langle \sigma (A), \mbox{\it
min-weight}(A) \rangle $}  ${\mathit{if}}~A$ is an old atom

\hspace*{8mm}\makebox[18mm][l]{$ \mu (A)=_{\mathit{def}}$}
\makebox[45mm][l]{$\langle \sigma (A), \mbox{\it
min-weight}(A)-1 \rangle \)}  ${\mathit{if}}~A$ is a new atom

\smallskip

\noindent {\rm{(v)}}~For any atom \(A\in\mathcal{B}_{\omega} \) 
such that there exists no proof tree for $A$ and $P_d$, we
define\/{\rm{:}}

\smallskip

\hspace*{8mm}\makebox[18mm][l]{$\mu (\neg A)=_{\mathit{def}}$} $\langle \sigma (A),0\rangle \).

\end{definition}

\begin{remark}\label{rem:measures}
\noindent {\rm(i)}~{\it{If}} $A$ is an old atom {\it{then}}
 \mbox{\it{min-weight}}$(A)\!>\!0$ {\it{else}}  \mbox{\it{min-weight}}$(A)\!\geq\! 0$.
 
\noindent
{\rm(ii)}~For any atom  \(A\in\mathcal{B}_{\omega} \), \( \mu (A) \) is 
undefined if there is no proof
tree for \( A \) and \( P_d \).
\end{remark}

Now we extend \( \mu \) to conjunctions of literals. First, we
introduce the binary operation \( \oplus:\, (W\times \mathbb
N)^{2}\rightarrow (W\times \mathbb{N}) \), where $W$ is the set of
countable ordinals and $\mathbb N$ is the set of natural numbers,
defined as follows:

\smallskip

$\langle \alpha _{1},m_{1}\rangle \oplus \langle \alpha
_{2},m_{2}\rangle = \left\{ \begin{array}{lll} \langle \alpha
_{1},\, m_{1}\rangle && \mathit{if}\ \ \alpha _{1}> \alpha _{2}\\
\langle \alpha _{1},\, m_{1}\! +\!
m_{2}\rangle && \mathit{if}\ \ \alpha _{1}=\alpha _{2} \\
\langle \alpha _{2},\, m_{2}\rangle && \mathit{if}\ \ \alpha _{1}<
\alpha _{2}
\end{array}\right.$

\smallskip{}

\noindent or equivalently,

\smallskip{}

$\langle \alpha _{1},m_{1}\rangle \oplus \langle \alpha
_{2},m_{2}\rangle =$

$=\langle \max(\alpha _{1},\alpha _{2}), 
\bif \alpha_1\!=\!\alpha_2 \bth m_1\!+\!m_2 
\bel \!(\! \!\bif \alpha_1\!>\!\alpha_2 \bth m_1 \!\bel m_2)\rangle$

\smallskip{}

\noindent Given a conjunction of literals \( L_{1}\wedge \ldots
\wedge L_{r} \) such that, for $i=1,\ldots,r$, with $r\!\geq\! 1$, 
there is a proof tree
for~$L_i$ and $P_d$, we define:
\smallskip{}

\( \mu (L_{1}\wedge \ldots \wedge L_{r})=_{\mathit{def}} \mu (L_{1})\oplus \cdots
\oplus \mu (L_{r}) \)

\smallskip

\noindent For  $\mathit{true}$, which is the empty conjunction of literals, we define:
\smallskip

$\mu(\mathit{true})=_{\mathit{def}}\langle0,0\rangle$

\smallskip
\noindent
Note that the definition of  $\mu(\mathit{true})$ is consistent with
the fact that $\mathit{true}$ is 
the neutral element for $\wedge$ and, thus, 
$\mu(\mathit{true})$ should be the neutral element for $\oplus$, which 
is $\langle0,0\rangle$.
 
% 
%We further 
%extend the definition of $\mu$ by stipulating that $\mu(\mathit{true})=\langle0,0\rangle$ because $\mathit{true}$. This definition is consistent with
%the fact that $\mathit{true}$ is the
%empty conjunction, that is, the neutral element for $\wedge$ and, thus, 
%$\mu(\mathit{true})$ should be the neutral element for $\oplus$, that is,
% $\langle0,0\rangle$.

 The following lemma follows from the definition of the measure $\mu$. Recall that
a new predicate can only be defined in terms of old predicates.

\begin{lemma}[Properties of $\mu$ for a Definition in $P_d$]
\label{lemma:mu-for-a-definition}
\noindent Let $\delta\in P_d$ be a $\sigma$-tight
clause introduced by the
definition rule {\rm{R1}} with $m\!=\!1$, that is, $\delta$
is the only clause defining the head predicate of $\delta$ in $P_d$.
Let $v$ be a valuation and  $v(\delta)$  be of the form\/{\rm{:}}
$K\leftarrow L_1\wedge \ldots \wedge L_q$.
We have that\/{\rm{:}} $\mu(K)=\mu(L_1)\oplus \ldots \oplus \mu(L_q)$.
\end{lemma}

\begin{proof}  Without loss of generality, we may assume that $L_1$ is an atom and $\sigma(K)=\sigma(L_1)$ because $\delta$ is $\sigma$-tight
and, thus, $L_1$ is $\sigma$-maximal. We have that:

\noindent
Thus

\smallskip
 $\mu(K)=\langle \sigma(K), \;${\it min}-{\it weight}$(K)\!-\!1\rangle$.

\smallskip
\noindent
Now, {\it min}-{\it weight}$(K) =$ \{by definition of 
{\it min}-{\it weight}\} $= $

\smallskip \indent
$= \min \{${\it weight}$(\sigma(K),T_K)\}$ where $T_K$ is a 
proof tree for $K$ and $P_d$ $=$

\smallskip \indent
= \{by definition of {\it weight} (see also Figure~\ref{fig:sigma-tight-lemma})\} =

\smallskip \indent
$= \big(\!\min \sum_{i=1,\ldots,q}$ {\it weight}$(\sigma(K),T_i)\!\big)\!+\!1$ 

\smallskip 
\makebox[20mm][l]{}where for $i\!=\!1,\ldots,q$, $T_i$ is a proof tree for $L_i$ and $P_d$ $=$

\smallskip \indent
= \{by definition of {\it weight} and Remark~\ref{rem:proof-tree}\} $=$

\smallskip \indent
$= \big(\!\min \sum_{i=1,\ldots,q\ \wedge\ \sigma(L_i)=\sigma(K)}$ 
{\it weight}$(\sigma(K),T_i)\!\big)\!+\!1=$  \{by $\min \sum  = \sum \min$\} $=$

\medskip \indent 
$= \big(\!\sum_{i=1,\ldots,q\ \wedge\ \sigma(L_i)=\sigma(K)} 
${\it min}-{\it weight}$(L_i)\!\big)\!+\!1=$  \{by $\sigma$-tightness\} $=$

\medskip \indent 
$= \big(\!\sum_{i=1,\ldots,q\ \wedge\ \sigma(L_i)=\sigma(L_1)} 
${\it min}-{\it weight}$(L_i)\!\big)\!+\!1$.

\smallskip\noindent 
Thus, 

\smallskip\noindent 
$\mu(K)=\langle \sigma(L_1),\ \sum_{i=1,\ldots,q\ \wedge\ \sigma(L_i)=\sigma(L_1)} 
${\it min}-{\it weight}$(L_i)\rangle$ = \{by definition of $\oplus\} =$

\smallskip \indent 
= $\mu(L_1)\oplus \ldots \oplus \mu(L_q)$.
\end{proof}

\begin{figure}[ht]
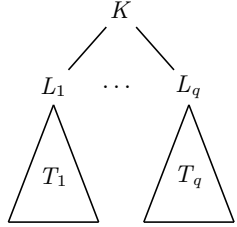

\begin{center}
\VCDraw{%
\begin{VCPicture}{(3,0)(5,5.3)} 
\FixStateDiameter{10mm}
\ChgStateLineWidth{1.3}
\ChgStateLineColor{white}
\SetEdgeArrowWidth{0pt}
% states 
\State[K]{(2.5,4.7)}{K}
\State[L_1]{(1,3)}{L1}
\State[L_q]{(4,3)}{Lq}
\State[T_1]{(1,1)}{T1}
\State[T_q]{(4,1)}{Tq}
\State[\ldots]{(2.4,3)}{ldots}

\Point{(1,2.6)}{T1top}
\Point{(0,0)}{T11}
\Point{(2,0)}{T12}

\Point{(4,2.6)}{Tqtop}
\Point{(3,0)}{Tq1}
\Point{(5,0)}{Tq2}
\EdgeL{K}{L1}{}
\EdgeL{K}{Lq}{}
\EdgeL{T1top}{T11}{}
\EdgeL{T1top}{T12}{}
\EdgeL{T11}{T12}{}

\EdgeL{Tqtop}{Tq1}{}
\EdgeL{Tqtop}{Tq2}{}
\EdgeL{Tq1}{Tq2}{}
\end{VCPicture}}
\end{center}
\vspace{-3mm}
\caption{A proof tree for $K$ and $P_d$. There is a valuation $v$ and  a clause $\delta\in P_d$ such that $v(\delta)$ is of the form:
$K\leftarrow L_1\wedge \ldots \wedge L_q$.
For $i=1,\ldots,q$, $T_i$ is a proof tree for $L_i$ and~$P_d$.
\label{fig:sigma-tight-lemma}}
\end{figure}

\medskip

\noindent Let $>$ denote the usual greater-than relation 
on $\mathbb N$. 
Let $>_{\mathit{lex}}$ denote the lexicographic ordering
over $W\times\mathbb N$. 

Let $\pi_1$ and $\pi_2$ denote, respectively, the first and second projection function on pairs. Given a pair
$A=\langle a,b\rangle$ by $A_1$ we denote $a$ and by $A_2$ we denote
$b$.

\begin{lemma}[Properties of $\oplus$]
\label{lemma:properties_of_oplus}
\noindent\makebox[5mm][r]{\rm{(i)}}~$\oplus$ is an associative, commutative
binary operator.

\noindent\makebox[5mm][r]{\rm{(ii)}}~For every $A,B,C\in W\times \mathbb N$ and
${\mathcal R}\in\{\geq_{\mathit lex},>_{\mathit lex}\}$, we have 
that\/{\rm{:}}

\noindent\makebox[8mm][r]{\rm{(ii.1)}}~$A \oplus B\geq_{\mathit{lex}} A$

\noindent\makebox[8mm][r]{\rm{(ii.2)}}~if \makebox[57mm][l]{$A \geq_{\mathit lex}
B$} then $A\oplus C\geq_{\mathit{lex}} B\oplus C$

\noindent\makebox[8mm][r]{\rm{(ii.3)}}~if \makebox[57mm][l]{$A>_{\mathit lex}
B$, $A_1\geq C_1$, and $A_2>0$} then $A\oplus C>_{\mathit{lex}}
B\oplus C$

\noindent\makebox[8mm][r]{\rm{(ii.4)}}~if \makebox[57mm][l]{$A\ {\mathcal R}\ B$
and $A_1>C_1$} then $A\ {\mathcal R}\ B\oplus C$

\noindent\makebox[8mm][r]{\rm{(ii.5)}}~if  \makebox[57mm][l]{$A\ {\mathcal R}\
B\oplus C$} then $A\ {\mathcal R}\ B$ and $A\ {\mathcal R}\ C$
\end{lemma}
 
\begin{proof} (i)~It follows immediately from the definition.

\smallskip
\noindent(ii.1)~By cases. If $A_1>B_1$ then $A \oplus B =
A\ \geq_{\mathit{lex}}A$. If $A_1=B_1$ then $A \oplus B  =
\langle A_1,A_2+B_2\rangle\ \geq_{\mathit{lex}}A$. If
$B_1>A_1$ then $A \oplus B = B >_{\mathit{lex}}A$.

%%%%%%%
\medskip
\noindent(ii.2)~
Let us consider the following two pairs:

$(\alpha)=_{\mathit{def}}A\oplus C=$

$= \langle \max(A_1,C_1), \bif A_1\!=\!C_1\bth A_2\!+\!C_2 
\bel \! \bif A_1\!>\!C_1 \bth A_2 \bel C_2\rangle$ 

\noindent
and 

$(\beta)=_{\mathit{def}}B\oplus C =$

$= \langle \max(B_1,C_1),
 \bif B_1\!=\!C_1\bth B_2\!+\!C_2 
\bel \! \bif B_1\!>\!C_1 \bth B_2 \bel C_2\rangle$.

\noindent 
We have to show that $(\alpha)\gel (\beta)$.

\noindent 
Since $A\gel B$, there are two cases. Case~(1): $A=B$,  and  
Case~(2): $A\gl B$. Case~(2) consists of two subcases:
Case~(2.1): $A_1>B_1$,  and Case~(2.2): $A_1=B_1$  and $A_2>B_2$.

\noindent
In Case~(1) we have that $(\alpha)= (\beta)$. Thus, we get $(\alpha)\gel (\beta)$ as desired.

\noindent
In Case~(2.1) we consider two subcases: Case~(2.1.1): $A_1>B_1$
and $B_1\geq C_1$, and Case~(2.1.2): $A_1>B_1$
and $B_1< C_1$.

\noindent
In Case~(2.1.1) we have that $\max(A_1,C_1)>\max(B_1,C_1)$ and thus, 
we get that $(\alpha)\gel (\beta)$.

\noindent
In Case~(2.1.2) $(\beta)$ reduces to $\langle C_1,  C_2\rangle$
and, since $A_1>B_1\geq C_1$,  we get that $(\alpha)\gel (\beta)$.

\noindent
In Case~(2.2) since $A_1=B_1$, $(\beta)$ reduces to

$\langle \max(A_1,C_1), \bif A_1=C_1\bth B_2\!+\!C_2 
\bel \! \bif A_1>C_1 \bth B_2 \bel C_2\rangle$

\noindent
and, since $A_2>B_2$, we get that $(\alpha)\gel (\beta)$.

%
%\smallskip
%\indent(ii.2)~We assume that $A \geq_{\mathit lex} B$ and we proceed by cases.

%\noindent
%(1)~Let $\pi_1(A)>\pi_1(C)$, then $\pi_1(A\oplus C)=\pi_1(A)$, $\pi_1(B\oplus
%C)\leq\pi_1(A)$, and we have two subcases: (1.a)~if $\pi_1(B)>\pi_1(C)$ then
%$\pi_2(B\oplus C)=\pi_2(B)$ and, by the assumptions, $A\oplus
%C\geq_{\mathit{lex}} B\oplus C$, and (2.a)~if $\pi_1(B)\leq\pi_1(C)$ then
%$\pi_1(B\oplus C)<\pi_1(A)$ and $A\oplus C\geq_{\mathit{lex}} B\oplus C$.

%\noindent
%(2)~Let $\pi_1(A)=\pi_1(C)$, then $\pi_2(A\oplus C)=\pi_2(A)+\pi_2(C)$, and we
%have three subcases: (2.a)~$\pi_1(B)>\pi_1(C)$ is impossible because $A
%\geq_{\mathit lex} B$, (2.b)~if $\pi_1(B)=\pi_1(C)$ then $\pi_2(B\oplus
%C)=\pi_2(B)+\pi_2(C)$ and, since $\pi_2(A)\geq\pi_2(B)$, $A\oplus
%C\geq_{\mathit{lex}} B\oplus C$, and (2.c)~if $\pi_1(B)<\pi_1(C)$ then
%$\pi_2(B\oplus C)=\pi_2(C)$, $\pi_1(B\oplus C)=\pi_1(C)$, and, thus, $A\oplus
%C\geq_{\mathit{lex}} B\oplus C$. 

%\noindent
%Finally, (3)~let $\pi_1(A)<\pi_1(C)$, then
%$A\oplus C = B\oplus C = C$ and, therefore, $A\oplus C\geq_{\mathit{lex}}
%B\oplus C$.

%%%%%%%
\medskip
\noindent(ii.3)~Let us consider again the two pairs:

$(\alpha)=_{\mathit{def}}A\oplus C=$

$=\langle \max(A_1,C_1), \bif A_1\!=\!C_1\bth A_2\!+\!C_2 
\bel \! \bif A_1\!>\!C_1 \bth A_2 \bel C_2\rangle$ 

\noindent
and 

$(\beta)=_{\mathit{def}}B\oplus C =$

$= \langle 
\max(B_1,C_1), \bif B_1\!=\!C_1\bth B_2\!+\!C_2 
\bel \! \bif B_1\!>\!C_1 \bth B_2 \bel C_2\rangle$.

\noindent
We have to show $(\alpha)\gl (\beta)$.

\noindent
Since $A\gl B$ there are two cases. Case~(1): $A_1>B_1$ and $A_1\geq C_1$ 
and $A_2>0$. Case~(2):  $A_1=B_1$ and $A_2> B_2$ 
and $A_1\geq C_1$ and $A_2>0$.

\noindent
For Case~(1) we consider two subcases: Case~(1.1) $A_1=C_1$ and 
Case~(1.2) $A_1>C_1$.

\noindent
In Case (1.1) we have that $(\alpha)$ reduces to $\langle C_1,A_2+C_2\rangle $
and 

\noindent
$(\beta)$ reduces to $\langle C_1, \bif B_1=C_1\bth B_2\!+\!C_2 
\bel \! \bif B_1>C_1 \bth B_2 \bel C_2 \rangle $

\noindent
and since $A_1>B_1$ and $A_1=C_1$, we get that $(\beta)$ further reduces to 
$\langle C_1,  C_2 \rangle $ and, since $A_2>0$, we get that $(\alpha)\gl (\beta)$.

\noindent
In Case (1.2) we have that $(\alpha)$ reduces to 
$\langle A_1,  \ldots \rangle $ and 
$(\beta)$ reduces to 

\noindent
$\langle \max(B_1,C_1), \ldots \rangle $, and since 
$A_1>B_1$  and $A_1>C_1$ we 
get that $(\alpha)\gl (\beta)$.

\smallskip

\noindent
For Case~(2) we consider two subcases: Case~(2.1) $A_1=B_1=C_1$ and 
Case~(2.2) $A_1=B_1>C_1$.

\noindent
In Case (2.1) we have that $(\alpha)$ reduces to 
$\langle A_1,  A_2+C_2 \rangle $ and 
$(\beta)$ reduces to $\langle A_1,  B_2+C_2 \rangle $, and since 
in Case~(2) we have that $A_2>B_2$, we get that $(\alpha)\gl (\beta)$.

\noindent
In Case (2.2) we have that $(\alpha)$ reduces to 
$\langle A_1,  A_2 \rangle $ and 
$(\beta)$ reduces to $\langle B_1,  B_2 \rangle $, and since 
$A_1=B_1$ and in Case~(2) we have that $A_2>B_2$, we 
get that $(\alpha)\gl (\beta)$.

%
%\smallskip
%\indent(ii.3)~We assume that $A >_{\mathit lex} B$, $\pi_1(A)\geq\pi_1(C)$, and
%$\pi_2(A)>0$ and we proceed by cases. 

%\noindent
%(1)~Let $\pi_1(A)>\pi_1(C)$, then $A\oplus
%C=A$ and, since by hypothesis $\pi_1(A)\geq\pi_1(B)$, we have two subcases:
%(1.a)~if $\pi_1(A)>\pi_1(B)$ then $\pi_1(A)>\pi_1(B\oplus C)$ and
%$A\oplus C>_{\mathit{lex}} B\oplus C$, and (1.b)~if $\pi_1(A)=\pi_1(B)$ then
%$\pi_1(B)>\pi_1(C)$, $B\oplus C=B$, and $A\oplus C>_{\mathit{lex}} B\oplus C$.

%\noindent
%(2)~Now let $\pi_1(A)=\pi_1(C)$, since $\pi_2(A)>0$, we have $\pi_2(A\oplus
%C)>\pi_2(C)$ and we have two subcases: (2.a)~if $\pi_1(A)>\pi_1(B)$ then
%$B\oplus C=C$ and $A\oplus C>_{\mathit{lex}} B\oplus C$, and (2.b)~if
%$\pi_1(A)=\pi_1(B)$ then, since $A>_{\mathit lex}B$ and, thus,
%$\pi_2(A)>\pi_2(B)$, we have $\pi_2(A\oplus C)>\pi_2(B\oplus C)$ and $A\oplus
%C>_{\mathit{lex}} B\oplus C$.

\medskip
\noindent(ii.4)~We have that:

$B\oplus C=\langle \max(B_1,C_1), \bif B_1=C_1\bth B_2\!+\!C_2 
\bel \bif B_1>C_1 \bth B_2 \bel C_2\rangle$

\noindent
We reason by cases. Case~(1): we assume $A=B$ and $A_1>C_1$
and we show $A\gel B \oplus C$. Case~(2):  we assume $A\gl B$ and $A_1>C_1$
and we show $A\gl B \oplus C$.

\smallskip

\noindent
Case~(1). Since $A=B$, from $A_1>C_1$ we get that $B_1>C_1$ and thus, $B \oplus C=B$.
Thus, $A\gel B \oplus C$.

\smallskip
\noindent
Case~(2). There are two subcases: (2.1)~$A_1>B_1$ and $A_1>C_1$, and
(2.2)~($A_1=B_1$ and $A_2>B_2$) and $A_1>C_1$.

\noindent
Case~(2.1). We have that: $A_1>\max(B_1,C_1)$ and thus, $A\gel B \oplus C$.

\noindent
Case~(2.2). Since $A_1=B_1$ and $A_1>C_1$, we have that: $B \oplus C=\langle B_1,B_2 \rangle=_{\mathit{def}}B$. Since
$A_1=B_1$ and $A_2>B_2$ we get $A\gl B$, and thus, $A\gl B \oplus C$.

%\smallskip
%\indent(ii.4)~Let us assume $A\ {\mathcal R}\ B$, for ${\mathcal
%R}\in\{\geq_{\mathit lex},>_{\mathit lex}\}$, and $\pi_1(A)>\pi_1(C)$, then we
%have $\pi_1(A)\geq\pi_1(B)$. We have two cases: (1)~if $\pi_1(A)>\pi_1(B)$ then
%${\mathcal R}$ is $>_{\mathit lex}$, $\pi_1(A)>\pi_1(B\oplus C)$ and, thus,
%$A >_{\mathit lex}B\oplus C$, and (2)~if $\pi_1(A)=\pi_1(B)$ then
%$\pi_1(B)>\pi_1(C)$ and, thus, $B\oplus C = B$, which entails $A\ {\mathcal R}\ 
%B\oplus C$.

%%% alberto
\medskip
\noindent(ii.5)~We have that:

$B\oplus C=\langle \max(B_1,C_1), \bif B_1=C_1\bth B_2\!+\!C_2 
\bel \bif B_1>C_1 \bth B_2 \bel C_2\rangle$

\noindent
We reason by cases: Case~(1)~$A=B\oplus C$, and Case~(2)~$A\gl B\oplus C$.
In order to show Point~(ii.5) in Case (1) we have to show 
$A\gel B$  and $A\gel C$,
and in Case~(2) we have to show $A\gl B$  and $A\gl C$.

\smallskip
\noindent
Case~(1)~Assume $A=B\oplus C$. 

\noindent Case (1.1): $B_1=C_1$. Thus, $A_1=B_1=C_1$
and $A_2=B_2+C_2$. Thus, $A\gel B$  and $A\gel C$.

\noindent
Case (1.2): $B_1>C_1$. Thus, $A_1=B_1$
and $A_2=B_2$. Thus, $A\gel B$  and $A\gel C$.

\noindent
Case (1.3): $B_1<C_1$. Like Case~(1.2), by interchanging
$B$ and $C$. 

\smallskip
\noindent
Case~(2)~Assume $A\gl B\oplus C$. 

\noindent Case (2.1): $A_1>\max(B_1,C_1)$. 
We get: $A\gl B$  and $A\gl C$. 

\noindent Case~(2.2): $A_1=\max(B_1,C_1)$.

\noindent 
Case~(2.2.1): $B_1=C_1$. We have:  $A_1=B_1=C_1$ and, since $A\gl B\oplus C$, 
we have: $A_2>B_2+C_2$. Thus, we get $A\gl B$  and $A\gl C$. 

\noindent 
Case~(2.2.2): $B_1>C_1$. Thus, $A_1=\max(B_1,C_1)=B_1$. Since $A\gl B\oplus C$ and  $A_1=\pi_1( B\oplus C)$, we
 have:  $A_2>\pi_2( B\oplus C)$, that is, 
 
 $A_2>\bif B_1\!=\!C_1\bth B_2\!+\!C_2 
\bel \bif B_1>C_1 \bth B_2 \bel C_2$, that is,

$A_2> B_2$. 

\noindent Thus, we get $A\gl B$  and, since $B_1>C_1$, we also get 
$A\gl C$. 

\noindent 
Case~(2.2.3): $B_1<C_1$. Like Case~(2.2.2), by interchanging $B$ and $C$.
\end{proof}

\begin{notation}
{\rm{By $\overline L$ we will denote  the
negative literal~$\neg L$, if $L$ is a positive literal, and the positive
literal $A$, if $L$ is the negative literal $\neg A$.}}
\end{notation}

\begin{lemma}\label{lemma:properties_of_mu}
For all atoms \(A\in\mathcal{B}_{\omega}\), literals $L_1,\ldots,L_m$, which are either atoms in $\mathcal{B}_{\omega}$ or
negation of atoms in $\mathcal{B}_{\omega}$, if for
$i=1,\ldots,m$, $\sigma(A)\geq\sigma(L_i)$ then
$\mu(\overline{A})\geq_{\mathit{lex}}\mu(\overline{L}_1)\oplus\cdots
\oplus\mu(\overline{L}_m)$.
\end{lemma}

\begin{proof}
The proof is by induction on $m$ by recalling that the $\oplus$ is 
associative and commutative. We do the induction step. The base case 
can be proved similarly to Cases~(1) and (2.1) below.

We assume that $\mu(\overline{A})\geq_{\mathit
lex}\mu(\overline{L}_1)\oplus\cdots\oplus\mu(\overline{L}_j)$, for some $j\geq 1$, and we show that $\mu(\overline{A})\geq_{\mathit
lex}\mu(\overline{L}_1)\oplus\cdots
\oplus\mu(\overline{L}_{j}) \oplus\mu(\overline{L}_{j+1})$.

By definition, $\mu(\overline{A})=\langle\sigma(A),0\rangle$. Let
$\mu(\overline{L}_1)\oplus\cdots\oplus\mu(\overline{L}_j)=\langle\beta,
w_1\rangle$, for some $\beta \in W$  and 
$w_1\in\mathbb{N}$. Thus, the induction hypothesis can be stated as follows:
$\langle\sigma(A),0\rangle
\geq_{\mathit lex}\langle\beta,w_1\rangle$.

We have the 
following two cases. 

Case (1). Assume that $\overline{L}_{j+1}$ is a positive literal, say $B$. 
Let $\mu(B)$ be $\langle\sigma(B),w_2\rangle$, for some $w_2 \in W$.
Since $\sigma(A)\geq \sigma({L}_{j+1}) > \sigma(B)$,  by
Lemma~\ref{lemma:properties_of_oplus}~(ii.4) we get that
$\mu(\overline{A})\geq_{\mathit lex}\mu(\overline{L}_1)\oplus\cdots\oplus\mu(\overline{L}_{j}) \oplus\mu(B)$. 

Case (2). Assume that $\overline{L}_{j+1}$ is a negative literal, say $\neg B$. Let $\mu(\neg B)$ be $\langle\sigma(B),0\rangle$.
By hypothesis, we have $\sigma(A)\geq\sigma(L_{j+1})=\sigma(B)$. We have three subcases.

\noindent 
Case~(2.1).~$\sigma(B)>\beta$. By induction hypothesis we have that $\langle\sigma(A),0\rangle
\geq_{\mathit lex}\langle\beta,w_1\rangle$. We also have that $\langle\beta,w_1\rangle \oplus\langle\sigma(B),
0\rangle=\langle\sigma(B),0\rangle$  and $\langle\sigma(A),0\rangle
\geq_{\mathit lex}\langle\sigma(B),0\rangle$. Thus, we get 
$\langle\sigma(A),0\rangle
\geq_{\mathit lex} \langle\beta,w_1\rangle \oplus\langle\sigma(B),0\rangle$.

\noindent 
Case~(2.2).~$\sigma(B)=\beta$. By induction hypothesis we 
have that $\langle\sigma(A),0\rangle
\geq_{\mathit lex}\langle\beta,w_1\rangle$.
We also have that 
 $\langle\beta,w_1\rangle\oplus
\langle\sigma(B),0\rangle=\langle\beta, w_1\rangle$.
Thus, we get 
$\langle\sigma(A),0\rangle
\geq_{\mathit lex} \langle\beta,w_1\rangle \oplus\langle\sigma(B),0\rangle$.

\noindent 
Case~(2.3).~$\sigma(B)<\beta$. As Case~(2.2). 
\end{proof} 

\medskip
\noindent Now we introduce the notion of a {\em $\mu$-consistent}
proof tree which will be used in
Proposition~\ref{prop:preserv-mu-consistency} below. This notion is
a generalization of the one of a {\em rank-consistent} proof tree
introduced in~\cite{TaS84}.

\begin{definition}[$\sigma$-max Derived   Clause]
We say that a clause $\gamma$ in a program $P_k$ of the sequence
$P_d,\ldots,P_n$ is a $\sigma${\em -max derived} clause if~$\gamma$ is a
descendant of a clause $\beta$ in $P_{j}$, with $d\!<\! j\!\leq\! k$, such that
$\beta$ has been derived by unfolding a clause $\alpha$ in $P_{j-1}$
w.r.t.~an old $\sigma$-maximal atom. {\rm{(}}Recall that, by definition, a clause is
a descendant of itself.{\rm{)}}
\end{definition}

\begin{definition}[\(\mu\)-consistent Proof Tree] \label{def:mu-consistency}
Let \( A \) be an atom in \(\mathcal{B}_{\omega} \) and $P_k$ be a
program in the transformation sequence \( P_{d},\ldots,P_{n} \). We
say that a proof tree \( T \) for \( A \) and \( P_{k} \) is
\mbox{\emph{\(\mu\)-consistent}} if for all atoms \( H\), all literals \(
L_{1},\ldots ,L_{r} \) which are the children of \( H \) in \( T \),
where $H\leftarrow L_{1}\wedge\ldots \wedge L_{r}$ is a clause
$v(\gamma)$ for some valuation $v$ and some clause $\gamma\in P_k$,  we
have that\/{\rm{:}}

\noindent {\em if} $H$ has a new predicate and~$\gamma$ is not
$\sigma$-max derived 

\noindent 
\makebox[7mm][l]{{\em then}} \( \mu (H)\geq_{\mathit{lex}}\mu
(L_{1})\oplus \cdots \oplus \mu (L_{r}) \) 

\noindent 
\makebox[7mm][l]{{\em else}} \( \mu
(H)>_{\mathit{lex}}\mu (L_{1})\oplus \cdots \oplus \mu (L_{r}) \).

The proof tree for the negated atom $\neg A$ and $P_k$, if any, is $\mu$-consistent.
{\rm{(}}Recall that this proof tree, if it exists, consists of the single root node labeled by~$\neg A$.{\rm{)}}

\end{definition}

Let us consider the following ordering on $\mathcal
B_{\omega}$ which will be used in the proof of
Proposition~\ref{prop:preserv-mu-consistency}.

\begin{definition}[Ordering $\succ$] \label{def:curly-greater}
Given any two atoms 
$A_1,A_2\in
\mathcal B_{\omega}$, we write $A_1 \succ A_2$ if either

\makebox[5mm][l]{\rm{(i)}}~$\mu(A_1)
>_{\mathit{lex}} \mu(A_2)$, or

\makebox[5mm][l]{\rm{(ii)}}~$\mu(A_1) = \mu(A_2)$
and $A_1$ is a new atom and
$A_2$ is an old atom. 

\noindent
By abuse of notation, 
given any two atoms $A_1,A_2\in \mathcal B_{\omega}$, 
we write $A_1 \succ \neg A_2$ if  $\sigma(A_1)\! >\! \sigma (A_2)$ 
$($that is, $\sigma(A_1)\! \geq\! \sigma (A_2))$.
\end{definition}

We have that $\succ$ is a well-founded  ordering on $B_{\omega}$.

\begin{lemma}\label{lemma:wfo} Let $T$ be a $\mu$-consistent proof tree 
for an atom $A$
and a program $P$. Then, for every atom \(B\) and literal \( L \)
which is a child of \( B \) in \( T \), we have $B\succ L$.
\end{lemma}

\begin{proof}
Let $L_1,\ldots,L_r$ be the children of $B$ in $T$, for some \(
\gamma \in P \) and valuation \( v \) such that \( v(\gamma ) \) is
\( B\leftarrow L_{1}\wedge \ldots \wedge L_{r} \), and let $L$ be
the literal $L_i$. If $L_i$ is the negated atom $\neg A_i$ then,
since $P$ is locally stratified w.r.t.~$\sigma$, we have $\sigma(B)
> \sigma (A_i)$ and $B\succ L_i$. Let us now consider the case where
$L_i$ is positive. 

If the predicate of $B$ is old then, by
$\mu$-consistency of~$T$, $\mu(B)>_{\mathit{lex}}\mu(L_1)
\oplus\cdots\oplus\mu(L_r)$. By
Lemma~\ref{lemma:properties_of_oplus}~(ii.1), $\mu(L_1)
\oplus\cdots\oplus\mu(L_r)\geq_{\mathit{lex}}\mu(L_i)$ and, thus,
$\mu(B)>_{\mathit{lex}}\mu(L_i)$. By definition of~$\succ$, we have 
that $B\succ L_i$. 

If the predicate of $B$ is new and $\gamma$ is $\sigma$-max
derived then, by $\mu$-consistency of $T$,
$\mu(B)>_{\mathit{lex}}\mu(L_i)$ and, thus,
$B\succ L_i$. 

Finally, if the predicate of $B$ is new and
 $\gamma$ is {\emph{not}} $\sigma$-max derived then
$\gamma$ is a descendant of a clause that has not been derived 
by folding and, thus, the predicate of $L_i$ is old. By
$\mu$-consistency, $\mu(B)\geq_{\mathit{lex}}\mu(L_i)$ and, since
the predicate of $B$ is new and the one of $L_i$ is old, we have
$B\succ L_i$.
\end{proof}

\begin{lemma} \label{lem:muconsistent-prooftree}
Consider the locally stratified $\omega$-program $P_d$ of the admissible transformation 
sequence $P_0,\ldots,P_d,\ldots,P_n$,
where\/{\rm{:}} {\rm{(1)}}~\(
P_{0},\ldots ,P_{d} \) is constructed by using rule {\rm{(R1)}}, and
 {\rm{(2)}}~$P_{d},\ldots,$ $P_{n}$ is constructed by 
 applying rules {\rm{(R2)}}--{\rm{(R7)}}. 
If there exists a proof tree
for $A$ and $P_d$ then there exists a \mbox{$\mu$-consistent} proof tree for $A$ and $P_d$.
\end{lemma}

\begin{proof}
Let us consider a proof tree $T$ for $A$ and $P_d$ such that 

\noindent
{\it{min}}-{\it{weight}}$(A)=$ {\it{weight}}$(\sigma(A),T)$.
We want to show that $T$ is $\mu$-consistent. That tree $T$ can be 
depicted as in Figure~\ref{fig:lemma-mu-consistency-in-Pd}. That
tree has been constructed by using at the top the clause $\gamma$ and a 
valuation $v$ such that $v(\gamma)$ is of the form
$A\leftarrow L_1\wedge\ldots\wedge L_n$.

\begin{figure}[ht]
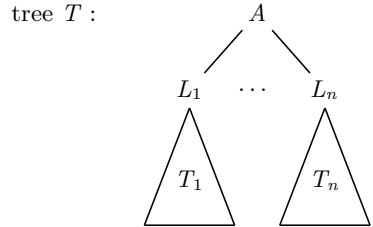

\begin{center}
\VCDraw{%
\begin{VCPicture}{(3,0)(5,5.3)} 
\FixStateDiameter{10mm}
\ChgStateLineWidth{1.3}
\ChgStateLineColor{white}
\SetEdgeArrowWidth{0pt}
% states 
\State[{\mathrm{tree}}~T:]{(-2.,4.7)}{T}
\State[A]{(2.5,4.7)}{A}
\State[L_1]{(1,3)}{L1}
\State[L_n]{(4,3)}{Ln}
\State[T_1]{(1,1)}{T1}
\State[T_n]{(4,1)}{Tn}
\State[\ldots]{(2.4,3)}{ldots}

\Point{(1,2.6)}{T1top}
\Point{(0,0)}{T11}
\Point{(2,0)}{T12}

\Point{(4,2.6)}{Tntop}
\Point{(3,0)}{Tn1}
\Point{(5,0)}{Tn2}
\EdgeL{A}{L1}{}
\EdgeL{A}{Ln}{}
\EdgeL{T1top}{T11}{}
\EdgeL{T1top}{T12}{}
\EdgeL{T11}{T12}{}

\EdgeL{Tntop}{Tn1}{}
\EdgeL{Tntop}{Tn2}{}
\EdgeL{Tn1}{Tn2}{}
\end{VCPicture}}
\end{center}
\vspace{-3mm}
\caption{A proof tree $T$ for $A$ and $P_d$ such that 
{\it{min}}-{\it{weight}}$(A)=$ {\it{weight}}$(\sigma(A),T)$. There is a valuation $v$ and  a clause $\gamma\in P_d$ such that $v(\gamma)$ is of the form:
$A\leftarrow L_1\wedge \ldots \wedge L_n$.
For $i=1,\ldots,n$, $T_i$ is a $\mu$-consistent proof tree for $L_i$ and $P_d$.
\label{fig:lemma-mu-consistency-in-Pd}}
\end{figure}

By induction on {\it{size}}$(T)$, we may assume that $T_1,\ldots,T_n$ are
$\mu$-consistent proof trees. Since $\gamma$ is locally stratified, we also 
have that
for $i\!=\!1,\ldots,n$, $\sigma(A)\geq\sigma(L_i)$.

\noindent
(Recall that if $L_i$, for some $i\in \{1,\ldots,
n\}$, is a negated atom, then $T_i$ consists of the single node $L_i$ and
$T_i$ is $\mu$-consistent.)

In order to prove the lemma we have to show the following two points:

\noindent  (P1)~if $A$ is a new atom then $\mu(A)\!\geq_{\mathit {lex}}\!
 \mu(L_1)\oplus\wedge\ldots\wedge\oplus  \mu(L_n)$, and 

\noindent 
 (P2)~if $A$ is an old atom then $\mu(A)\!>_{\mathit {lex}}\!
 \mu(L_1)\oplus\wedge\ldots\wedge\oplus  \mu(L_n)$.
 
(Note that $A\leftarrow L_1\wedge\ldots\wedge L_n$ is not an instance of a
$\sigma$-max derived clause belonging to $P_d$, because no such a clause exists in $P_d$
and, thus, if Points~(P1) and (P2) hold then the proof tree $T$ is $\mu$-consistent.)

Now, let us consider the following two cases: (1)~$A$ is a new atom, and 
(2)~$A$ is an old atom.

\medskip
{\it{Case}}~(1): $A$ is a new atom. We have that:

\smallskip

\noindent 
$\mu(A)=\langle\sigma(A),$ {\it min-weight}$(A)\!-\!1\rangle =
\langle\sigma(A),\sum_{i=1,\ldots,n \ \wedge \ \sigma(L_i)=\sigma(A)}$ {\it min-weight}$(L_i)\rangle$. 
\smallskip

Now, we consider two subcases.

\noindent
{\it{Case}}~(1.1): for $i=1,\ldots,n$,  $\sigma(A)\!>\!\sigma(L_i)$.
In this case we have that: 

\smallskip

$\langle\sigma(A),\sum_{i=1,\ldots,n \ \wedge \ \sigma(L_i)=\sigma(A)}$ {\it min-weight}$(L_i)\rangle = $  

\smallskip

$=\langle\sigma(A),0\rangle>_{\mathit{lex}} \mu(L_1)\oplus\ldots\oplus \mu(L_n)$.

\smallskip
\noindent
This last  inequality holds because $\pi_1(\mu(L_1)\oplus\ldots\oplus \mu(L_n))=$

\noindent
$=\max\{\sigma(L_i)\mid i=1,\ldots,n\}<\sigma(A)$, because
for $i=1,\ldots,n$,  $\sigma(A)\!>\!\sigma(L_i)$.

\noindent
{\it{Case}}~(1.2): there exists $i\in\{1,\ldots,n\}$ such that $\sigma(A)\!=\!\sigma(L_i)$. In this case we have that:

\smallskip

$\langle\sigma(A),\sum_{i=1,\ldots,n \ \wedge \ \sigma(L_i)=\sigma(A)}$ {\it min-weight}$(L_i)\rangle =  \mu(L_1)\oplus\ldots\oplus \mu(L_n)$,
\smallskip

\noindent
because $\mu(L_p)\oplus\mu(L_q)\!=\!\mu(L_p)$, whenever 
$\pi_1(\mu(L_p))\!>\!\pi_1(\mu(L_q))$.
This concludes the proof of Case~(1) and of Point~(P1).

\medskip
{\it{Case}}~(2): $A$ is an old atom. We have that:

\smallskip

\noindent 
$\mu(A)=\langle\sigma(A),$ {\it min-weight}$(A)\rangle =$ 

\smallskip
= \{the proof tree $T$  for $A$ and $P_d$ is such that

\makebox[6mm][]{}{\it{min-weight}}$(A) =$ {\it{weight}}$(\sigma(A), T)$\} =

\smallskip 
$=\langle\sigma(A),\ \big(\sum_{i=1,\ldots,n \ \wedge \ \sigma(L_i)=\sigma(A)}$ {\it min-weight}$(L_i)\big)+1\rangle$.
\smallskip

\noindent Let $M$ be the subset of $\{1,\ldots,n\}$ such that for all $j\in M$,
$\sigma(L_j)\!=\!\sigma(A)$. 
We have that: 

\smallskip

$\langle\sigma(A),\ \big(\sum_{i=1,\ldots,n \ \wedge \ \sigma(L_i)=\sigma(A)}$ {\it min-weight}$(L_i)\big)+1\rangle=$

\smallskip

$=\langle\sigma(A), \ \big(\sum_{j\in M}$ {\it min-weight}$(L_j)\big)+1\rangle
>_{\mathit{lex}}\mu(L_1)\oplus\ldots\oplus \mu(L_n)$.

\smallskip
\noindent
This last inequality holds because $\sum_{j\in M}$ {\it min-weight}$(L_j)=
\pi_2(\mu(L_1)\oplus\ldots\oplus \mu(L_n)$. This concludes the proof of Case~(2), of Point~(P2), and of the lemma.
\end{proof}

\begin{proposition} \label{prop:preserv-mu-consistency}
Let $P_0$ be a locally stratified $\omega$-program, $\sigma$ be a
local stratification for $P_0$, and \( P_{0},\ldots ,P_{d},\ldots
,P_{n} \) be an admissible transformation sequence where\/{\rm{:}} {\rm{(1)}}~\(
P_{0},\ldots ,P_{d} \) is constructed by using rule {\rm{(R1)}}, and
 {\rm{(2)}}~$P_{d},\ldots,$ $P_{n}$ is constructed by applying rules {\rm{(R2)}}--{\rm{(R7)}}.
Then, for every atom \( A \in \mathcal{B}_{\omega}\), we have that,
for $k=d,\ldots,n$\/{\rm{:}}

\smallskip

\noindent \textup{({\em Soundness})} if there exists a proof tree
for \( A \) and \( P_{k} \), then there exists a proof tree for \( A
\) and \( P_d\), and
\smallskip{}

\noindent \textup{({\em Completeness})} if there exists a \(\mu
\)-consistent proof tree for \( A \) and \( P_d\), then there exists
a \mbox{\( \mu\)-con\-sist}\-ent proof tree for \( A \) and \( P_{k}\).
\end{proposition}

\begin{proof} We prove the ({\em Soundness}) and ({\em Completeness}) properties
by induction on $k$.

\noindent
Clearly they hold for $k=d$.

\noindent
Now, let us assume, by induction, that: 

\smallskip

%\noindent
%\framebox{(IndHyp) the ({\em Soundness}) and ({\em Completeness})
%properties hold for any~$k$, with $d\!\leq\! k\!<\!n$.}

\noindent
\framebox{\begin{minipage}[c][1.2\totalheight]{0.98\columnwidth}
(IndHyp) the ({\em Soundness}) and ({\em Completeness})
properties hold for~$k$, with $d\!\leq\! k\!<\!n$.
\end{minipage}}

\smallskip

We have to show that they hold
 for $k\!+\!1$. 
 
 In order to prove that the 
({\em Soundness}) and ({\em Completeness})
 properties hold for $k\!+\!1$, it is sufficient to prove that:

\smallskip

\noindent (S) for every atom \( A\in\mathcal{B}_{\omega} \), 
if there exists a proof tree for $A$ and \( P_{k+1} \)
then there exists a proof tree for \( A \) and \(
 P_{k}\), and

\noindent (C) for every atom \( A\in\mathcal{B}_{\omega} \),
if there exists a \(\mu\)-consistent proof tree for
$A$ and \( P_k \) then there exists a \(\mu\)-consistent proof tree
for \( A \) and \(  P_{k+1} \).

\smallskip

\noindent We proceed by complete induction on the ordinal \( \sigma
(A) \) associated with the atom~\( A \). The inductive hypotheses
(IS) and (IC) for (S) and (C), respectively, are as follows:
\smallskip{}

\noindent
\framebox{\begin{minipage}[c][1.2\totalheight]{0.98\columnwidth}
(IS) for every atom \( A'\in\mathcal{B}_{\omega} \) such
that \( \sigma (A')\! <\! \sigma (A) \), if there exists a proof
tree for \( A' \) and \(  P_{k+1} \) then there exists a proof tree
for \( A' \) and \(  P_{k} \),
\end{minipage}}

\medskip{}
and

\medskip{}
\noindent 
\framebox{\begin{minipage}[c][1.2\totalheight]{0.98\columnwidth}
(IC) for every atom \( A'\in\mathcal{B}_{\omega} \) such
that \( \sigma (A')\! <\! \sigma (A) \), if there exists a \mbox{\(
\mu\)-consistent} proof tree for \( A' \) and \(  P_{k}\) then there
exists a \( \mu \)-consistent proof tree for \( A' \) and \(
P_{k+1}\).
\end{minipage}}

%(IC) for every atom \( A'\in\mathcal{B}_{\omega} \) such
%that \( \sigma (A')\! <\! \sigma (A) \), if there exists a \mbox{\(
%\mu\)-consistent} proof tree for \( A' \) and \(  P_{k}\) then there
%exists a \( \mu \)-consistent proof tree for \( A' \) and \(
%P_{k+1}\).\smallskip{}

\medskip{}
By the inductive hypotheses  (IS) and (IC), we have that:
\medskip{}

\noindent
\framebox{\begin{minipage}[c][1.2\totalheight]{0.98\columnwidth}
(ISC) for every atom \( A'\in\mathcal{B}_{\omega} \) such
that \( \sigma (A')\! <\! \sigma (A) \) (and thus, $A\succ A'$), 
there exists a proof tree $T'$
for \( A' \) and \(  P_{k} \) iff there exists a proof tree $U'$ for
\( A' \) and \(  P_{k+1} \).
\end{minipage}}

\smallskip{}
%\noindent (ISC) for every atom \( A'\in\mathcal{B}_{\omega} \) such
%that \( \sigma (A')\! <\! \sigma (A) \) (and thus, $A\succ A'$), there 
%exists a proof tree
%for \( A' \) and \(  P_{k} \) iff there exists a proof tree for \(
%A' \) and \(  P_{k+1} \).

\medskip{}

\noindent \framebox{\emph{Proof of} (S).} Given a proof tree \( U \) for \( A
\) and \(  P_{k+1} \) we have to prove that there exists a proof
tree \( T \) for \( A \) and \(  P_{k} \). The proof is by complete
induction on \emph{\( \mathit{size}(U) \)}. The inductive 
hypothesis~is:

\medskip{}

\noindent
\framebox{\begin{minipage}[c][1.2\totalheight]{0.98\columnwidth}
\noindent (Isize) for every atom \(A'\in\mathcal{B}_{\omega}\),
for every proof tree \( U' \) for $A'$ and \(  P_{k+1} \), if \(
\mathit{size}(U')<\mathit{size}(U) \) then there exists a proof tree
\( T' \) for \( A' \) and \(  P_{k} \).
\end{minipage}}

%\noindent (Isize) for every atom \(A'\in\mathcal{B}_{\omega}\),
%for every proof tree \( U' \) for $A'$ and \(  P_{k+1} \), if \(
%\mathit{size}(U')<\mathit{size}(U) \) then there exists a proof tree
%\( T' \) for \( A' \) and \(  P_{k} \).

\medskip{}

Let \( \eta  \) be a clause in \(  P_{k+1} \) and \( v \) be a
valuation. Let \( v(\eta)\) be a clause of the form 
$ A\leftarrow L_{1}\wedge \ldots \wedge L_{r} $
 used at the root of \( U \). We
proceed by considering the following cases: \emph{either} (Case
1)~\( \eta \) belongs to \(  P_{k} \) \emph{or} (Case 2)~\( \eta \)
does not belong to \(  P_{k} \) and it has been derived from a
clause in \(P_{k} \) by applying a transformation rule among R2, R3,
R4, R6, and R7. These two cases are mutually exclusive and
exhaustive because rule~R5 removes a clause.

We have that, for \( i=1,\ldots ,r \), there is a proof tree \(
T_{i} \) for $L_i$ and \(  P_{k} \). Indeed, (i)~if \( L_{i} \) is
an atom then, by induction on (Isize), there exists a
proof tree \( T_{i} \) for \( L_{i} \) and \(  P_{k} \), and (ii)~if
\( L_{i} \) is a negated atom \( \neg A_{i} \) then, by the fact
that program \(  P_{k+1} \) is locally stratified w.r.t.~\( \sigma
\) and by the inductive hypothesis (ISC), there is no proof tree for
\( A_{i} \) and \(  P_{k} \) and hence, by definition, there is a
proof tree \( T_{i} \) for $L_i$ and \(  P_{k} \).

\medskip

\noindent \emph{Case} 1. A proof tree $T$ for $A$ and $ P_k$ can be constructed by using $v(\eta)$ and the proof trees 
$T_1,\ldots,T_r$ for $L_1,\ldots,L_r$, respectively, and $P_k$.

\medskip

\noindent \emph{Case} 2.1 ($P_{k+1}$ is derived from $P_k$ by using
rule R2.) Clause \( \eta  \) is derived by instantiating a variable
$X$ in a clause $\gamma\in  P_k$. We have that $\gamma$ is a clause
of the form $\tA \leftarrow \tL_1\wedge\ldots \wedge\tL_r$ and $\eta$ is of the form $(\tA\leftarrow
\tL_1\wedge\ldots \wedge\tL_r)\{X/\llbracket s|X\rrbracket \}$ for some $s\in\Sigma$. Thus,
$v(\tA\{X/\llbracket s|X\rrbracket \})=A$ and, for $i\!=\!1,\ldots,r$,
$v(\tL_i\{X/\llbracket s|X\rrbracket \})=L_i$.

Let $v'$ be the valuation
such that $v'(X)=v(\llbracket s|X\rrbracket) $ and $v'(Y)=v(Y)$ for
every variable $Y$ different from $X$. Then $v'(\gamma) = v(\eta)$
and a proof tree $T$ for $A$ and $ P_k$ can be constructed from
$T_1,\ldots,T_r$ by using $v'(\gamma)$ at the root of $T$.

\medskip

\noindent \emph{Case} 2.2 ($P_{k+1}$ is derived from $P_k$ by using
rule R3.)  Clause \( \eta  \) is derived by unfolding a clause
$\gamma\in  P_k$ w.r.t.~a positive literal, say $\tK$, in its body using
clause $\gamma_i$. Recall that clauses $\gamma$ and $\gamma_i$ are 
assumed to have no variables in common (see rule~R3). Without loss
of generality, we may assume that: (i)~$\eta$
 is of the form $(\tA\leftarrow \tL_{1}\wedge \ldots \wedge \tL_{r})\vartheta_i$,
(ii)~$\gamma$ is of the form $\tA\leftarrow \tK \wedge \tL_{q+1}\wedge \ldots \wedge \tL_{r}$, with $0\!\leq\!q\!\leq\!r$, and
(iii)~$\gamma_i$ is of the form $\tH\leftarrow \tL_{1}\wedge \ldots \wedge \tL_{q}$, where $\vartheta_i$ is an (idempotent and 
 without identity bindings) most general unifier of $\tK$ and $\tH$.

Let $v'$ be the valuation
such that:(i)~$v'(X)=v(X\vartheta_i)$ for every variable $X$ in the domain of~$\vartheta_i $, and (ii)~$v'(Y)=v(Y)$ for every variable $Y$ not in the domain of ~$\vartheta_i$. For this choice of $v'$ we have that
   $v'(\tK)=$ \{by definition of~$v'$\} $=v(\tK\vartheta_i)=$ \{since $\tK\vartheta_i=\tH\vartheta_i$\} $=v(\tH\vartheta_i)=$ \{by definition of $v'$\} $=v'(\tH)$.
   
For instance, given $\gamma$: $p(X)\If q(X,Y) \wedge s(X,Y,W)$ and
 $\gamma_i$: $q(Z,a)\If r(Z)$, by unfolding~$\gamma$ w.r.t.~$q(X,Y)$
 using $\gamma_i$, we get a most general unifier $\vartheta_i=\{Z/X,Y/a\}$ 
  and the clause $\eta$: $p(X)\If r(X) \wedge s(X,a,W)$. 
 Thus, if $v(\eta)\!=\!p(b)\If r(b) \wedge s(b,a,c)$,
   we have $v'(X)\!=\!b$, $v'(W)\!=\!c$, $v'(Z)\!=\!b$, and
  $v'(Y)\!=\!a$.

Now, since $v'(\tK)\!=\!v'(\tH)$, given the proof trees 
$T_1,\ldots,T_r$ for $L_1,\ldots,L_r$, respectively, and $P_k$,
we can construct a proof tree $T$ for $A$ and $ P_k$ as follows.
Let $K$ denote $v'(\tK)$.
(i)~We first construct a proof tree $T_K$ for~$K$
 and $ P_k$ from $T_1,\ldots,T_q$ by using clause $v'(\gamma_i)$ at the root of $T_K$, and then,
(ii)~we construct~$T$ from $T_K,T_{q+1},\ldots,T_r$
by using clause~$v'(\gamma)$ at the root of $T$.

\medskip 

\noindent \emph{Case} 2.3 ($P_{k+1}$ is derived from $P_k$ by using
rule R4.)  Clause \( \eta  \) is derived by unfolding a clause
$\gamma\in  P_k$ w.r.t.~a negative literal, say $\neg \tK$, in its body. Recall that we have assumed that $v(\eta)$ is of the form $A
\leftarrow  L_{1} \wedge \ldots \wedge L_r$. Without loss
of generality, we may assume that: 

\noindent
(i)~there exist $m$ 
substitutions $\vartheta_1,\ldots,\vartheta_m$ and $m$ clauses $\gamma_1,\ldots,\gamma_m$ in $
P_k$ such that, for $i=1,\ldots,m$, $\vartheta_i$ is a most general
unifier of $\tK$ and $\Mathit{hd}(\gamma_i)$, $\tK\!=\!\Mathit{hd}(\gamma_i)\vartheta_i$,
and $v(\gamma_i\vartheta_i)$ is of the form
$K\leftarrow B_i$, and

\noindent
(ii)~$v(\gamma)$ is of the form $A
\leftarrow \neg K \wedge L_{m+1} \wedge \ldots \wedge L_r$, with
$0\leq m\leq r$, (note that, by Condition~(1) of rule R4, $\gamma$ is not
instantiated by the negative unfolding).
Thus, $v(\eta)=A
\leftarrow  L_{1} \wedge \ldots \wedge L_r$,  is derived 
from $A
\leftarrow \neg (B_1\vee \ldots \vee B_m) \wedge L_{m+1} \wedge \ldots \wedge L_r$
by pushing $\neg$ inside and pushing $\vee$ outside.

Now, let us assume by absurdum that there exists a proof tree $U_K$ for
$K$ and $P_{k+1}$. Then, there exists a valuation $v'$ such that the children of
the root of $U_K$ are labeled by the literals $M_1,\ldots,M_s$, where
$v'(\mathit{bd}(\gamma_i\vartheta_i))=M_1\wedge\ldots\wedge M_s$, for
some~$i$, with $1\leq i\leq m$. Since $\gamma_i$ has no existential variables, without 
loss of generality we take $v'(X)=v(X)$, for every variable $X$. 
By the definition of the negative unfolding rule, 
there exist $j\in\{1,\ldots,s\}$ and $h\in\{1,\ldots,m\}$ such that $M_j=\overline L_h$.
By hypothesis, there exists a proof tree for $L_h$ and $P_k$ and, thus, $U_K$
is not a proof tree for $K$ and $P_{k+1}$. This is a contradiction and, thus, we have that
there is no proof tree for $K$ and $P_{k+1}$. Since $\sigma(K)<\sigma(A)$, by
the inductive hypothesis (ISC), we have that there is no proof tree for
$K$ and $ P_k$. Hence, there is a proof tree $T_{\neg K}$ for
$\neg K$ and $ P_{k}$. Thus, we can construct a proof tree $T$ for $A$
and $ P_k$ from $T_{\neg K}, T_{m+1},\ldots,T_r$ by using clause~$v(\gamma)$ at the root of $T$.

\medskip 

\noindent \emph{Case} 2.4 ($P_{k+1}$ is derived from $P_k$ by using
rule  R6.) Let us assume that clause \( \eta  \) of the form 
$\tA \leftarrow \tL_1 \wedge\tL_2 \wedge\ldots \wedge \tL_r$
is derived by positive folding from a
clause $\gamma\in P_k$ of the form $\tA\leftarrow \tM'_1\wedge\ldots \wedge \tM'_s\wedge\tL_2 \wedge\ldots \wedge \tL_r$  using a clause $\delta\in\mbox{\it Defs}_k$ of the form
$\tK\leftarrow \tM_1\wedge\ldots \wedge \tM_s$. Without loss of generality, we may assume that $\tL_1=\tK\vartheta$, where $\vartheta$ is a substitution such that, for $i\!=\!1,\ldots,s$, $\tM_i\vartheta\!=\!\tM'_i$. Thus, the literal 
$L_1$ in the body of $v(\eta)$ is $v(\tK\vartheta)$.
We have that
$\delta\in P_d$ and the definition of the head predicate of $\delta$ in $ P_d$
consists of clause $\delta$ only. 

By induction on $k$, we have that the ({\em Soundness})
property holds for $k$. 
We know that there is a proof tree for $L_1$ and $P_k$. 
Hence, by Conditions~(i) and~(ii) of rule
R6, there exists a proof tree for $L_1$ and $P_d$, for some
valuation $v'$ such that $v'(\delta)$ is of the form $L_1\leftarrow
M_1\wedge\ldots \wedge M_s$ (note that if $X\!\in\! {\mathit{vars}}(\eta)$  then $v'(X)=v(X)$). 

By induction on $k$, we have that the ({\em Soundness}) and
({\em Completeness}) properties hold for $k$. Thus, there are proof
trees $U_1,\ldots,U_s$ for $M_1,\ldots, M_s$, respectively,  and 
$ P_k$. 

Finally, by
 induction on (Isize), we know that there exist the proof trees
$T_2,\ldots,T_r$ for $L_2,\ldots,L_r$, respectively, and $ P_{k}$.
As a consequence, we can construct a proof tree $T$ for $A$ and $
P_k$ from $U_1,\ldots,U_s,T_2,\ldots,T_r$ by using clause
$v(\gamma)$  at the root of $T$.

\medskip 

\noindent \emph{Case} 2.5 ($P_{k+1}$ is derived from $P_k$ by using
rule R7.) Clause \( \eta \) is derived by negative folding from a
clause $\gamma\in P_k$ using clauses $\delta_1,\ldots,\delta_m$ in
$\mbox{\it Defs}_k$. Thus, we have that: (i)~$v(\gamma)$ is of the
form $A\leftarrow N_1\wedge \ldots\wedge N_m\wedge
L_2\wedge\ldots\wedge L_r$, (ii)~for $i=1,\ldots,m$, $v(\delta_i)$
is of the form $K\leftarrow B_i$, where either $N_i$ is a positive
literal $A_i$ and $B_i$ is $\neg A_i$, or $N_i$ is a negative
literal $\neg A_i$ and $B_i$ is $A_i$, and (iii)~$v(\eta)$ is of the
form $A\leftarrow \neg K\wedge L_2\wedge\ldots\wedge L_r$. Thus, $L_1=\neg K$.

By the inductive hypothesis (ISC), there exists a proof tree for
$L_1$ and $P_k$ and, since $L_1=\neg K$, there is no
proof tree for $K$ and $ P_{k}$. By induction on $k$, we have
that the ({\em
Completeness}) holds for $k$ and, therefore, there exists no proof
tree for $K$ and $ P_d$. We have that
$\{\delta_1,\ldots,\delta_m\}\subseteq P_d$ and the clauses defining
the head predicate of $\delta_1,\ldots,\delta_m$ in $ P_{d}$ are
$\{\delta_1,\ldots,\delta_m\}$. Thus, there are no proof trees for
$B_1,\ldots, B_m$ and~$P_d$. 

By induction on $k$,  the ({\em
Soundness}) property holds for $k$ and, therefore, there are no
proof trees for $B_1,\ldots, B_m$ and $ P_k$. Thus, there are proof
trees $U_1,\ldots,U_m$ for $N_1,\ldots,N_m$, respectively,  and $ P_k$. Finally, by
induction on (Isize), we have that there are the proof trees
$T_2,\ldots,T_r$ for $L_2,\ldots,L_r$, respectively, and $ P_{k}$.
We can construct a proof tree $T$ for $A$ and $ P_k$ from
$U_1,\ldots,U_m,T_2,\ldots,T_r$ by using clause $v(\gamma)$
 at the root of $T$.

\medskip{}

\noindent {\framebox{\emph{Proof of} (C).}} Given a \(\mu \)-consistent proof
tree $T$ for \( A \) and \(  P_{k} \), we prove that there exists a
\( \mu \)-consistent proof tree $U$ for \( A \) and \(
 P_{k+1} \).

The proof is by well-founded induction on 
$\succ\; \subseteq \mathcal{B}_{\omega}\!\times\! \mathcal{B}_{\omega}$. The
inductive hypothesis is:

\medskip{}

\noindent
\framebox{\begin{minipage}[c][1.2\totalheight]{0.98\columnwidth}
(I\( \mu  \)) for every atom \(A'\in\mathcal{B}_{\omega}
\) such that \( A\succ A'\), if there exists a \( \mu \)-consistent
proof tree \( T' \) for \( A' \) and \( P_{k} \) then there exists a
\(\mu\)-consistent proof tree \( U' \) for \( A' \) and \(  P_{k+1} \).
\end{minipage}}

%\noindent (I\( \mu  \)) for every atom \(A'\in\mathcal{B}_{\omega}
%\) such that \( A\succ A'\), if there exists a \( \mu \)-consistent
%proof tree \( T' \) for \( A' \) and \( P_{k} \) then there exists a
%\(\mu\)-consistent proof tree \( U' \) for \( A' \) and \(  P_{k+1} \).

\medskip{}

Let \( \gamma  \) be a clause in \(  P_{k} \) and \( v \) be a
valuation such that \( v(\gamma) \) is the clause of the form \(
A\leftarrow L_{1}\wedge \ldots \wedge L_{r} \) used at the root of
\( T \). We consider the following cases: \emph{either} (Case 1) \(
\gamma \) belongs to \(  P_{k+1} \) \emph{or} (Case 2) \( \gamma \)
does not belong to \(  P_{k+1} \) because it has been replaced by
zero or more clauses derived by applying a transformation rule among
R2--R7.

\medskip

\noindent \emph{Case} 1. By the $\mu$-consistency of $T$ and
Lemma~\ref{lemma:wfo}, for \( i=1,\ldots ,r \),  we have $A \succ
L_i$. Hence, by the inductive hypotheses (I$\mu$) and (ISC), there
exists a $\mu$-consistent proof tree \(U_{i} \) for \( L_{i} \) and
\( P_{k+1} \). A $\mu$-consistent proof tree $U$ for $A$ and $
P_{k+1}$ is constructed by using $v(\gamma)$ at the root 
of $U$ and the proof trees
$U_1,\ldots,U_r$ for $L_1,\ldots,L_r$, respectively, and $P_{k+1}$.

\medskip

\noindent \emph{Case} 2.1 ($P_{k+1}$ is derived from $P_k$ by using
rule R2.) Suppose that by instantiating a variable $X$ of clause
$\gamma$ in $P_k$ we derive clauses $\gamma_1,\ldots,\gamma_h$ in
$P_{k+1}$. For $i=1,\ldots,h,$ $\gamma_i$ is $\gamma\{X/\llbracket
s_i|X\rrbracket \}$, with $s_i\in\Sigma$. Hence, there exist
$i\in\{1,\ldots,h\}$ and a valuation $v'$ such that
$v(\gamma)=v'(\gamma_i)$.  By the $\mu$-consistency of $T$ and
Lemma~\ref{lemma:wfo}, for \( i=1,\ldots ,r \), we have $A \succ
L_i$. Hence, by the inductive hypotheses (I$\mu$) and (ISC), 
for \( i=1,\ldots ,r \), there
exists a $\mu$-consistent proof tree \(U_{i} \) for \( L_{i} \) and
\( P_{k+1} \). A  proof tree $U$ for $A$ and $
P_{k+1}$ is constructed by using $v'(\gamma_i)$ at the root of
$U$ and the proof trees
$U_1,\ldots,U_r$ for $L_1,\ldots,L_r$,  respectively, and $P_{k+1}$.

The proof tree $U$ is $\mu$-consistent because: (i)~by (I$\mu$), 
we have that  $U_1,\ldots,U_r$ are $\mu$-consistent,
(ii)~$\gamma_i$ is 
$\sigma$-max derived iff $\gamma$ is 
$\sigma$-max derived, and (iii)~since~$T$ is $\mu$-consistent,
we have that if~$\gamma$ is not $\sigma$-max derived then 
$\mu(A)\geq_{\mathit{lex}} \mu(L_1)\oplus\ldots\oplus\mu(L_r)$
else $\mu(A)>_{\mathit{lex}} \mu(L_1)\oplus\ldots\oplus\mu(L_r)$.

\medskip 

\noindent \emph{Case} 2.2 ($P_{k+1}$ is derived from $P_k$ by using
rule R3.) Suppose that by unfolding $\gamma$ w.r.t.~an atom~$B$ in
its body we derive clauses $\eta_1,\ldots,\eta_m$ in $P_{k+1}$.
Without loss of generality, we assume that $B$ is the leftmost
literal in the body of $\gamma$. Hence, there exists a clause
$\gamma_i$ in (a variant of) $P_k$ such that: (i)~$v(\gamma_i)$ is
of the form $L_1 \leftarrow M_1\wedge \ldots \wedge M_q$,
(ii)~$v(\eta_i)$ is $A\leftarrow M_1 \wedge \ldots \wedge M_q \wedge
L_2 \wedge \ldots \wedge L_r$, and (iii)~$v(\gamma_i)$ is the clause
which is used for constructing the children of $L_1$ in $T$. By the
\mbox{$\mu$-consistency} of~$T$ and Lemma~\ref{lemma:wfo}, for \(
i=1,\ldots ,q\), we have $A \succ M_i$ and, for \( i=2,\ldots ,r \),
we have $A \succ L_i$. Hence, by the inductive hypotheses (I$\mu$)
and (ISC), for \( i=1,\ldots ,q,\) there exists a $\mu$-consistent
proof tree \(V_{i} \) for \( M_{i} \) and \( P_{k+1} \) and, for \(
i=2,\ldots ,r \), there exists a $\mu$-consistent proof tree \(U_{i}
\) for \( L_{i} \) and \( P_{k+1} \). A proof tree $U$ for $A$ and $ P_{k+1}$ is
constructed by using $v(\eta_i)$ at the root of $U$ and the proof trees
$V_1,\ldots,V_q,U_2,\ldots,U_r$ for $M_1, \ldots,M_q,L_2,\ldots,
L_r$,  respectively, and $P_{k+1}$.

It remains to show that the proof tree $U$ is $\mu$-consistent.
There are two cases:~(a) and~(b).

\noindent
{\it{Case}}~(a): in this first case we assume that
 $A$ is new {\it{and}} $\eta_i$ is not $\sigma$-max
derived. 
%
%Since $\eta_i$ is derived from $\gamma$ by 
%unfolding a $\sigma$-maximal atom, we get that $\gamma$ is
%not $\sigma$-max derived. Since the transformation sequence is admissible,
%we get that $L_1$ is an old atom. 

\noindent
Since $T$ is $\mu$-consistent
we get $\mu(A)\geq_{\mathit{lex}} \mu(L_1) \oplus \mu(L_2) \oplus 
\ldots  \oplus \mu(L_r)$ and
 $\mu(L_1)\geq_{\mathit{lex}} \mu(M_1) \oplus 
\ldots  \oplus \mu(M_q)$. By
Lemma~\ref{lemma:properties_of_oplus}~(ii.2), we get
$\mu(A)\geq_{\mathit{lex}} \mu(M_1) \oplus 
\ldots  \oplus \mu(M_q) \oplus \mu(L_2) \oplus 
\ldots  \oplus \mu(L_r)$.
%By $\mu$-consistency of $T$, we have $\mu(A)>_{\mathit
%lex}\mu(L_1)\oplus\cdots\oplus\mu(L_r)$ and $\mu(L_1)\geq_{\mathit
%lex}\mu(M_1)\oplus\cdots\oplus\mu(M_q)$. By
%Lemma~\ref{lemma:properties_of_oplus}~(ii.2), we get
%$\mu(A)>_{\mathit lex}
%\mu(M_1)\oplus\cdots\oplus\mu(M_q)\oplus\mu(L_2)\oplus\cdots\oplus\mu(L_r)$,
%which entails that $U$ is \mbox{$\mu$-consistent.}

\noindent 
{\it{Case}}~(b): in this second case, we assume that $A$ is
old {\it{or}} $\eta_i$ is $\sigma$-max derived. We have two subcases (b.1)
and (b.2).

\noindent
{\emph{Subcase}}~(b.1): $A$ is
old. Since $T$ is $\mu$-consistent, we get that
$\mu(A)>_{\mathit{lex}}\mu(L_1)\wedge \ldots \wedge \mu(L_r)$
and $\mu(L_1)\geq_{\mathit{lex}}\mu(M_1)\wedge \ldots \wedge \mu(M_q)$.
By Lemma~\ref{lemma:properties_of_oplus}~(ii.2) we get 
$\mu(A)>_{\mathit{lex}}\mu(M_1)\wedge \ldots \wedge \mu(M_q)$.

{\emph{Subcase}}~(b.2): 
$\eta_i$ is $\sigma$-max derived. We may assume that $A$ is new, 
because in Subcase~(b.1) we have considered that $A$ is old. 
Now we consider two subcases of this Subcase~(b.2).  

\noindent
Subcase~(b.2.1):  $\eta_i$ is $\sigma$-max derived, $A$ is new, and $\gamma$ is $\sigma$-max derived, and 

\noindent
Subcase~(b.2.2):
$\eta_i$ is $\sigma$-max derived, $A$ is new,  and $\gamma$ is not $\sigma$-max derived.

\noindent
{\it{Subcase}}~(b.2.1).~Since $T$ is $\mu$-consistent
we get $\mu(A)>_{\mathit{lex}} \mu(L_1) \oplus \mu(L_2) \oplus 
\ldots  \oplus \mu(L_r)$ and
 $\mu(L_1)\geq_{\mathit{lex}} \mu(M_1) \oplus 
\ldots  \oplus \mu(M_q)$. By
Lemma~\ref{lemma:properties_of_oplus}~(ii.2), we get
$\mu(A)>_{\mathit{lex}} \mu(M_1) \oplus 
\ldots  \oplus \mu(M_q) \oplus \mu(L_2) \oplus 
\ldots  \oplus \mu(L_r)$.

\noindent
{\it{Subcase}}~(b.2.2). Since $T$ is $\mu$-consistent and $L_1$ is old, we get:
$(\dagger 1)$~$\mu(L_1)>_{\mathit{lex}}\mu(M_1)\oplus\ldots \oplus \mu(M_q)$, and  $(\dagger 2)$~$\pi_2(\mu(L_1))>0$.
Since $\eta_i$ is $\sigma$-maximal
derived, we have that, for $i\!=\!2,\ldots,r$, 
$\sigma(L_1)\!\geq\!\sigma(L_j)$. Thus,
$(\dagger 3)$~$\sigma(L_1)\geq\pi_1(\mu(L_2)\oplus\ldots \oplus \mu(L_r))$. From $(\dagger 1)$, $(\dagger 2)$, and $(\dagger 3)$, 
by Lemma~\ref{lemma:properties_of_oplus}~(ii.3), we get:
$(\dagger 4)$
$\mu(L_1) \oplus \mu(L_2)\oplus\ldots \oplus \mu(L_r)>_{\mathit{lex}}\mu(M_1)\oplus\ldots \oplus \mu(M_q)
\oplus \mu(L_2)\oplus\ldots \oplus \mu(L_r)$. Since
 $T$ is $\mu$-consistent, we have that 
 $\mu(A)\geq_{\mathit{lex}}\mu(L_1)\oplus\ldots \oplus \mu(L_r)$, and
 by $(\dagger 4)$ we get:  
 $\mu(A)>_{\mathit{lex}} \mu(M_1)\oplus\ldots \oplus \mu(M_q)
\oplus \mu(L_2)\oplus\ldots \oplus \mu(L_r)$,
 as desired.

This concludes the proof that $U$ is a $\mu$-consistent
proof tree.

\medskip 

\noindent \emph{Case} 2.3 ($P_{k+1}$ is derived from $P_k$ by using
rule R4.) Suppose that we unfold $\gamma$ w.r.t.~a negated atom in
its body and we derive clauses $\eta_1,\ldots,\eta_s$ in $P_{k+1}$.
Without loss of generality, we assume that we unfold $\gamma$
w.r.t.~the leftmost literal in its body. Let $\gamma_1, \ldots,
\gamma_m$  be all clauses in (a variant of) $P_k$ whose heads are
unifiable with the leftmost literal in the body of $\gamma$. We may
assume that, for $i=1,\ldots,m,$ $v(\gamma_i)$ is of the form $A_1
\If B_i$, where $L_1 = \neg A_1$ and $B_i$ is a conjunction of
literals. Since there is no proof tree for $A_1$ and $P_k$, for
$i=1,\ldots,m,$ there exists a literal $R_i$ in $B_i$ such that
there is no proof tree for $R_i$ and $P_k$. By definition,
there is a proof tree for $\overline R_i$ and $P_k$. 
Moreover, (i)~$A
\succ \neg A_1$  because by hypothesis the proof tree $T$ is
$\mu$-consistent, and (ii)~$\sigma(\neg A_1) \geq \sigma(\overline R_i)$,
because $P_k$ is locally stratified w.r.t.~$\sigma$. 

Now we have two cases: (i)~$R_i$ is an atom, and (ii)~$R_i$ is a 
negated atom, say $\neg C_i$. 
In Case~(i) we have that $\sigma(A)>\sigma(A_1)\geq\sigma(R_i)$ and,
thus, $A\succ \overline R_i$. In Case~(ii) we have that 
$\sigma(A)>\sigma(A_1)\geq\sigma(\neg C_i)$  and, thus,
$\sigma(A)>\sigma(C_i)=\sigma(\overline R_i)$ and 
$\mu(A)>\mu(\overline R_i)$. Hence, 
$A \succ\overline R_i$.
Thus, in both cases $A \succ\overline R_i$.

Since $A \succ \overline R_i$,
 by the inductive hypotheses (I$\mu$) and (ISC), we have that,
for \( i=1,\ldots ,m,\) there exists a $\mu$-consistent proof tree
\(V_{i} \) for $\overline R_i$ and \( P_{k+1} \). By the
$\mu$-consistency of $T$, for \( i=2,\ldots ,r \), there exists a
$\mu$-consistent proof tree \(U_i \) for \(L_{i} \) and \( P_{k+1}
\). By the definition of rule R4, there exists a clause $\eta_p$
among the clauses $\eta_1,\ldots,\eta_s$ derived from $\gamma$, such
that $v(\eta_p)$ is of the form $A \If \overline R_1 \wedge \ldots
\wedge \overline R_m \wedge L_2 \wedge \ldots \wedge L_r$. 
(To see this, recall that by pushing $\neg$ inside and $\vee$ outside, from 
$\neg((C_1\wedge C_2)\vee (D_1\wedge D_2))$ we get 
$(\overline C_1\wedge \overline D_1)\vee (\overline C_1\wedge 
\overline D_2)\vee (\overline C_2\wedge \overline D_1)\vee 
(\overline C_2\wedge \overline D_2)$.)

A proof tree $U$ for $A$ and $ P_{k+1}$ is constructed by using $v(\eta_p)$ 
at the root of $U$ and
the proof trees $V_1,\ldots,V_m,U_2,\ldots,U_r$ for $\overline R_1, \ldots,
\overline R_m,L_2,\ldots,L_r$,  respectively, and $P_{k+1}$.

In order to show that $U$ is $\mu$-consistent
we need to consider two cases. 
In the first case, we assume that $A$ is old or
$\eta$ is $\sigma$-max derived. Thus, in this case, also $\gamma$ is $\sigma$-max derived.
By $\mu$-consistency of $T$, we have $\mu(A)>_{\mathit
lex}\mu(L_1)\oplus\cdots\oplus\mu(L_r)$. By local stratification of~$P_k$ and by
Lemma~\ref{lemma:properties_of_mu}, $\mu(L_1) \geq_{\mathit
lex}\mu(\overline{R}_1)\oplus\cdots\oplus\mu(\overline{R}_m)$. Therefore, by
Lemma~\ref{lemma:properties_of_oplus}~(ii.2), $\mu(A)>_{\mathit
lex}\mu(\overline{R}_1)\oplus\cdots\oplus\mu(\overline{R}_m
)\oplus\mu(L_2)\oplus\cdots\oplus\mu(L_r)$ and $U$ is $\mu$-consistent.

In the second case, $A$ is new and $\eta$ is not $\sigma$-max derived. As a
consequence, also $\gamma$ is not $\sigma$-max derived. By $\mu$-consistency of
$T$ we have $\mu(A)\geq_{\mathit lex}\mu(L_1)\oplus\cdots\oplus\mu(L_r)$.
By local stratification of~$P_k$ and  by Lemma~\ref{lemma:properties_of_mu}, $\mu(L_1) \geq_{\mathit
lex}\mu(\overline{R}_1)\oplus\cdots\oplus\mu(\overline{R}_m)$ and, by
Lemma~\ref{lemma:properties_of_oplus}~(ii.2), $\mu(A)\geq_{\mathit
lex}\mu(\overline{R}_1)\oplus\cdots\oplus\mu(\overline{R}_m
)\oplus\mu(L_2)\oplus\cdots\oplus\mu(L_r)$. Therefore, $U$ is $\mu$-consistent.

\medskip 

\noindent \emph{Case} 2.4 ($P_{k+1}$ is derived from $P_k$ by using
rule R5.) Suppose that the clause $\gamma$ is removed from $P_k$  by
subsumption.
%and $v(\gamma)$ is of the form $A\leftarrow L_1\wedge
%L_2\wedge \dots \wedge L_r$. Without loss of generality, we may
%assume that there exists a clause $\gamma_1$ in
%$P_k-\{\gamma\}$ and a valuation $v'$ such that $v'(\gamma_1)$ is of
%the form $A \leftarrow L_2\wedge \dots \wedge L_r$.
%Clause $\gamma_1$ belongs to $P_{k+1}$
%and, therefore, a proof tree $U$ for $A$ and $P_{k+1}$ can be
%constructed by using $v'(\gamma_1)$ at the root of $U$. The proof tree $U$
%is $\mu$-consistent because it is constructed from the $\mu$-consistent
%proof tree $T$ by deleting the subtree rooted in the node labelled by 
%$L_1$.
Hence, there exists a clause $\gamma_1$ in
$P_k-\{\gamma\}$ and a valuation $v'$ such that $v'(\gamma_1)$ is of
the form $A \leftarrow$. The clause $\gamma_1$ belongs to $P_{k+1}$
and, therefore, a proof tree $U$ for $A$ and $P_{k+1}$ can be
constructed by using $v'(\gamma_1)$ at the root of $U$. The proof tree $U$
consists of the root $A$ with the single child $\true$. 
Now we prove that the proof tree $U$ is $\mu$-consistent, that is, 
$\mu(A)>_{\mathit lex}\mu(\mathit{true})$. We have to prove that $\mu(A)>_{\mathit lex}\langle0,0\rangle$. We have the following three cases: (a), (b.1), and (b.2).

\noindent 
{\it{Case}}~(a). $A$ is an old atom. In this case we have that 
$\mu(A)>_{\mathit{lex}}\langle0,0\rangle$,
because, as stated in Remark~\ref{rem:measures}, for any old atom $B$, we have that
$\mathit{min}$-$\mathit{weight}(B)\!>\!0$.

%the proof tree for $A$ and $P_d$ has at least the root node $A$.

\noindent 
{\it{Case}}~(b). $A$ is a new atom. Since $A$ is new, there is a valuation $v'$ and
  a clause $\delta$  in $P_d$ such that
 $v'(\delta)$ is of the form $A \leftarrow G$, for some goal $G$.
 Now, let us consider the following two subcases.
 
\noindent 
{\it{Case}}~(b.1)~$G$ is of the form: $G_L\wedge B\wedge G_R$ and $B$ is an old atom. By (1)~the hypothesis that $T$ is a $\mu$-consistent proof tree for $A$ in $P_k$, 
(2)~the ({\it Soundness}) property, and (3)~Lemma~\ref{lem:muconsistent-prooftree}, 
we have that there 
exists a $\mu$-consistent proof tree $T_d$ for $A$ and $P_d$ where $B$ is a child of $A$.
By $\mu$-consistency of $T_d$, we have that $\mu(A)\geq_{\mathit{lex}}\mu(B)$.
Since 
$\mu(B)=_{\mathit{def}}\!\langle\sigma(B),{\mathit{min}}$-${\mathit{weight}}(B)\rangle$ and, since $B$ is an old atom, by
Remark~\ref{rem:measures}, we have that
${\mathit{min}}$-${\mathit{weight}}(B)\!>\!0$. Thus, we get that 
 $\mu(A)\!>_{\mathit lex}\!\langle0,0\rangle$.
 
\noindent 
{\it{Case}}~(b.2)~$G$ is of the form: $G_L\wedge \neg B\wedge G_R$ and $B$ is an old atom. Since $\delta$ is locally stratified, $\sigma(A)\!>\!\sigma(B)$ and, thus, $\sigma(A)\!>\!0$. Hence, 
$\mu(A)=_{\mathit{def}}\langle\sigma(A),{\mathit{min}}$-${\mathit{weight}}(A)\!-\!1\rangle >_{\mathit lex}\!\langle0,0\rangle$.

%
%In cases~(a) and (b), we have that $\mu(A)\! >_{\mathit{lex}}\!\mu(\true)$, because $\mu(\true) \!=\! \langle0,0\rangle$. 
%\noindent 
%{\it{Case}}~(c). $A$ is a new atom and the clause $\gamma_1$ is not
% $\sigma$-max derived. In this case we have to show that 
% $\mu(A)\geq_{\mathit{lex}}\langle0,0\rangle$. This trivially holds.
 This concludes 
the proof tree $U$ is $\mu$-consistent.

\medskip 

\noindent \emph{Case} 2.5 ($P_{k+1}$ is derived from $P_k$ by using
rule R6.) Let us assume that clause \( \eta  \) of the form 
$\tA \leftarrow \tK\vartheta \wedge\tL_{q+1} \wedge\ldots \wedge \tL_r$
is derived by positive folding from a
clause $\gamma\in P_k$ of the form $\tA\leftarrow \tL_1\wedge\ldots 
\wedge \tL_q\wedge\tL_{q+1} \wedge\ldots \wedge \tL_r$  using a clause 
$\delta\in\mbox{\it Defs}_k$ of the form
$\tK\leftarrow \tL'_1\wedge\ldots \wedge \tL'_q$ and
where
$\vartheta$ is a substitution such that, for $i\!=\!1,\ldots,q$, 
$\tL'_i\vartheta\!=\!\tL_i$. 
We have that
$\delta\in P_d$ and the definition of the head predicate of $\delta$ in $ P_d$
consists of clause $\delta$ only. 

Thus, there is a valuation $v$ such that $v(\tA)=A$ and
in the proof tree $T$ for $A$ and $P_k$ the children of $A$ are
the nodes $L_1,\ldots,L_q,L_{q+1},\ldots,L_r$ such that for $i=1,\ldots,q$, $L_i=v(\tL'_i)$
and for $i=q+1,\ldots,r$, $L_i=v(\tL_i)$. By the induction hypothesis (IndHyp) 
there exist proof trees for $v'(\tL'_1),\ldots, v'(\tL'_q)$ and $P_k$,  for some
valuation $v'$ such that, for $i\!=\!1,\ldots,q$, $v'(\tL'_i\vartheta)=v(\tL_i)$. 
Let $K$ be $v'(\tK\vartheta)$. 

Since $\delta\in P_d$ and $M(P_d)\models \delta$, by
Theorem~\ref{th:proof_trees_perfect_model} and
Definition~\ref{def:mu-consistency}, there is a $\mu$-consistent
proof tree for $K$ and $P_d$. 
By induction hypothesis, the ({\em
Completeness}) property holds for $k$ and, thus, we have that
there exists a $\mu$-consistent proof tree for $K$ and $P_k$. By the
hypothesis that the transformation sequence
$P_0,\ldots,P_d,\ldots,P_n$ is admissible and by Condition~(2) of
Definition~\ref{def:adm-transformation}, either $A$
 is old or $\gamma$ is $\sigma$-max derived. Thus, by the
$\mu$-consistency of the proof tree $T$, we have that $\mu(A)>_{\mathit lex}
\mu(L_1)\oplus\cdots\oplus\mu(L_q)\oplus\mu(L_{q+1})
\oplus\cdots\oplus\mu(L_r)$. 

Since $\delta$ is a clause in $\mbox{\it Defs}_k$, by 
Lemma~\ref{lemma:mu-for-a-definition}  we have that $\mu(K) =
\mu(L_1)\oplus\cdots\oplus\mu(L_q)$ 
and, thus,
$\mu(A)>_{\mathit lex} \mu(K)\oplus\mu(L_{q+1})
\oplus\cdots\oplus\mu(L_r)$. 

Moreover, by
Lemma~\ref{lemma:properties_of_oplus}~(ii.5), $\mu(A)>_{\mathit lex}
\mu(K)$. Thus, $A\succ K$ and, by the inductive hypothesis~(I$\mu$), 
there exists a $\mu$-consistent proof tree $U_{K}$ for $K$
and $P_{k+1}$. By the $\mu$-consistency of~$T$ and
Lemma~\ref{lemma:wfo}, for \( i=q+1,\ldots,r \), we have $A \succ
L_i$. Hence, by the inductive hypotheses~(I$\mu$) and (ISC), 
for $i\!=q+1,\ldots,r$, there
exists a $\mu$-consistent proof tree~\(U_{i} \) for \( L_{i} \) and
\( P_{k+1} \). A proof tree $U$ for $A$ and $ P_{k+1}$ is
constructed by using $v'(\eta)$  at the root of~$U$ and the proof trees
$U_{K},U_{q+1},\ldots,U_r$ for $K,L_{q+1},\ldots,L_r$, respectively, and 
$P_{k+1}$. The proof tree~$U$ is
$\mu$-consistent because, as we have shown above,
$\mu(A)>_{\mathit lex}
\mu(K)\oplus\mu(L_{q+1}) \oplus\cdots\oplus\mu(L_r)$.

\medskip 

\noindent \emph{Case} 2.6 ($P_{k+1}$ is derived from $P_k$ by using
rule R7.) Suppose that we fold $\gamma$ using clauses
$\delta_1,\ldots,\delta_q$, belonging to (a variant of)
$\Mathit{Defs}_k$, and we derive a clause $\eta$ in $P_{k+1}$.
Without loss of generality, by the definition of rule R7 and the
commutativity of $\wedge$, we may assume that (i)~$v(\gamma)$ is of
the form $A\If L_1\wedge \ldots \wedge L_q \wedge L_{q+1}\wedge
\ldots \wedge L_r$, (ii)~for $i=1,\ldots,q$, $v(\delta_i)$~is of the
form $K\If M_i$, where $M_i=A_i$, if $L_i=\neg A_i$, and $M_i=\neg
A_i$, if $L_i=A_i$, and (iii)~$v(\eta)$ is of the form $A\If \neg K
\wedge L_{q+1}\wedge \ldots \wedge L_r$. By the inductive
hypothesis, the ({\em Soundness}) and ({\em Completeness})
properties hold for $k$ and, therefore, for $i=1,\ldots,q,$ there is
no proof tree for $M_i$ and $P_d$. Since $M(P_d)\models K \Iff M_1
\vee \ldots \vee M_q$, there is no proof tree for $K$ and $P_d$. By
the inductive hypothesis, the ({\em Soundness}) property holds for
$k$ and, thus, we have that there is no proof tree for $K$ and~$P_k$. 
By the hypothesis that the transformation sequence
$P_0,\ldots,P_d,\ldots,P_n$ is admissible and by Condition (3) of
Definition~\ref{def:adm-transformation}, $\sigma(A)> \sigma(K)$.
Hence, by the inductive hypothesis (IS), there is no proof tree for
$K$ and $P_{k+1}$, that is, there is a proof tree $U_{\neg K}$ for
$\neg K$ and $P_{k+1}$. By the $\mu$-consistency of $T$ and
Lemma~\ref{lemma:wfo}, for \( i=q\!+\!1,\ldots ,r \), we have $A \succ
L_i$. Hence, by the inductive hypotheses (I$\mu$) and (ISC), there
exists a $\mu$-consistent proof tree \(U_{i} \) for \( L_{i} \) and
\( P_{k+1} \). A proof tree $U$ for $A$ and $ P_{k+1}$ is
constructed by using $v(\eta)$  at the root of $U$ and the proof trees $U_{\neg
K},U_{q+1},\ldots,U_r$ for $\neg K,L_{q+1},\ldots,L_r$, respectively,
and $P_{k+1}$.

In order to show that $U$ is $\mu$-consistent we need to consider two cases. 

In
the first case, we assume that $A$ is old or $\gamma$ is $\sigma$-max derived. Thus, in this case, also $\eta$ is $\sigma$-max derived. By $\mu$-consistency of $T$, we
have $\mu(A)>_{\mathit lex}\mu(L_1)\oplus\cdots\oplus\mu(L_q)\oplus
\mu(L_{q+1})\oplus\cdots \oplus\mu(L_r)$. By
Lemma~\ref{lemma:properties_of_oplus}~(ii.5), we have that $\mu(A)>_{\mathit
lex}\mu(L_{q+1})\oplus\cdots\oplus\mu(L_r)$. Since the transformation sequence $P_0,\ldots,P_n$ is admissible, clause $\eta$ is locally
stratified and, thus, $\sigma(A)> \sigma(K)$. Hence,
$\pi_1(\mu(A))=$ \{by definition of $\mu\}=\sigma(A)> \sigma(K)=$
 \{by definition of $\mu\}= \pi_1(\mu(\neg K))$. Therefore, by 
Lemma~\ref{lemma:properties_of_oplus}~(ii.4), we have that: $\mu(A)>_{\mathit lex}\mu(\neg K)\oplus\mu(L_{q+1})\oplus\cdots\oplus\mu(L_r)$. Thus, $U$ is 
\mbox{$\mu$-consistent.}

In the second case, $A$ is new and $\gamma$ is not $\sigma$-max derived. As a
consequence, also $\eta$ is not $\sigma$-max derived. By $\mu$-consistency of
$T$ we have $\mu(A)\geq_{\mathit lex}\mu(L_1)\oplus\cdots\oplus\mu(L_q)\oplus
\mu(L_{q+1})\oplus\cdots\oplus\mu(L_r)$. And, by
Lemma~\ref{lemma:properties_of_oplus}~(ii.5), $\mu(A)\geq_{\mathit
lex}\mu(L_{q+1})\oplus\cdots\oplus\mu(L_r)$. Since $\pi_1(\mu(A))>\pi_1(\mu(\neg
K))$ (see the first case), 
by Lemma~\ref{lemma:properties_of_oplus}~(ii.4), we have that:
$\mu(A)\geq_{\mathit lex}\mu(\neg
K)\oplus\mu(L_{q+1})\oplus\cdots\oplus\mu(L_r)$. Thus, $U$ is
$\mu$-consistent. This completes the proof.
\end{proof}

\medskip
The correctness of admissible transformation sequences, that is,
Theorem~\ref{th:corr_of_rules} of
Section~\ref{sec:corr_of_rules}, follows immediately from
Theorem~\ref{th:proof_trees_perfect_model} and
Proposition~\ref{prop:preserv-mu-consistency} because: (i)~$P_d=
P_0\cup \mathit{Defs_n}$, and (ii)~a $\mu$-consistent proof tree is a
proof tree.

\section*{Acknowledgements}
We thank Hirohisa Seki for stimulating conversations on
the topics of this paper. We also thank John Gallagher for
his comments and 
the anonymous referees of ICLP 2010 for their constructive criticism.

We also acknowledge the financial support of: (i)~PRIN 2008 
(Progetto di Ricerca di Interesse Nazionale)
Project no.~9M932N-003, and (ii)~the GNCS Group of \mbox{INdAM}
(Istituto Nazionale di Alta Matematica) under the grant
`Contributo Progetto 2009'.

\end{document}